\documentclass{acmsmall}
\makeatletter
\def\@biblabel#1{[#1]} % restore basic form of \@biblabel macro
\def\thebibliography#1{%
    \footnotesize
    \refsection*{{\refname}
        \@mkboth{\uppercase{\refname}}{\uppercase{\refname}}%
    }
    \list{\@biblabel{\@arabic\c@enumiv}}% %the default form of first arg is {}
       {\settowidth\labelwidth{\@biblabel{#1}}%
        \leftmargin\labelwidth
        \advance\leftmargin\bibindent
        \itemindent-\bibindent
        \itemsep2pt
        \parsep \z@
        \usecounter{enumiv}% % default is to use enumi
        \let\p@enumiv\@empty
        \renewcommand\theenumiv{\@arabic\c@enumiv}%
    }%
    \let\newblock\@empty
    \sloppy
    \sfcode`\.=1000\relax
}
\makeatother

\usepackage{times}
\usepackage{float}
\usepackage{graphicx}
\usepackage{graphics}
\usepackage[dvips]{epsfig}
\usepackage{subfig}
\usepackage{verbatim}
\usepackage{graphicx,float,latexsym,amssymb,amsfonts,amsmath,amstext,times,epsfig,setspace}
\usepackage{algorithm}
\usepackage{algorithmic}
\usepackage{fancyhdr}
\usepackage{stfloats}
\usepackage{tabu}
\usepackage{booktabs}
\usepackage{comment}
\usepackage{bbm}

\newcommand{\calr}{{\cal R}}

\newcommand{\Prob}[1]{\mbox{ Pr}\left[{#1}\right]}

\newcommand{\objsize}{\mbox{\small \em Osize}}

\newcommand{\symsize}{\mbox{\small \em Ssize}}
\newcommand{\blksize}{\mbox{\small \em Bsize}}

\newcommand{\liqsystem}{liquid system}
\newcommand{\Liqsystem}{Liquid system}
\newcommand{\liqcode}{large code}

\newcommand{\liqorg}{flow storage organization}

\newcommand{\liqrepair}{lazy repair}
\newcommand{\Liqrepair}{Lazy repair}

\newcommand{\transientf}{transient node failure}
\newcommand{\Transientf}{Transient node failure}
\newcommand{\nodef}{node failure}
\newcommand{\Nodef}{Node failure}
\newcommand{\latentf}{sector failure}
\newcommand{\Latentf}{Sector failure}
\newcommand{\objloss}{object loss}
\newcommand{\Objloss}{Object loss}
\newcommand{\srcdata}{source data}
\newcommand{\Srcdata}{Source data}

\newcommand{\rcapacity}{raw capacity}

\newcommand{\erasurerate}{erasure rate}

\newcommand{\storeoverhead}{storage overhead}

\newcommand{\rrepairrate}{read repair rate}
\newcommand{\Rrepairrate}{Read repair rate}

\newcommand{\repairrate}{repair rate}

\newcommand{\placegroup}{placement group}
\newcommand{\Placegroup}{Placement group}

\newcommand{\tTot}{T}
\newcommand{\rmtar}{{\scriptscriptstyle\text {tar}}}
\newcommand{\rmnom}{{\scriptscriptstyle\text {nom}}}
\newcommand{\rme}{{\scriptscriptstyle\text {e}}}
\newcommand{\expectation}{\ensuremath{\mathbb{E}}}
\newcommand{\FDis}{{\cal F}_{\scriptscriptstyle\text{Dis}}}
\newcommand{\FAll}{{\cal F}_{\scriptscriptstyle\text{All}}}
\newcommand{\indicator}{\mathbbm{1}}
\newcommand{\nObjects}{{\cal O}}

\newcommand{\Rpeak}{\calr_{\scriptscriptstyle\text {peak}}}
\newcommand{\Ravg}{\calr_{\scriptscriptstyle\text {avg}}}
\newcommand{\Rnnn}{R_{\scriptscriptstyle\text{99\%}}}
\newcommand{\Rnnnnn}{R_{\scriptscriptstyle\text{99.99\%}}}

\newcommand{\tradsystem}{small code system}
\newcommand{\Tradsystem}{Small code system}
\newcommand{\tradorg}{block storage organization}

\newcommand{\tradcode}{small code}
\newcommand{\Tradcode}{Small code}
\newcommand{\tradrepair}{reactive repair}
\newcommand{\Tradrepair}{Reactive repair}

\newcommand{\caldsrc}{{\cal D_{\scriptscriptstyle\text{src}}}}
\newcommand{\caldall}{{\cal D_{\scriptscriptstyle\text{raw}}}}

\newcommand{\expp}[1]{10^{#1}}
\newcommand{\expm}[1]{10^{- #1}}
\newcommand{\mttl}{\mbox{MTTDL}}

\newcommand{\lambdap}{\lambda}
\newcommand{\lambdal}{\lambda_{\scriptscriptstyle\text {SF}}}
\newcommand{\lambdat}{\lambda_{\scriptscriptstyle\text {TF}}}
\newcommand{\Trit}{T_{\scriptscriptstyle\text {RIT}}}

\title{Liquid Cloud Storage}
\author{Michael~G.~Luby
\affil{Qualcomm Technologies, Inc. (luby@qti.qualcomm.com)}
Roberto~Padovani
\affil{Qualcomm Technologies, Inc. (padovani@qti.qualcomm.com)}
Thomas~J.~Richardson
\affil{Qualcomm Technologies, Inc. (tomr@qti.qualcomm.com)}
Lorenz~Minder
\affil{Qualcomm Technologies, Inc. (lminder@qti.qualcomm.com)}
Pooja~Aggarwal
\affil{Qualcomm Technologies, Inc. (poojaa@qti.qualcomm.com)}
}

\begin{abstract}
A {\em \liqsystem} provides durable object storage based on spreading
redundantly generated data across a network of hundreds to thousands of
potentially unreliable storage nodes. A \liqsystem\ uses a combination
of a {\em \liqcode}, {\em \liqrepair}, and a {\em \liqorg}. We show that
a \liqsystem\ can be operated to enable flexible and essentially optimal
combinations of storage durability, \storeoverhead, repair bandwidth
usage, and access performance.
\end{abstract}


\terms{}
\acmformat{}
\begin{document}
%\begin{bottomstuff}
%\end{bottomstuff}
\maketitle

\renewcommand{\headrulewidth}{0pt}% disable the underline of the header part

%\IEEEpeerreviewmaketitle

\section{Introduction}

Distributed storage systems generically consists of a large number of interconnected storage nodes, 
with each node capable of storing a large quantity of data. 
The key goals of a distributed storage systems are to store as much \srcdata\ as possible, 
to assure a high level of durability of the \srcdata,  and to minimize the access time to \srcdata\ by users or applications. 

Many distributed storage systems are built using commodity hardware and managed by complex software, 
both of which are subject to failure.
For example, storage nodes can become unresponsive, sometimes due to software issues where the node can recover after a period of time (\transientf s), and sometimes due to hardware failure in which case the node never recovers (\nodef s).
The paper \cite{Ford10} provides a more complete description of different types of failures.
 
A {\em sine qua non} of a distributed storage system is to ensure that \srcdata\ is durably stored,
where durability is often characterized in terms of the {\em mean time to the loss of any \srcdata} ({\em \mttl})
often quoted in millions of years.
To achieve this in the face of component failures a fraction of the raw storage capacity is used to store
redundant data.
% to achieve a target \mttl.
\Srcdata\ is partitioned into objects, and, for each object, redundant data is generated and stored on the nodes.
An often-used trivial redundancy is {\em triplication} in which three copies of each object are stored on three different nodes. 
Triplication has several advantages for access and ease of repair but its overhead is quite high:
 two-thirds of the raw storage capacity is used to store redundant data. 
Recently, Reed-Solomon (RS) codes have been introduced in production systems to 
both reduce \storeoverhead\ and improve \mttl\ compared to triplication,
see e.g., \cite{Huang12}, \cite{Ford10}, and~\cite{Dimakis13}.  As explained in~\cite{Dimakis13}
the reduction in \storeoverhead\ and improved \mttl\ achieved by RS codes 
comes at the cost of a {\em repair bottleneck}~\cite{Dimakis13} which arises from
the significantly increased network traffic needed to repair data lost due to failed nodes.

The repair bottleneck has inspired the design of new types of erasure codes,
e.g., {\em local reconstruction codes} (LR codes)~\cite{Gopalan12}, \cite{Dimakis13}, 
and {\em regenerating codes} (RG codes)~\cite{Dimakis07}, \cite{Dimakis10},
that can reduce repair traffic compared to RS codes.
The RS codes, LR codes and RG codes used in practice typically use small values of $(n,k,r)$, 
which we call {\em \tradcode s}.
Systems based on \tradcode s, which we call {\em \tradsystem s}, therefore
spread data for each object across a relatively small number of nodes. This situation necessitates using a {\em reactive repair} strategy where, in order to achieve a large \mttl, data lost due to \nodef s must be quickly recovered.
The quick recovery can demand a large amount network bandwidth and the repair bottleneck therefore remains.

These issues and proposed solutions have motivated the search for an understanding of tradeoffs.
Tradeoffs between \storeoverhead\ and repair traffic to protect individual objects of \tradsystem s 
are provided in~\cite{Dimakis07}, \cite{Dimakis10}.   
Tradeoffs between \storeoverhead\ and repair traffic that 
apply to all systems are provided in~\cite{CapacityBounds16}. 
 
Initially motivated to solve the repair bottleneck, we have designed a new class of distributed storage systems 
that we call {\em \liqsystem s}:
\Liqsystem s use {\em \liqcode s} to spread the data stored for each object across a large number of nodes,
and use a {\em \liqrepair\ strategy} to slowly repair data lost from failed nodes.\footnote{The
term {\em \liqsystem} is inspired by how the repair policy operates: instead of reacting and separately
repairing chunks of data lost from each failed node, the lost chunks of data from many failed nodes
are lumped together to form a ``liquid'' and repaired as a flow.}
In this paper we present the design and its most important properties.
Some highlights of the results and content of paper include:
\begin{itemize}
\item
 \Liqsystem s solve the repair bottleneck, with less \storeoverhead\ and larger \mttl\ than \tradsystem s.
 \item
 \Liqsystem s can be flexibly deployed at different \storeoverhead\ and repair traffic operating points.
 \item
\Liqsystem s avoid \transientf\ repair, which is unavoidable for \tradsystem s.
\item
\Liqsystem s obviate the need for \latentf\ data scrubbing, which is necessary for \tradsystem s.
\item
\Liqsystem\ regulated repair provides an extremely large \mttl\ even when \nodef s are essentially adversarial.  
\item
Simulation results demonstrate the benefits of \liqsystem s compared to \tradsystem s.
\item
\Liqsystem\ prototypes demonstrate improved access speeds compared to \tradsystem s.
 \item
 Detailed analyses of practical differences between \liqsystem s and \tradsystem s are provided.
 \item
\Liqsystem s performance approach the theoretical limits proved in~\cite{CapacityBounds16}.
\end{itemize}

\section{Overview}
\label{overview sec}

\subsection{Objects}

We distinguish two quite different types of data objects.
{\em User objects} are organized and managed according to semantics that are meaningful to applications
that access and store data in the storage system.  {\em Storage objects}
refers to the organization of the \srcdata\ that is stored, managed and accessed by the storage system. 
This separation into two types of objects is common in the industry, e.g., see~\cite{Azure11}. 
When user objects (\srcdata) are to be added to the storage system, the user objects are mapped to 
(often embedded in) storage objects and a record of the mapping between user objects and storage objects is maintained.
For example, multiple user objects may be mapped to consecutive portions of a single storage object, 
in which case each user object constitutes a byte range of the storage object.
The management and maintenance of the mapping between user objects and storage objects is outside the scope of this paper.

The primary focus of this work is on storing, managing and accessing storage objects,
which hereafter we refer to simply as {\em objects}. 
Objects are immutable, i.e., once the data in an object is determined then the data within the object does not change.
Appends to objects and deletions of entire objects can be supported.
The preferred size of objects can be determined based on system tradeoffs between the number 
of objects under management and the granularity at which object repair and other system operations are performed.  
Generally, objects can have about the same size, or at least a minimum size, e.g., $1$ GB.
The data size units we use are KB = $1024$ bytes, MB = $1024$ KB, 
GB = $1024$ MB, TB = $1024$ GB, and PB = $1024$ TB. 

\subsection{Storage and repair}

The number of nodes in the system at any point in time is denoted by $M$, 
and the storage capacity of each node is denoted by $S$,
and thus the overall \rcapacity\ is $\caldall = M \cdot S$.
The \storeoverhead\ expressed as a fraction is $\beta = 1 - \frac{\caldsrc}{\caldall}$, 
and expressed as a per cent is $\beta \cdot 100 \%$,
where $\caldsrc$ is the aggregate size of objects (\srcdata) stored in the system.

When using an $(n, k, r)$ erasure code, each object is segmented into $k$ source fragments and 
an encoder generates $r = n-k$ repair fragments from the $k$ source fragments.
An {\em Encoding Fragment ID}, or EFI,  is used to uniquely identify each fragment of an object,
where EFIs $0, 1, \ldots, k-1$ identify the source fragments of an object, 
and EFIs $k, \ldots, n-1$ identify the repair fragments generated from source fragments.  
Each of these $n = k+r$ fragments is stored at a different node.  
An erasure code is MDS (maximum distance separable) if the object can be recovered from any $k$ of the $n$ fragments.

A {\em \placegroup} maps a set of $n$ fragments with EFIs $\{0,1,\ldots,n-1\}$ to a set of $n$ of the $M$ nodes.
Each object is assigned to one of the \placegroup s, and the selected \placegroup\ determines 
where the fragments for that object are to be stored.
To ensure that an equal amount of data is stored on each node, each 
\placegroup\ should be assigned an equal amount of object data 
and each node should accommodate the same number of \placegroup s.

For the systems we describe, $\frac{\caldsrc}{\caldall} = \frac{k}{n}$ and thus $\beta = \frac{r}{n}$, 
i.e., the amount of \srcdata\ stored in the system is maximized subject 
to allotting enough space to store $n$ fragments for each object.
\Objloss\ occurs when the number of fragments available for an object is less than $k$, 
and thus the object is no longer recoverable.
A repairer is an agent that maintains recoverability of objects
by reading, regenerating and writing fragments lost due to \nodef s, 
to ensure that each object has at least $k$ available fragments.
In principle the repairer could regenerate $r$ missing fragments for an object when 
only $k$ fragments remain available.   
However, this repair policy results in a small \mttl, 
since one additional \nodef\ before regenerating the $r$ missing fragments for the object causes \objloss.

A repairer determines the rate at which repair occurs, measured in terms of a {\em \rrepairrate},
since data read by a repairer is a fair measure of repair traffic that moves through the network.
(The amount of data written by a repairer is essentially the same for all systems, 
and is generally a small fraction of data read by a repairer.)

Of primary interest is the amount of network bandwidth $\Rpeak$ that needs to be dedicated 
to repair to achieve a given \mttl\ for a particular system, as this value of $\Rpeak$ 
determines the network infrastructure needed to support repair for the system. 
Generally, we set a global upper bound $\Rpeak$ on the allowable \rrepairrate\ used by a repairer, 
and determine the \mttl\ achieved for this setting of $\Rpeak$ by the system.
The average \rrepairrate\ $\Ravg$ is also of some interest, since this determines the average amount
of repair traffic moved across the network over time.

The fixed-rate repair policy works in similar ways for \tradsystem s and \liqsystem s in our simulations:  
The repairer tries to use a \rrepairrate\ up to $\Rpeak$ whenever there are any objects to be repaired.  

\subsection{\Tradsystem s}

Replication, where each fragment is a copy of the original object is an example of a trivial MDS erasure code.
For example, triplication is a simple $(3,1,2)$ MDS erasure code, 
wherein the object can be recovered from any one of the three copies.  
Many distributed storage systems use replication. 

A Reed-Solomon code (RS code)~\cite{CCauchy95}, \cite{CRizzo97}, \cite{RFC5510} is an MDS code
that is used in a variety of applications and is a popular choice for storage systems.
For example, the \tradsystem s described in~\cite{Ford10}, \cite{Huang12} 
use a $(9,6,3)$ RS code, and those in~\cite{Dimakis13} use a $(14,10,4)$ RS code.  

When using a $(n,k,r)$ \tradcode\ for a \tradsystem\ with $M$ nodes, 
an important consideration is the number of \placegroup s $P$ to use.
 Since $n \ll M,$  a large number $P$ of \placegroup s are typically used 
 to smoothly spread the fragments for the objects across the nodes,
e.g., Ceph~\cite{Ceph} recommends $P = \frac{100 \cdot M}{n}$.
A \placegroup\ should avoid mapping fragments to nodes with
correlated failures, e.g., to the same rack and more generally to the same failure domain.
Pairs of \placegroup s should avoid mapping fragments to the same pair of nodes.
\Placegroup s are updated as nodes fail and are added. 
These and other issues make it challenging to design \placegroup s management for \tradsystem s. 

{\em \Tradrepair} is used for \tradsystem s, i.e., the repairer quickly regenerates fragments as soon as they are 
lost from a failed node (reading $k$ fragments for each object to regenerate the 
one missing fragment for that object). 
This is because, once a fragment is lost for an object due to a \nodef, the probability of $r$ 
additional \nodef s over a short interval of time when $r$ is small is significant enough 
that repair needs to be completed as quickly as practical.
Thus, the peak \rrepairrate\ $\Rpeak$ is higher than the average \rrepairrate\ $\Ravg$, 
and $\Ravg$ is $k$ times the \nodef\ \erasurerate. 
In general, \tradrepair\ uses large bursts of repair traffic for short periods of time 
to repair objects as soon as practical after a node storing data for the objects is declared to have permanently failed.

The peak \rrepairrate\ $\Rpeak$ and average \rrepairrate\ $\Ravg$  needed 
to achieve a large \mttl\ for \tradsystem s can be substantial.
Modifications of standard erasure codes have been  
designed for storage systems to reduce these rates,  
e.g., {\em local reconstruction codes} (LR codes)~\cite{Gopalan12}, \cite{Dimakis13}, 
and {\em regenerating codes} (RG codes)~\cite{Dimakis07}, \cite{Dimakis10}.  
Some versions of LR codes have been used in deployments, e.g., by Microsoft Azure~\cite{Huang12}.
The encoding and decoding complexity of RS, LR, and RG codes grows non-linearly (typically quadratric or worse) 
as the values of $(n,k,r)$ grow, which makes the use of such codes with large values of $(n,k,r)$ less practical.  

\subsection{\Liqsystem s}
\label{liqsystem overview sec}

We introduce {\em \liqsystem s} for reliably storing objects in a distributed set of storage nodes.
\Liqsystem s  use a combination of a {\em \liqcode} with large values of $(n,k,r)$, {\em \liqrepair}, and a {\em \liqorg}.
Because of the large code size the \placegroup\ concept is not of much importance: 
For the purposes of this paper we assume that one \placegroup\ is used for all objects, 
i.e., $n=M$ and a fragment is assigned to each node for each object.

The RaptorQ code \cite{CRaptorQ11},~\cite{RFC6330} is an example of 
an erasure code that is suitable for a \liqsystem.
RaptorQ codes are fundamentally designed to support large values of $(n,k,r)$ with very low encoding and decoding complexity and to have exceptional recovery properties.  Furthermore, RaptorQ codes are {\em fountain codes}, which means that as many encoded symbols $n$ as desired can be generated on the fly for a given value of $k$.  The fountain code property provides flexibility when RaptorQ codes are used in applications, including \liqsystem s. The monograph \cite{CRaptorQ11} provides a detailed exposition of the design and analysis of RaptorQ codes. 
  
The value of $r$ for a \liqsystem\ is large, and thus {\em \liqrepair} can be used, 
where lost fragments are repaired at a steady background rate using a reduced amount of bandwidth. 
The basic idea is that the \rrepairrate\ is controlled so that objects are typically repaired
when a large fraction of $r$ fragments are missing (ensuring high repair efficiency), 
but long before $r$ fragments are missing (ensuring a large \mttl), and thus
the \rrepairrate\ is not immediately reactive to individual \nodef s.
For a fixed \nodef\ rate the \rrepairrate\ can in principle be fixed as described in Appendix~\ref{liquid mttl analysis sec}.
When the \nodef\ rate is unknown a priori,
the algorithms described in Section~\ref{self-regulating sec} and  Appendix~\ref{liquid mttl analysis sec}
should be used to adapt the \rrepairrate\ to changes in conditions 
(such as \nodef\ rate changes) to ensure high repair efficiency and a large \mttl.

We show that a \liqsystem\ can be operated to enable flexible and essentially optimal combinations of storage durability,  \storeoverhead, repair bandwidth usage, and access performance, exceeding the performance of \tradsystem s.

Information theoretic limits on tradeoffs between storage overhead and the \rrepairrate\ are introduced and proved
in~\cite{CapacityBounds16}.  The \liqsystem s for which prototypes and simulations are described 
in this paper are similar to the \liqsystem s described in Section~9.A of~\cite{CapacityBounds16} 
that are within a factor of two of optimal.  More advanced \liqsystem s similar to those described in Section~9.B of~\cite{CapacityBounds16} that are asymptotically optimal would perform even better.

Fig.~\ref{fig_lds_sys_arch} shows a possible \liqsystem\ architecture.
In this architecture, a standard interface, such as the S3 or SWIFT API,
is used by clients to make requests to store and access objects.  
Object storage and access requests are distributed across access proxy servers within the liquid access tier.  
Access proxy servers use a RaptorQ codec to encode objects into fragments and to decode
fragments into objects, based on the \liqorg.
Access proxy servers read and write fragments from and to 
the storage servers that store the fragments on disks within the liquid storage tier. 
 \Liqrepair\ operates as a background process within the liquid storage tier. 
 Further details of this example \liqsystem\ architecture are discussed in Section~\ref{operational sec}.

\begin{figure}
\centering
\includegraphics[scale=0.6]{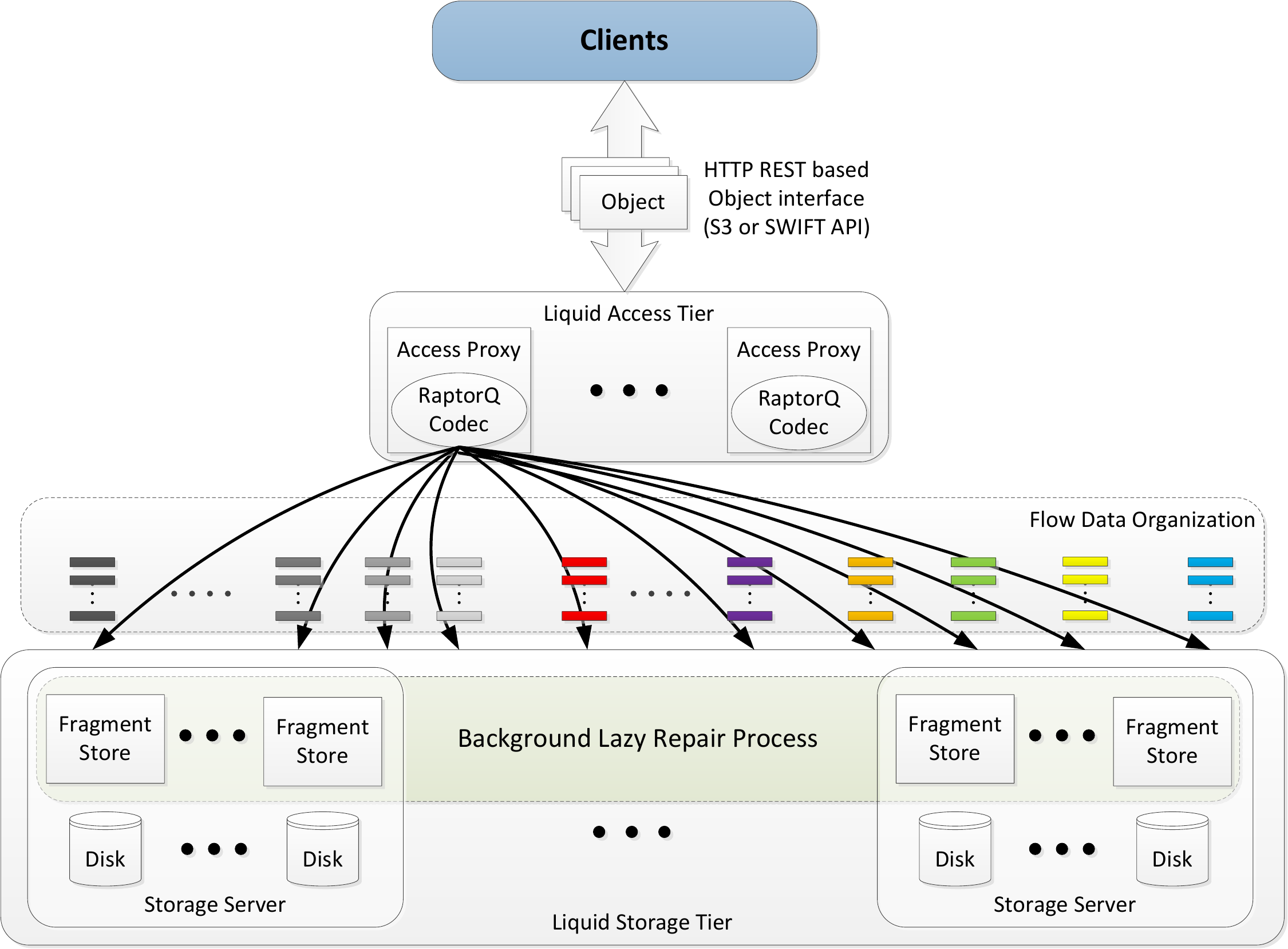}
\caption{A \liqsystem\ architecture.}
\label{fig_lds_sys_arch}
\end{figure}

\section{System failures}

Since distributed storage systems consist of a large number of
components, failures are common and unavoidable.  The use of commodity
hardware can additionally increase the frequency of failures and
outages, as can necessary maintenance and service operations such as
software upgrades or system reconfiguration on the hardware or software
level.

Different failure types can affect the system in various ways.  For
example, a hard drive can fail in its entirety, or individual sectors of
it can become unreadable.  \Nodef s can be permanent, e.g. if the
failure results in the entire node needing to be replaced, or it can be
transient, if it is caused issues such as network outages, reboots,
software changes.  Network failures can cause subsets of many nodes to
be unreachable, e.g. a full rack or even a complete site.

Moreover, distributed storage systems are heterogeneous in nature, with
components that are commonly being upgraded and replaced.  Device models
do not remain the same; for example, a replacement hard drive will often
be of a different kind, with different and possibly unknown expected
lifetime.  Even within a device model, different manufacturing runs can
result in variation in expected lifetimes.  Further, upgrades may
require a large set of machines to be replaced in a short time frame,
leading to bursts of increased churn.  Changes to the underlying
technology, such as a transition from hard drives to SSDs, are also
likely to profoundly impact the statistics and nature of storage
failures.

For these reasons, the failure behavior of a distributed storage system
is difficult to model.  Results obtained under complicated realistic
error models are difficult to interpret, and hard to generalize.  
The results in this paper are obtained under a fairly
straightforward failure model which focuses on two types of failures:
\Nodef{}s and \latentf{}s.

\subsection{\Nodef s}
\label{node failures sec}

Although \nodef\ processes can be difficult to accurately model,
analysis based on a random \nodef\ models can provide insight into the
strengths and weaknesses of a practical system, and can provide a first
order approximation to how a practical system will operate.  
The time till failure of an individual node is often modeled
by a Poisson random variable with rate $\lambdap$, in which case
the average time till an individual node fails is $1/\lambdap.$
Nodes are assumed to fail independently.  Since failed nodes are
immediately replaced, the aggregate \nodef\ process is Poisson with rate
$\lambdap \cdot M$.  Our simulations use this model with $1/\lambdap= 3$
years, and also use variants of this model where the \nodef\ rate bursts
to $3$ times this rate ($1/\lambdap= 1$ year) or $10$ times this rate
($1/\lambdap= 0.3$ years) for periods of time.

Nodes can become unresponsive for a period of time and then become
responsive again (\transientf), in which case the data they store is
temporarily unavailable, or nodes can become and remain unresponsive
(\nodef), in which case the data they store is permanently unavailable.
\Transientf s are an order of magnitude more common than \nodef s in
practice, e.g., see \cite{Ford10}.%, \cite{}. XXX other citation?

Unresponsive node are generally detected by the system within a few
seconds.  However, it is usually unknown whether the unresponsiveness is
due to a \transientf\ or a \nodef.  Furthermore, for \transientf{}s,
their duration is generally unknown, and can vary in an extremely large
range:  Some \transientf{}s have durations of just seconds when others
can have durations of months, e.g., see~\cite{Ford10}.

Section~\ref{self-regulating sec} discusses more general and realistic
\nodef\ models and introduces repair strategies that automatically
adjust to such models. 

\subsection{\Latentf s}

\Latentf s are another type of data loss that occurs in practice and can
cause operational issues.  A \latentf\ occurs when a sector of data
stored at a node degrades and is unreadable.   Such data loss can only
be detected when an attempt is made to read the sector from the node.
(The data is typically stored with strong checksums, so that the
corruption or data loss becomes evident when an attempt to read the
data is made.) Although the \latentf\ data loss rate is fairly small,
the overhead of detecting these types of failures can have a negative
impact on the \rrepairrate\ and on the \mttl.  

Data {\em scrubbing} is often employed in practice to detect \latentf s:
An attempt is made to read through all data stored in the system at a
regular frequency.  As an example, \cite{Cowling16} reports that Dropbox
scrubs data at the frequency of once each two weeks, and reports that
read traffic due to scrubbing can be greater than all other read data
traffic combined.  Most (if not all) providers of data storage systems
scrub data, but do not report the frequency of scrubs.  Obviously,
scrubbing can have an impact on system performance.

The motivation for data scrubbing  is that otherwise \latentf s could
remain undetected for long periods of time, and thus have a large
negative impact on the \mttl.  This impact can be substantial even if
the data loss rate due to \latentf s  is relatively much smaller than
the data loss rate due to \nodef s.  The reason for this is that time
scale for detecting and repairing \latentf s is so much higher than 
for detecting and repairing \nodef s.

We do not employ a separate scrubbing process in our simulations, but
instead rely on the repair process to scrub data:  Since each node fails
on average in $1/\lambdap$ years, data on nodes is effectively scrubbed
on average each $1/\lambdap$ years.  
\Latentf s do not significantly affect the \mttl\ of \liqsystem s,
but can have a significant affect on the \mttl\ for \tradsystem s.

We model the time to \latentf\ of an individual $4$ KB sector of data on
a node as a Poisson random variable with rate $\lambdal$, and thus the
average time till an individual $4$ KB sector of data on a node
experiences a \latentf\ is $1/\lambdal$, where \latentf s of different
sectors are independent.  Based on \latentf\ statistics derived from the
data of \cite{BackBlaze16}, our simulations use this model with
$1/\lambdal= 5 \cdot \expp{8}$ years.  Thus, the data loss rate due to
\latentf s is more than $\expp{8}$ times smaller than the data loss rate
due to \nodef s in our simulations.

\section{Repair}
\label{repair section}

Although using erasure coding is attractive in terms of reducing
\storeoverhead s and improving durability compared to replication, it
can potentially require large amounts of bandwidth to repair for fragment
loss due to node failures.  Since the amount of required repair
bandwidth can be high and can cause system-wide bottlenecks, this is
sometimes referred to as a {\em repair bottleneck}.  For example,
\cite{Dimakis13} states the goal of their work as follows:
\begin{quote}
Our goal is to design more efficient coding schemes that would allow a
large fraction of the data to be coded without facing this repair
bottleneck. This would save petabytes of \storeoverhead s and
significantly reduce cluster costs.
\end{quote}

The repair bottleneck sometimes has proven to be an obstacle in the
transition from simple replication to more storage efficient and durable
erasure coding in storage clusters. It simply refers to the fact that,
when using a standard erasure code such as a RS code, the loss of a
fragment due to a \nodef\  requires the transfer of $k$ fragments for
repair, thus increasing {\em $k$-fold} the required bandwidth and I/O
operations in comparison to replication.  For example, \cite{Dimakis13}
states the following when describing a deployed system with 3000 nodes:
\begin{quote}
\ldots the repair traffic with the current configuration 
is estimated around 10-20\% of the total average of 
2 PB/day cluster network traffic. As discussed, (14,10,4) RS 
encoded blocks require approximately 10x more network
repair overhead per bit compared to replicated blocks. We
estimate that if 50\% of the cluster was RS encoded, the 
repair network traffic would completely saturate the cluster
network links.
\end{quote}

The recognition of a repair bottleneck presented by the use of standard
MDS codes, such as RS codes, has led to the search for codes that
provide {\em locality}, namely where the repair requires the transfer of
$\ell < k$ fragments for repair. The papers \cite{Gopalan12}, \cite{Huang12},
\cite{Dimakis13} introduce LR codes and provide tradeoffs between
locality and other code parameters.  We discuss some properties of
\tradsystem s based on LR codes in Section~\ref{simulations 3010 fixed
sec}.

The following sub-sections provide the motivation and details for the
repair strategies employed in our simulations. 

\subsection{Repair initiation timer}

In practice, since the vast majority of unresponsive nodes are
\transientf s, it is typical to wait for a {\em repair initiation time}
$\Trit$ after a node becomes unresponsive before initiating repair.  For
example, $\Trit = 30$ minutes for the \tradsystem\ described in
\cite{Huang12},  and thus the policy reaction is not immediate but
nearly immediate. 

In our simulations repairers operate as follows.  If a node is
unresponsive for a period at most $\Trit$ then the event is classified
as a \transientf\ and no repair is triggered.  After being unresponsive
for a period of $\Trit$ time, a \nodef\ is declared, the node is
decommissioned, all fragments stored on the node are declared to be
missing (and will be eventually regenerated by the repairer).  The
decommissioned node is immediately replaced by a new node that initially
stores no data (to recover the lost raw capacity).

\Transientf s may sometimes take longer than $30$ minutes to resolve.
Hence, to avoid unnecessary repair of \transientf s,  it is desirable to
set $\Trit $ longer, such as $24$ hours.  A concern with setting $\Trit$
large is that this can increase the risk of \srcdata\ loss,
i.e., it can significantly decrease the $\mttl.$  
 
Since a large value for $\Trit$ has virtually no impact on the \mttl\
for \liqsystem s, we can set $\Trit$ to a large value and
essentially eliminate the impact that \transientf s have on the
\rrepairrate.  Since a large value for $\Trit$ has a large negative
impact on the \mttl\ for \tradsystem s, $\Trit$ must be set
to a small value to have a reasonable \mttl, in which case the
\rrepairrate\ can be adversely affected if the unresponsive time period
for \transientf\ extends beyond $\Trit$ and triggers unnecessary repair.  
We provide some simulation results including \transientf s in Section~\ref{simulations transient sec}. 

\subsection{Repair strategies}

In our simulations, a repairer maintains a repair queue of objects to repair. An object is added to the repair queue 
when at least one of the fragments of the object have been determined to be missing. 
Whenever there is at least one object in the repair queue the repair policy works to repair objects.  
When an object is repaired, at least $k$ fragments are read to recover the object and 
then additional fragments are generated from the object and stored at nodes 
which currently do not have fragments for the object.

Objects are assigned to \placegroup s to determine which nodes store the fragments for the objects.
Objects within a \placegroup\ are {\em prioritized} by the repair policy: 
objects with the {\em least} number of available fragments within a \placegroup\ 
are the next to be repaired within that \placegroup.
Thus, objects assigned to the same \placegroup\ have a {\em nested object structure}: 
The set of available fragments for an object within a \placegroup\ is a subset of the set of available 
fragments for any other object within the \placegroup\ repaired more recently.
Thus, objects within the same \placegroup\ are repaired in a round-robin order.
Consecutively repaired objects can be from different \placegroup s if there are multiple \placegroup s.

\subsection{Repair for \tradsystem}
\label{SC repair subsec}

In our simulations, \tradsystem s (using either replication or \tradcode s) use {\em \tradrepair}, i.e., 
$\Rpeak$ is set to a significantly higher value than the required average repair bandwidth,
and thus repair occurs in a burst when a \nodef\ occurs and 
the lost fragments from the node are repaired as quickly as possible.  
Typically only a small fraction of objects are in the repair queue at any point in time for a \tradsystem\ using \tradrepair.

We set the number of \placegroup s to $P = \frac{100 \cdot M}{n}$ 
to the recommended Ceph~\cite{Ceph} value, and thus
there are $100$ \placegroup s assigned to each of the $M$ nodes, with each node storing fragments for \placegroup s assigned to it. 
The repair policy works as follows: 

\begin{itemize}
\item
If there are at most $100$ \placegroup s that have objects to repair then each such \placegroup\ 
repairs at a \rrepairrate\ of $\Rpeak/100$.
\item
If there are more than $100$ \placegroup s that have objects to repair then the $100$ \placegroup s
with objects with the least number of available fragments are repaired at a \rrepairrate\ of $\Rpeak/100$.
\end{itemize}
This policy ensures that, in our simulations, the global peak \rrepairrate\ is at most $\Rpeak$, 
that the global bandwidth $\Rpeak$ is fully utilized in the typical case when 
one node fails and needs repair (which triggers $100$ \placegroup s
to have objects to repair), and that the maximum traffic 
from and to each node is not exorbitantly high.  (Unlike the case for \liqsystem s,
there are still significant concentrations of repair traffic in the network with this policy.)

The need for \tradsystem s to employ \tradrepair, i.e., use a high $\Rpeak$ value, is a consequence 
of the need to achieve a large \mttl. \Tradcode s 
are quite vulnerable, as the loss of a small number of fragments can lead to \objloss. 
For instance, with the $(14,10,4)$ RS code the loss of five or more fragments leads to \objloss. 
Fig.~\ref{fig_obj_loss_prob} illustrates the drastically different time scales at which repair needs to operate 
for a \tradsystem\ compared to a \liqsystem.  In this example, nodes fail on average every three years and 
Fig.~\ref{fig_obj_loss_prob} shows the probability that $N \le k-1$ of $n$ nodes are functioning at time $t$ 
when starting with $N=n$ nodes functioning at time zero.   
As shown, a failed node needs to be repaired within two days to avoid a $\expm{10}$  
probability \objloss\ event for a $(14,10,4)$ RS code, whereas a failed node needs to be repaired 
within $10$ months for a $(3010,2150,860)$ \liqcode.  
Note that the \storeoverhead\ $\beta = r/n = 0.286$  is the same in these two systems.
A similar observation was made in \cite{Quantitative02}. 

Exacerbating this for a \tradsystem\ is that there are many \placegroup s and hence many combinations of small 
number of \nodef s that can lead to \objloss, thus requiring a high $\Rpeak$ value to ensure a target \mttl.
For example, $P = \frac{100 \cdot M}{n}= 21500$  for a $(14,10,4)$ \tradcode\ in a $M = 3010$ node system, 
and thus an object from the system is lost if any object from any of the $21500$ 
\placegroup s is not repaired before more than four nodes storing fragments for the object fail.
Furthermore, a \tradsystem\ using many \placegroup s is difficult to analyze, 
making it difficult to determine the value of $\Rpeak$ that ensures a target \mttl.  

\begin{figure}
\centering
\includegraphics[width=0.75 \textwidth]{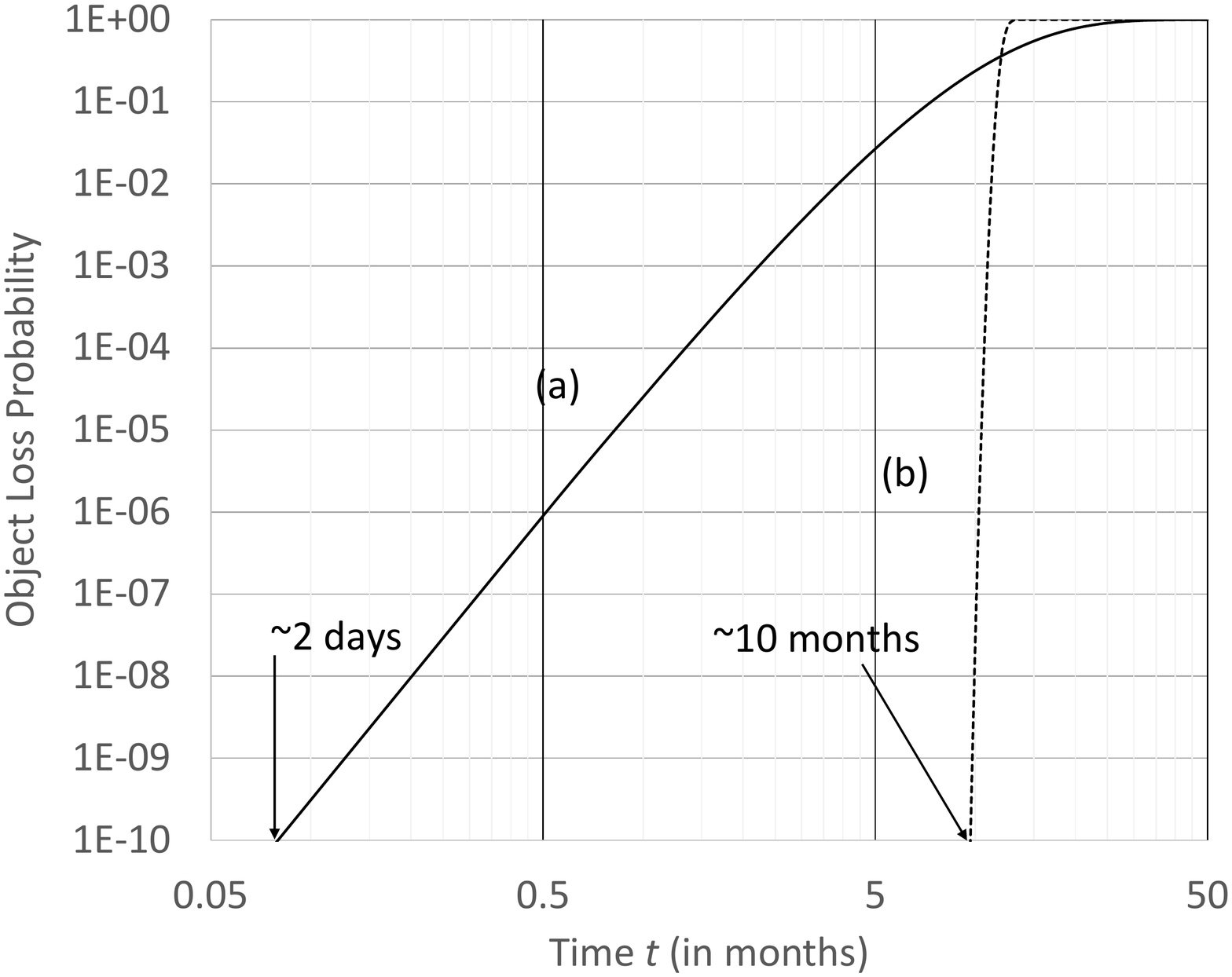}
\caption{Cumulative probability of \objloss\ versus time for: (a) $(14,10,4)$ \tradcode\ ; (b) $(3010,2150,860)$ \liqcode.}
\label{fig_obj_loss_prob}
\end{figure} 

%\begin{equation}
%\Prob{N(t) \le k-1| N(0)=n}=\sum_{i=0}^{k-1} \binom{n}{i} p^i (1-p)^{n-i}
%\end{equation}
%where $p=e^{-\lambda t}$.

Although \tradcode s outperform replication in both \storeoverhead\ and reliability metrics, they require 
a high $\Rpeak$ value to achieve a reasonable $\mttl$, which leads to the {\em repair bottleneck} described earlier. 
For example, $\Rpeak$ used by \tradsystem\ can be
controlled to be a fraction of the total network bandwidth available in the cluster, e.g., $10\%$ of the total available bandwidth is mentioned in \cite{Huang12} and $10\% - 20\%$ was estimated in \cite{Dimakis13} to protect a small fraction of the object data with a \tradcode.   
However, using these fractions of available bandwidth for repair may or may not achieve a target \mttl.
As we show in Section~\ref{repair sim sec} via simulation, 
\tradsystem s require a large amount of available bandwidth 
for bursts of time in order to achieve a reasonable target \mttl, and in many cases cannot achieve a reasonable
\mttl\ due to bandwidth limitations, or due to operating with a large $\Trit$ value.

\subsection{Repair for fixed rate \liqsystem}
\label{repair liquid sec}

The simplest repair policy for a \liqsystem\ is to use a fixed \rrepairrate: 
the value of the peak \rrepairrate\ $\Rpeak$ is set and the \rrepairrate\ is $\Rpeak$ 
whenever there are objects in the repair queue to be processed. 
A \liqsystem\ employs a repairer that can be described as {\em lazy}: 
The value of $\Rpeak$ is set low enough that the repair queue contains almost all objects all the time.
Overall, a \liqsystem\ using \liqrepair\ uses a substantially lower average repair bandwidth 
and a dramatically lower peak repair bandwidth than a \tradsystem\ using \tradrepair. 

There are two primary reasons that a \liqsystem\ can use \liqrepair: (1) the usage of a \liqcode, 
i.e., as shown in Fig.~\ref{fig_obj_loss_prob}, the $(3010,2150,860)$ \liqcode\ has substantially more time to repair than the $(14,10,4)$ \tradcode\ to achieve the same \mttl; (2) the nested object structure.

A \liqsystem\ can store a fragment of each object at each node, and thus uses a single \placegroup\ for all objects. 
For example, using a $(3010,2150,860)$ \liqcode, an object is lost 
if the system loses more than $860$ nodes storing fragments for the object before the object is repaired. 
Since there is only one \placegroup\, the nested object structure implies that if the object 
with the fewest number of available of fragments can be recovered then all objects can be recovered. 
This makes it simpler to determine an $\Rpeak$ value that achieves a target \mttl\ for a fixed \nodef\ rate.
The basic intuition is that the repair policy should cycle through and repair the objects fast 
enough to ensure that the number of nodes that fail between repairs of the same object is at most $r$.  
Note that $\caldsrc/\Rpeak$ is the time between repairs of the same object when 
the aggregate size of all objects is $\caldsrc$ and the peak global read repair bandwidth is set to $\Rpeak$.   
Thus, the value of $\Rpeak$ should be set so that the probability that there are more
than $r$ node failures in time $\caldsrc/\Rpeak$ is extremely tiny.  
Eq.~(\ref{eqn:MTTDLestimate}) from Appendix~\ref{liquid mttl analysis sec} provides methodology 
for determining a value of $\Rpeak$ that achieves a target \mttl.

Unlike \tradsystem s, wherein $\Trit$ must be relatively small, a \liqsystem s can use large $\Trit$ values,
which has the benefit of practically eliminating unnecessary repair due to transient failures.

\section{Repair simulation results}
\label{repair sim sec}

The storage capacity of each node is $S = 1$ PB in all simulations.
The system is fully loaded with \srcdata, i.e.,
$\caldsrc = S \cdot M \cdot (1-\beta)$, where $\beta$ is the \storeoverhead\ expressed as a fraction.

There are \nodef s in all simulations and node lifetimes are modeled as independent 
exponentially distributed random variables with parameter $\lambdap.$  
Unless otherwise specified, we set $1/\lambdap$ to $3$ years.
When \latentf s are also included, sector lifetimes are also exponentially distributed with
$1/\lambdal$ fixed to $5 \cdot \expp{8}$ years.
 
At the beginning of each simulation the system is in a perfect state of repair, i.e.,
all $n$ fragments are available for all objects.
For each simulation, the peak global read repair bandwidth is limited to a specified $\Rpeak$ value.
The \mttl\ is calculated as the number of years simulated divided by one more than the 
observed number of times there was an \objloss\ event during the simulation. 
The simulation is reinitialized to a perfect repair state and the simulation 
is continued when there is an \objloss\ event.
To achieve accurate estimates of the $\mttl$, each simulation runs for a maximum number of years,
or until $200$ \objloss\ events have occurred.

Fig.~\ref{fig_reactive_vs_lazy_repair} presents a comparison of repair traffic 
for a \liqsystem\ and a \tradsystem.
In this simulation, \liqrepair\ for the \liqsystem\ uses a steady \rrepairrate\ of $704$ Gbps 
to achieve an \mttl\ of $\expp{8}$ years, whereas the \tradrepair\ 
for the \tradsystem\ uses \rrepairrate\ bursts of $6.4$ Tbps to achieve an \mttl\ slightly smaller than $\expp{7}$ years.  
For the \tradsystem, a \rrepairrate\ burst starts each time a node fails and lasts for approximately $4$ hours. 
One of the $3010$ nodes fails on average approximately every $8$ hours. 

\begin{figure}
\centering
\includegraphics[width=0.75 \textwidth]{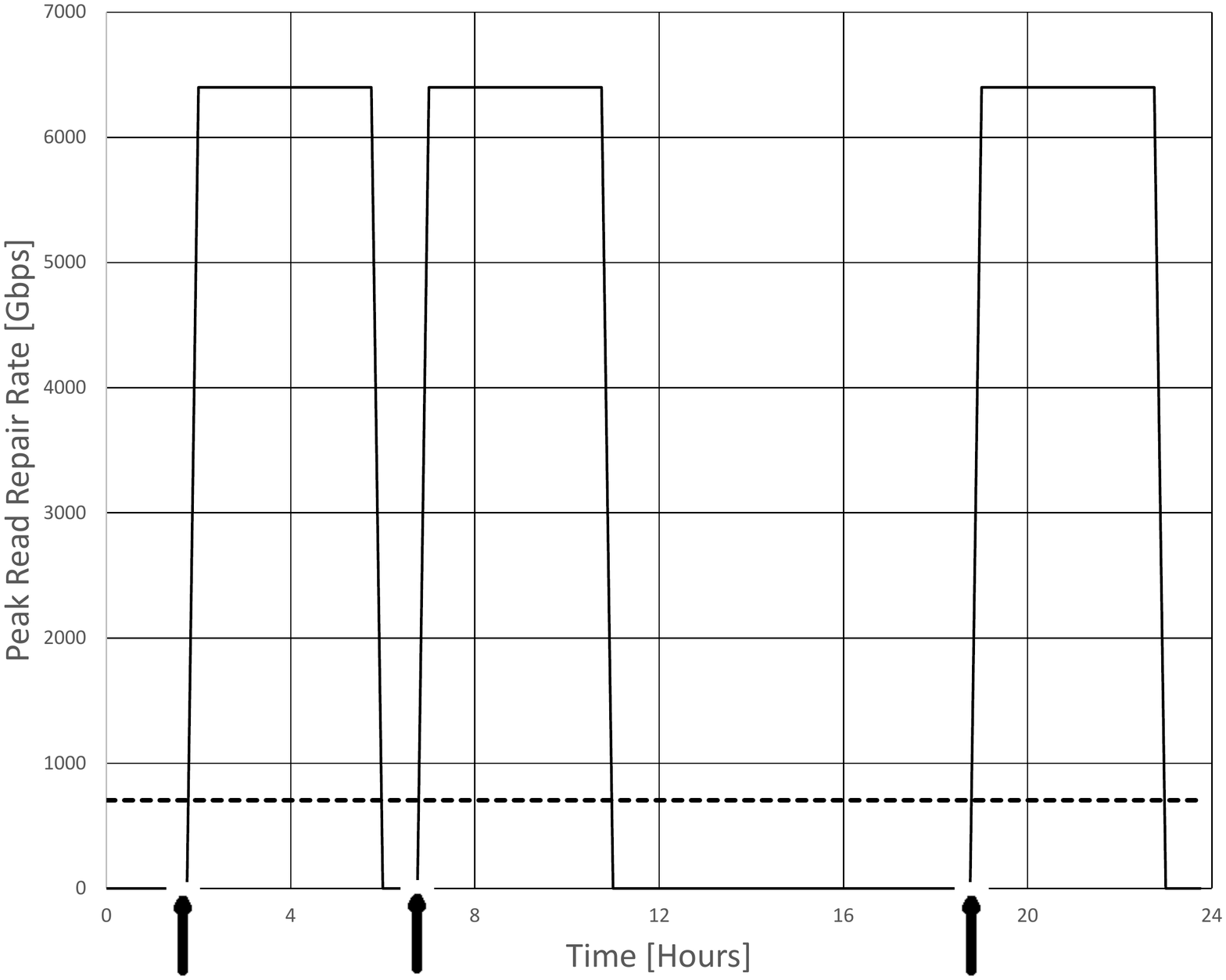}
\caption{\Rrepairrate\ comparison over a 24 hour period of \tradrepair\ using a $(14,10,4)$ \tradcode\ (solid line) and \liqrepair\ using a $(3010,2150,860)$ \liqcode\ (dash line).  Arrows indicate times of \nodef s.}
\label{fig_reactive_vs_lazy_repair} 
\end{figure}  

In Fig.~\ref{fig_mttdl_vs_rpeakg_ms}, we plot  $\Rpeak$ against achieved \mttl\  for a cluster of $M=402$ nodes.
The \liqsystem s use fixed \rrepairrate\ equal to the indicated $\Rpeak$ and two cases of $\beta = 33.3\%$ and $\beta = 16.7\%$ are shown.
The curves depict \mttl\ bounds as calculated according to Eq.~\eqref{eqn:MTTDLestimate} of Appendix~\ref{liquid mttl analysis sec}.
As can be seen, the bounds and the simulation results agree quite closely.
Fig.~\ref{fig_mttdl_vs_rpeakg_ms} also shows
simulation results for a \tradsystem\ with \storeoverhead\ $33.3\%,$ which illustrates
a striking difference in performance between \tradsystem s and \liqsystem s. 

\begin{figure}
\centering
\includegraphics[width=0.75 \textwidth]{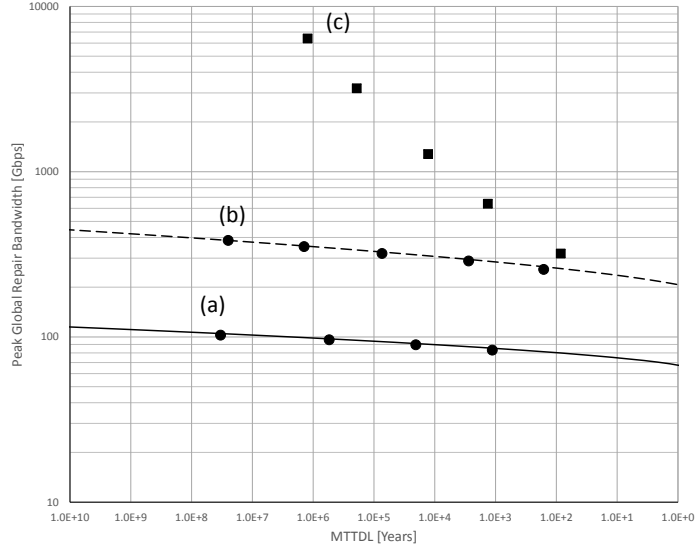}
\caption{Plot of $\mttl$ bound vs. $\Rpeak$ for \liqsystem s with: (a) $(402,268,134)$ code 
($33.3\%$ \storeoverhead); (b) $(402,335,67)$ code ($16.7\%$ \storeoverhead). 
Round markers on curve are simulation results. Square markers (c) are simulation results 
for \tradsystem\ with $(9,6,3)$ code ($33.3\%$ \storeoverhead).}
\label{fig_mttdl_vs_rpeakg_ms}
\end{figure}

\subsection{Fixed \nodef\ rate with $402$ nodes}
\label{simulations 402 fixed sec}

Fig.~\ref{fig_repair_plot_402_fixed} 
plots the simulation results for $\Rpeak$ against $\mttl$ for a wide variety of $M = 402$ node systems.
Each simulation was run for $\expp{9}$ years or until $200$ \objloss\ events occurred.
The number of nodes, the \nodef\ rate, and the \tradcode s used 
in these simulations are similar to~\cite{Huang12}. 
The square icons correspond to $16.7\%$ \storeoverhead\ and the circle icons correspond to $33.3\%$ \storeoverhead. 
For the \tradsystem s,  a $(18,15,3)$ \tradcode\ is used for $16.7\%$ \storeoverhead,
and a $(9,6,3)$ \tradcode\ is used for $33.3\%$ \storeoverhead.  
For the \liqsystem s, a $(402,335,67)$ \liqcode\ is used for $16.7\%$ \storeoverhead,
and a $(402,268,134)$ \liqcode\ is used for $33.3\%$ \storeoverhead.

The unshaded icons correspond to systems where there are only \nodef s,
whereas the shaded icons correspond to systems where there are both \nodef s and \latentf s. 

The \tradsystem s labeled as ``SC, 30min'' and ``SC, 24hr'' use the \rrepairrate\ strategy described 
in Section~\ref{SC repair subsec} based on the indicated $\Rpeak$ value,
with $\Trit$ set to 30 minutes and 24 hours respectively.
The triplication systems labeled as ``Trip, 30min'' and ``Trip, 24hr'' use the \rrepairrate\ strategy described 
in Section~\ref{SC repair subsec} based on the indicated $\Rpeak$ value
with $\Trit$ is set to 30 minutes and 24 hours respectively.

The \liqsystem s labeled as ``LiqF, 24hr'' use a fixed \rrepairrate\ set to the shown $\Rpeak$, which was
calculated according to Eq.~\eqref{eqn:MTTDLestimate} from Appendix~\ref{liquid mttl analysis sec}
with a target $\mttl$ of $\expp{7}$ years.
The value of $\Trit$ is set to 24 hours
As would be expected since Eq.~\eqref{eqn:MTTDLestimate} is a lower bound, 
the observed $\mttl$ in the simulations is somewhat above $\expp{7}$ years. 

The \liqsystem s labeled as ``LiqR, 24hr'' have $\Trit$ set to 24 hours and use a version of the regulator algorithm described 
in Section~\ref{self-regulating sec} to continually
and dynamically recalculate the \rrepairrate\ according to current conditions.
The regulator used (per object) node failure arrival rate estimates based on the average 
of the last $\frac{7}{6}\cdot r - F$ \nodef\ inter-arrival times 
where $F$ denotes the number of missing fragments for an object at the time it forms the estimate.
The maximum \rrepairrate\ was limited to the shown $\Rpeak$ value,
which is about three times the average rate.
In these runs there were no \objloss\ events in the $\expp{9}$ years 
of simulation when there are both \nodef s and \latentf s.  
This is not unexpected, since the estimate 
(using the \mttl\ lower bound techniques described in Section~\ref{self-regulating sec} 
but with a slightly different form of \nodef\ rate estimation) for the
actual $\mttl$ is above $\expp{25}$ years for the $33.3\%$ \storeoverhead\ case 
and above $2.5\cdot\expp{12}$ years for the $16.7\%$ \storeoverhead\ case. 
These estimates are also known to be conservative when compared to simulation data.
The same estimate shows that with this regulator an \mttl\ of $\expp{9}$ years 
would be attained with the smaller overhead of $r=115$ instead of $r=134$.  
As shown in Table~\ref{tab_regualator_examples}, the average \rrepairrate\ $\Ravg$ is around $1/3$ of $\Rpeak$, and 
the $99\%$ \rrepairrate\ $\Rnnn$ is around $1/2$ of $\Rpeak$.
The large \mttl\ achieved by the regulator arises from its ability to raise the repair rate when needed.  
Other versions of the regulator in which the peak \rrepairrate\ is not so limited can achieve much larger \mttl .
For example, a version with a \rrepairrate\ limit set to be about five times the average rate for the given 
\nodef\ rate estimate achieves an \mttl\ of around $\expp{39}$ years 
for the $33.3\%$ \storeoverhead\ case.

\begin{figure}
\centering
\includegraphics[width=1 \textwidth]{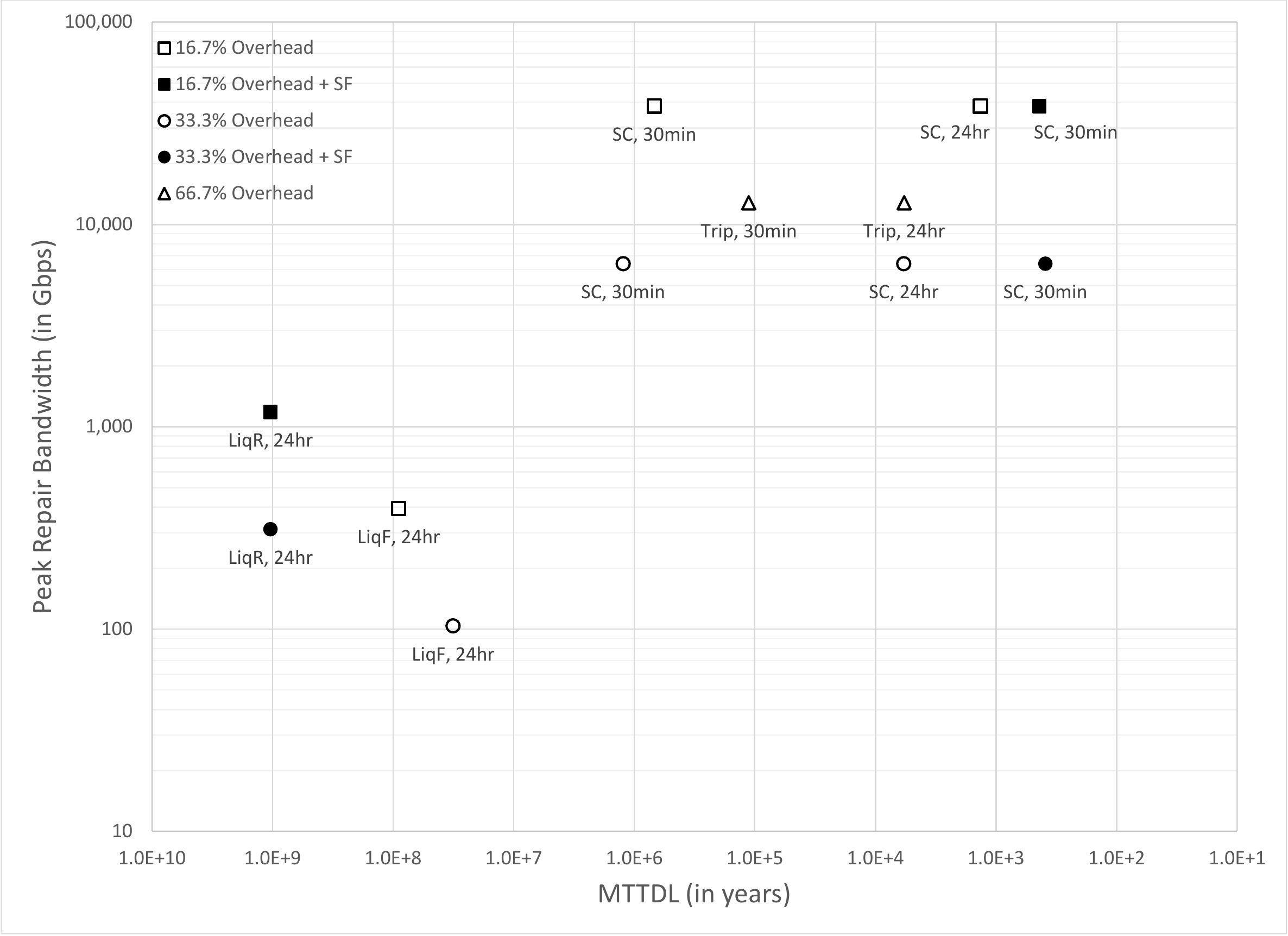}
\caption{Peak \rrepairrate\ versus \mttl\  for a $M=402$ node system for fixed $1/\lambdap$.} 
\label{fig_repair_plot_402_fixed}
\end{figure}

As can be seen, the value of $\Rpeak$ required for \liqsystem s is significantly smaller than that required
for \tradsystem s and for triplication.  Furthermore, although regulated rate \liqsystem\ can provide 
an  \mttl\ of greater than $\expp{25}$ years even when $\Trit$ is set to $24$ hours 
and there are \latentf s in addition to \nodef s, the \tradsystem s and triplication do not provide 
as good an \mttl\ even when $\Trit $ is set to $30$ minutes and there are no \latentf s, 
and provide a poor \mttl\ when $\Trit$ is set to $24$ hours or when there are \latentf s.
The \mttl\ would be further degraded for \tradsystem s if both $\Trit$ were set to 24 hours and there were \latentf s.

\subsection{Fixed \nodef\ rate with $3010$ nodes}
\label{simulations 3010 fixed sec} 

Fig.~\ref{fig_repair_plot_3010_fixed} shows detailed simulation results,
plotting the value of $\Rpeak$ and the resulting value of $\mttl$ for various $M = 3010$ node systems.
Each simulation was run for $\expp{8}$ years or until $200$ \objloss\ events occurred.
The number of nodes, the \nodef\ rate, and the \tradsystem s used 
in these simulations are similar to~\cite{Dimakis13}. 
The square icons correspond to $14.3\%$ \storeoverhead\ and the circle icons correspond to $28.6\%$ \storeoverhead. 
For the \tradsystem s,  a $(24,20,4)$ \tradcode\ is used for $14.3\%$ \storeoverhead,
and a $(14,10,4)$ \tradcode\ is used for $28.6\%$ \storeoverhead.  
For the \liqsystem s, a $(3010,2580,430)$ \liqcode\ is used for $14.3\%$ \storeoverhead,
and a $(3010,2150,860)$ \liqcode\ is used for $28.6\%$ \storeoverhead.  
The remaining parameters and terminology used in Fig.~\ref{fig_repair_plot_3010_fixed}
are the same as in Fig.~\ref{fig_repair_plot_402_fixed}.
There were no \objloss\ events in the $\expp{8}$ years of simulation  for any of the \liqsystem s
shown in Fig.~\ref{fig_repair_plot_3010_fixed}.

The systems ``LiqR, 24hr'' \liqsystem s used a regulated repair rate similar to that described for the
402 node system described above.  The target repair efficiency was set at $\frac{2}{3}r.$
In this case the average repair rate is higher than for the fixed rate case because the fixed rate
case was set using the target \mttl\ which, due the larger scale, admits more efficient repair than the
402 node system for the same target \mttl.
Estimates of $\mttl$ for the regulated case using techniques from Section~\ref{self-regulating sec} indicate that 
the \mttl\ is greater than an amazing $\expp{140}$ years for the $28.6\%$ \storeoverhead\ case
and $\expp{66}$ years for the $14.3\%$ \storeoverhead\ case even without the benefit of \nodef\ rate estimation. 
As before, the regulated repair system could be run at substantially lower overhead while still maintaining 
a large \mttl.
For the ``LiqR, 24hr'' \liqsystem s, the average \rrepairrate\ $\Ravg$ is around $1/3$ of $\Rpeak$, and 
the $99\%$ \rrepairrate\ $\Rnnn$ is around $1/2$ of $\Rpeak$.

The conclusions drawn from comparing \tradsystem s to \liqsystem s shown in Fig.~\ref{fig_repair_plot_402_fixed} 
are also valid for Fig.~\ref{fig_repair_plot_3010_fixed} but to an even greater degree.  The benefits
of the \liqsystem\ approach increase as the system size increases.

\begin{figure}
\centering
\includegraphics[width=1 \textwidth]{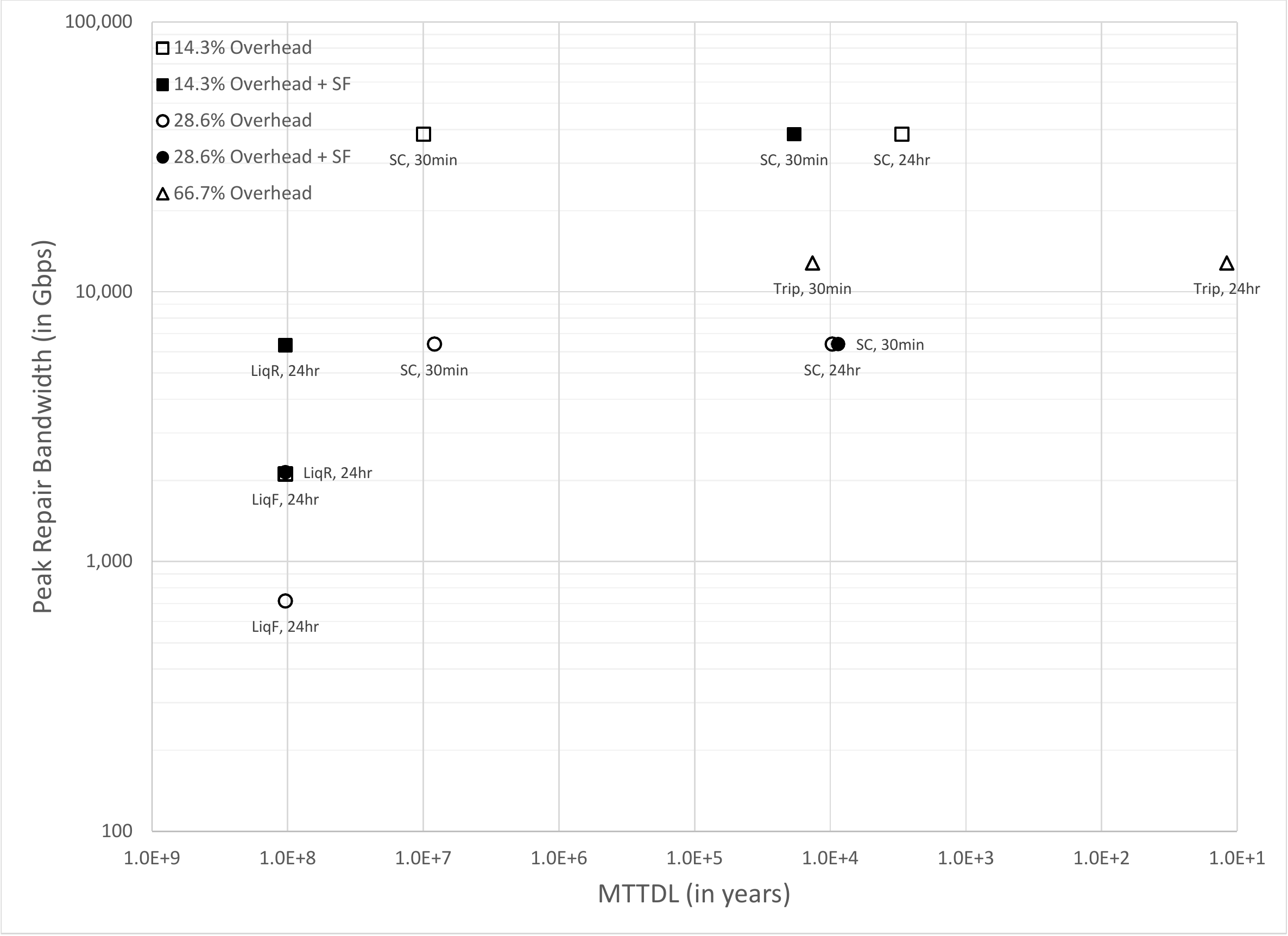}
\caption{Peak \rrepairrate\ versus \mttl\  for a $M=3010$ node system for fixed $1/\lambdap$.} 
\label{fig_repair_plot_3010_fixed}
\end{figure}

The paper \cite{Dimakis13} describes a {\em repair bottleneck} for a \tradsystem, 
hereafter referred to as the Facebook system, which has around $3000$ nodes and 
storage capacity per node of approximately $S = 15$ TB.
The Facebook system uses a $(14,10,4)$ RS code to protect 8\% of the \srcdata\
and triplication to protect the remaining 92\% of the \srcdata.
The amount of repair traffic estimated in \cite{Dimakis13} for the Facebook system 
is around 10\% to 20\% of the total average of $2$ PB/day of cluster network traffic. 
The paper \cite{Dimakis13} projects that the network would be completely 
saturated with repair traffic if even $50\%$ of the \srcdata\ in the Facebook system
were protected by the RS code.

Taking into account the mix of replicated and RS protected \srcdata\ and 
the mixed repair traffic cost for replicated and RS protected \srcdata, 
there are around $6.5$ to $13$ \nodef s per day in the Facebook system.
An average of $6.5$ to $13$ \nodef s per day for a $3000$ node system implies that each individual
node on average fails in around $115$ to $230$ days, which is around $2.5$ to $5$ times the \nodef\ rate of
$1/\lambdap = 3$ years we use in most of our simulations.
If all \srcdata\ were protected by a $(14,10,4)$ RS code in the Facebook system then the repair traffic per day 
would average around  $1$ to $2$ PB.  However, the \nodef\ per day statistics in \cite{Dimakis13}
have wide variations,  which supports the conclusions in \cite{Dimakis13} that protecting most of the \srcdata\
with the RS code in the Facebook system will saturate network capacity.  
 
Note that $1$ to $2$ PB/day of repair traffic for the Facebook system implies an average 
\rrepairrate\ of around $90$ to $180$ Gbps.
If the storage capacity per node were $S = 1$ PB instead of $S = 15$ TB then 
the Facebook system average \rrepairrate\ $\Ravg$ would be approximately $6$ to $12$ Tbps,
which is around $2.5$ to $5$ times the average \rrepairrate\ $\Ravg$ for the very similar 
``SC, 30min'' \tradsystem\ shown in Fig.~\ref{fig_repair_plot_3010_fixed}.
This relative difference in the value of $\Ravg$ makes sense, since the Facebook system \nodef\ rate is around 
$2.5$ to $5$ times larger than for the``SC, 30min'' \tradsystem. 

The ``SC, 30min'' \tradsystem\ shown in Fig.~\ref{fig_repair_plot_3010_fixed}
achieves an \mttl\ of just under $\expp{7}$ years using a peak \rrepairrate\ $\Rpeak = 6.4$ Tbps and using $\Trit = 30$ minutes.
Assume the \nodef\ rate for the Facebook system is $5$ times the \nodef\ rate 
for the ``SC, 30min'' \tradsystem\ shown in Fig.~\ref{fig_repair_plot_3010_fixed}, 
and suppose we set $\Trit = 30/5 = 6$ minutes for the Facebook system.
The scaling observations in Section~\ref{scaling params sec} show that 
this Facebook system achieves a \mttl\ of just under $\expp{7}/5 = 2 \cdot \expp{6}$ years when $\Rpeak= 6.4\cdot 5 = 32$ Tbps.

The average \nodef\ rates reported in \cite{Dimakis13}
for the Facebook system are considerably higher than the \nodef\ rates we use in our simulations.
One possible reason is that the Facebook system uses a small $\Trit$ value, e.g., $\Trit = 10$ minutes,
after which unresponsive nodes trigger repair.  This can cause unnecessary repair 
of fragments on unresponsive nodes that would have recovered if $\Trit$ were larger.

It is desirable to set $\Trit$ as large as possible to avoid unnecessary repair traffic, but only if a reasonable $\mttl$
can be achieved.  The results in Fig.~\ref{fig_repair_plot_3010_fixed} indicate that a \tradsystem\ 
cannot offer a reasonable $\mttl$ with $\Trit = 24$ hours. On the other hand, 
the results from Section~\ref{simulations transient sec} indicate that the average \rrepairrate\ increases
significantly (indicating a significant amount of unnecessary repair) if a smaller value of $\Trit$ 
is used for a \tradsystem\ when there are \transientf s. 
In contrast, a regulated \liqsystem\ 
can provide a very large $\mttl$ operating with $\Trit = 24$ hours and with \latentf s, 
thus avoiding misclassification of \transientf s and unnecessary repair traffic.

Results for \tradsystem s using LR codes~  \cite{Gopalan12}, \cite{Huang12}, \cite{Dimakis13} can be deduced 
from results for \tradsystem s using RS codes.  
A typical $(14,10,2,2)$ LR code (similar to the one described in \cite{Huang12} which defines a $(16,12,2,2)$ LR code) partitions $10$ source fragments into two groups of five,
generates one local repair fragment per group and two global repair fragments for a total
of $14$ fragments, and thus the \storeoverhead\ is the same as for a $(14,10,4)$ RS code.
At least five fragments are read to repair a lost fragment using the LR code, whereas $10$ fragments
are read to repair a lost fragment using the RS code.  However,
there are fragment loss patterns the LR code does not protect against that the RS code does.
Thus, the LR code operating at half the $\Rpeak$ value used for the RS code achieves
a \mttl\ that is lower than the \mttl\ for the RS code. 

Similarly, results for \cite{Dimakis13} which use a $(16,10,4,2)$ LR code can be compared 
against the $(14,10,4)$ RS code that we simulate. The LR code operating at half the $\Rpeak$ 
value used for the RS code achieves an \mttl\ that may be as large as the \mttl\ for the RS code.
However, the LR code \storeoverhead\ is $\beta = 0.375$, which is higher than
the RS code \storeoverhead\ of $\beta = 0.286$.

\subsection{Varying \nodef\ rate}
\label{simulations varying sec}

In this subsection we demonstrate the ability of the regulated repair system to respond to bursty \nodef\ processes.
The \nodef s in all simulations described in this subsection are generated by a time varying periodic Poisson process
repeating the following pattern over each ten year period of time:
the \nodef\ rate is set to $1/\lambdap = 3$ years for the first nine years of each ten year period, 
and then $1/\lambdap = 1$ year for the last year of each ten year period.
Thus, in each ten year period, the \nodef\ rate is the same as it was in
the previous subsections for the first nine years, followed by a \nodef\ rate 
that is three times higher for the last year.

Except for using variable $\lambdap$ for \nodef s,
Fig.~\ref{fig_repair_plot_402_variable} is similar to Fig.~\ref{fig_repair_plot_402_fixed}.
The \tradsystem s labeled as ``SC, 30min'' in Fig.~\ref{fig_repair_plot_402_variable} use the \rrepairrate\ strategy described 
in Section~\ref{SC repair subsec} based on the shown $\Rpeak$ value, and $\Trit$ is set to 30 minutes.
Thus, even when the \nodef\ rate is lower during the first nine years of each period,
the \tradsystem\ still sets \rrepairrate\ to $\Rpeak$ whenever repair is needed.

The \liqsystem s labeled as ``LiqR, 24hr'' in Fig.~\ref{fig_repair_plot_402_variable} use the regulator algorithm described 
in Section~\ref{self-regulating sec} to continually
and dynamically recalculate the \rrepairrate\ according to current conditions, 
allowing a maximum \rrepairrate\ up to the shown $\Rpeak$ value, and $\Trit$ is set to 24 hours.  
In these runs there were no \objloss\ events in the $\expp{9}$ years of simulation with both
\nodef s and \latentf s.  
[[We may want to compute something here   TBD?]]
%This is not unexpected, since we estimate
%(using the techniques described in Section~\ref{self-regulating sec}) that the
%actual $\mttl$ is greater than $\expp{70}$ years.

Fig.~\ref{fig:rr_spiky_graph_3x} is an example trace of the \rrepairrate\ as function of time for the
\liqsystem\ labeled ``LiqR, 24hr'' in  Fig.~\ref{fig_repair_plot_402_variable}, which shows how the
the \rrepairrate\ automatically adjusts as the \nodef\ rate varies over time.   
In these simulations, the average \rrepairrate\ $\Ravg$ is around $1/9$ of $\Rpeak$ during the first nine years
of each period and around $1/3$ of $\Rpeak$ during the last year of each period, and 
the $99\%$ \rrepairrate\ $\Rnnn$ is around $1/6$ of $\Rpeak$ during the first nine years
of each period and around $1/2$ of $\Rpeak$ during the last year of each period.
Thus, the regulator algorithm does a good job of matching the \rrepairrate\ to what is needed
according to the current \nodef\ rate.

Except for using variable $\lambdap$ for \nodef s,
Fig.~\ref{fig_repair_plot_3010_variable} is similar to Fig.~\ref{fig_repair_plot_3010_fixed}.
The conclusions comparing \tradsystem s to \liqsystem s shown in Fig.~\ref{fig_repair_plot_402_variable} 
are also valid for Fig.~\ref{fig_repair_plot_3010_variable}.

\begin{figure}
\centering
\includegraphics[width=0.82 \textwidth]{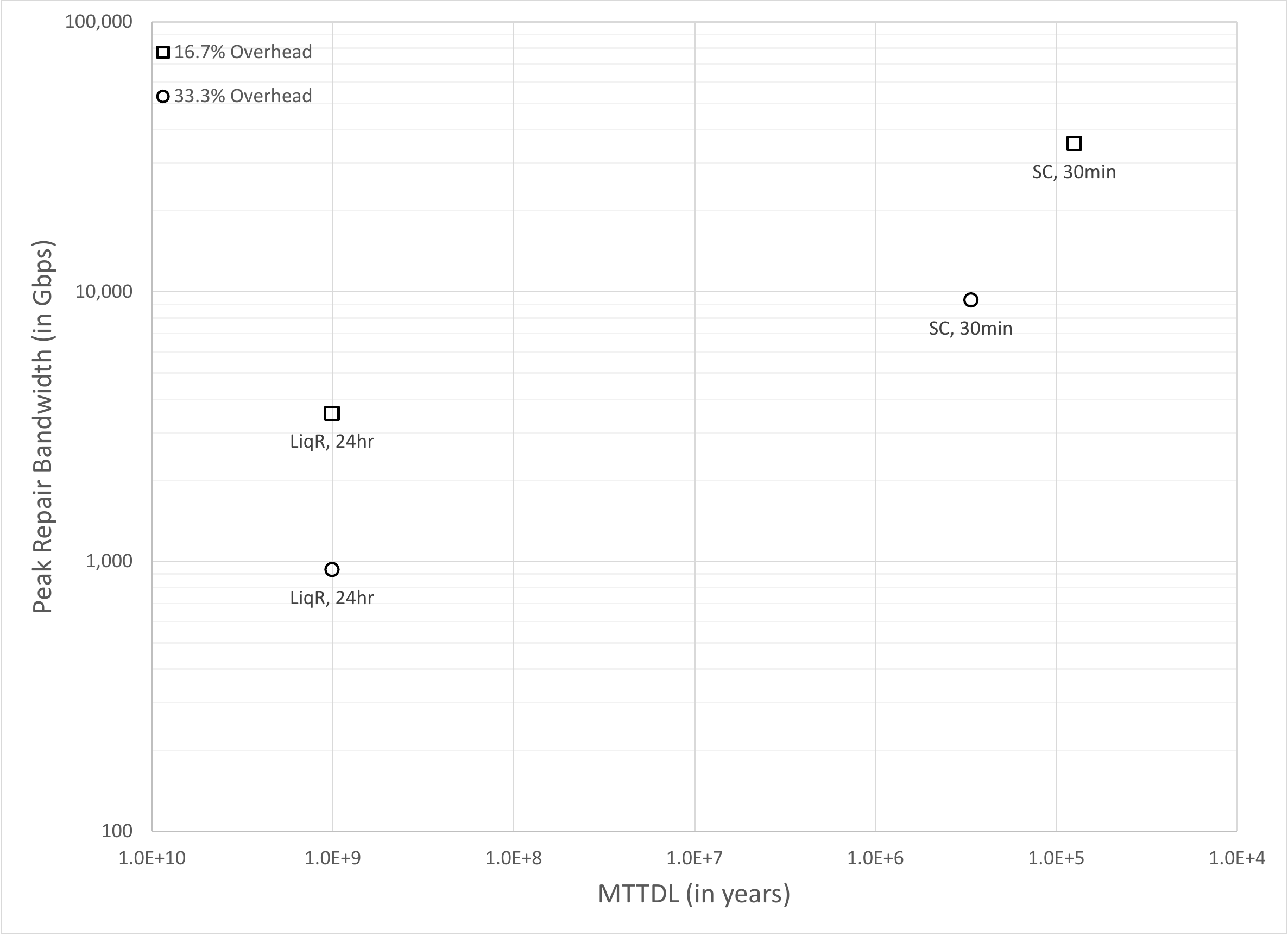}
\caption{Peak \rrepairrate\ versus \mttl\  for a $M=402$ node system for variable $1/\lambdap$.} 
\label{fig_repair_plot_402_variable}
\end{figure}

\begin{figure}
\centering
\includegraphics[width=0.82 \textwidth]{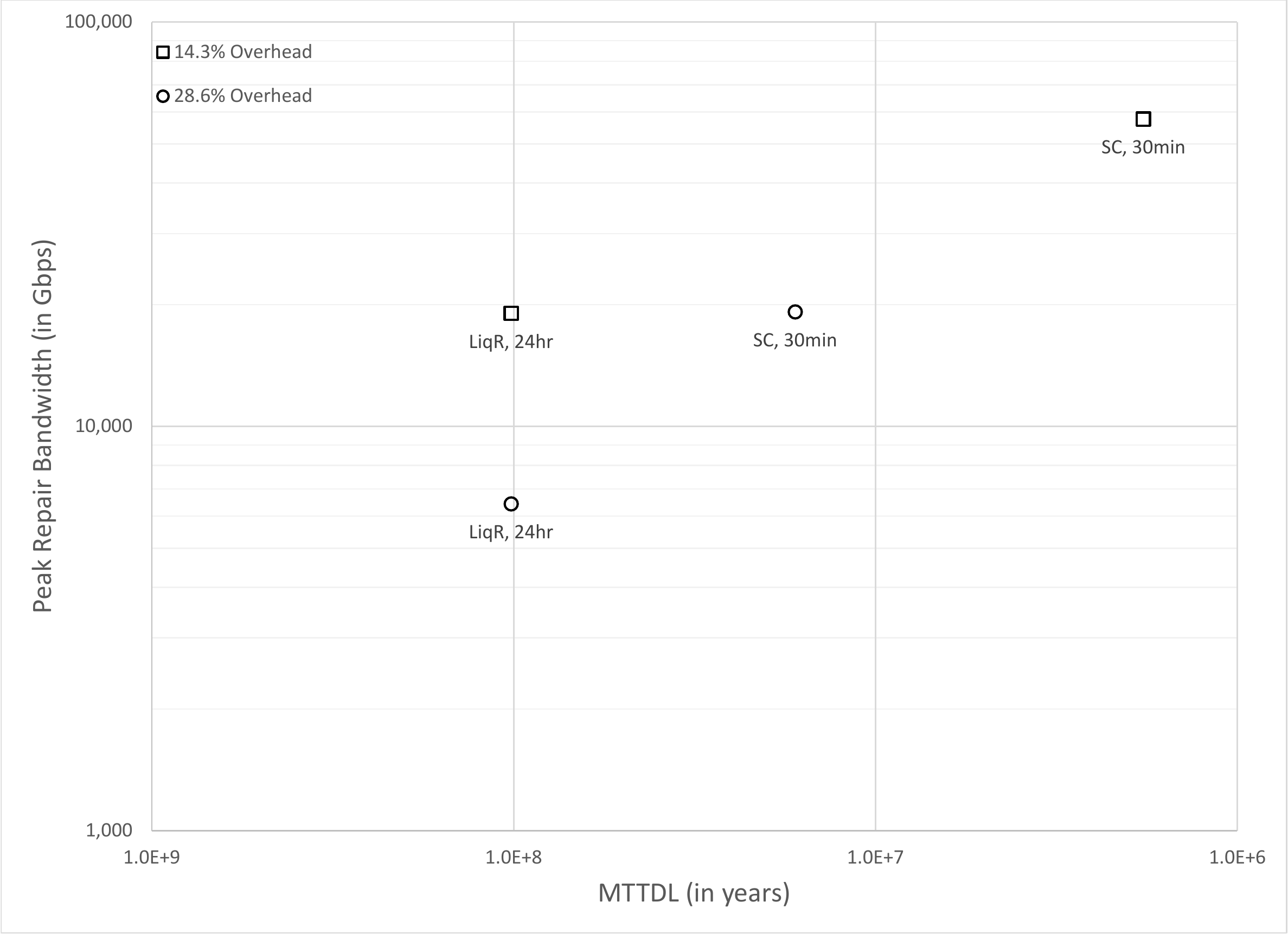}
\caption{Peak \rrepairrate\ versus \mttl\  for a $M=3010$ node system for variable $1/\lambdap$.} 
\label{fig_repair_plot_3010_variable}
\end{figure}

\subsection{\Transientf s}
\label{simulations transient sec}

In this subsection we demonstrate the ability of the \liqsystem s to
efficiently handle \transientf\ processes.

Fig.~\ref{fig_repair_plot_402_temp_fail} plots the simulation results
for $\Ravg$ against \mttl\ for a number of $M = 402$ node systems.  Like
in previous subsections, the node lifetimes for \nodef s are modeled as
independent exponentially distributed random variables with parameter
$\lambdap$ and is set to $1/\lambdap = 3$ years.  The occurrence times
of \transientf s are modeled as independent exponentially distributed
random variables with parameter $\lambdat$ and is set to
$1/\lambdat=0.33$ years. Thus, \transientf s occur at $9$ times the rate
at which \nodef s occur, consistent with \cite{Ford10}. The durations of
\transientf s are modeled with log-logistic random variables, having a
median of $60$ seconds and shape parameter $1.1$.  These choices were
made so as to mimic the distribution provided in \cite{Ford10}.  Figure
\ref{Ford figure} shows the graph from \cite{Ford10} with the log-logistic
fitted curve overlaid.  With this model, less than $10\%$ of the
\transientf{}s last for more than $15$ minutes.

\begin{figure}
\centering
\includegraphics[width=0.75 \textwidth]{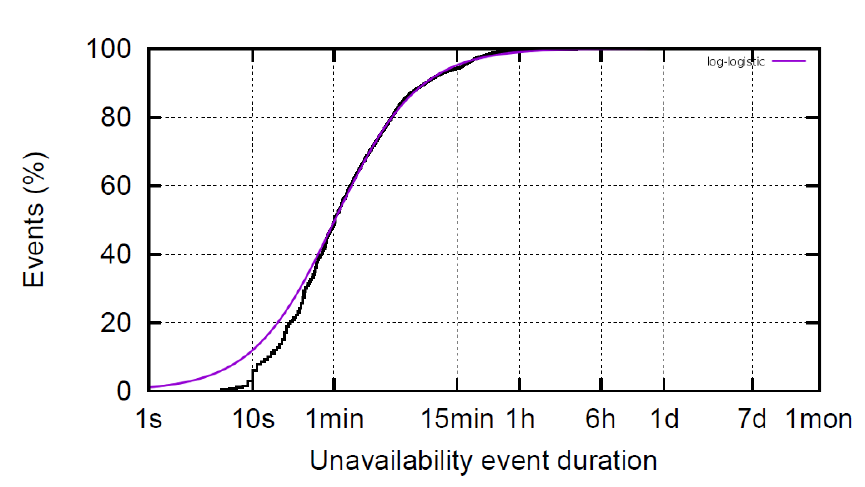}
\caption{Log-logistic transient failure duration model versus measured data}
\label{Ford figure}
\end{figure}

The unshaded markers mark simulations with just \nodef s and the shaded
markers mark simulations with both \nodef s and \transientf s.

The $\Rpeak$ values in all simulations is the same as $\Rpeak$ value
used correspondingly in Fig.~\ref{fig_repair_plot_402_fixed}.  The
\tradsystem s labeled as ``SC, 30min'' and  ``SC, 24hr`` in
Fig.~\ref{fig_repair_plot_402_temp_fail} use the \rrepairrate\ strategy
described in Section~\ref{SC repair subsec}, with $\Rpeak$ value of
$6400$ Gbps and $\Trit$ is set to 30 minutes and 24 hours respectively.
The \liqsystem\ labeled as ``LiqR, 24hr'' in
Fig.~\ref{fig_repair_plot_402_temp_fail} uses the regulator algorithm
described in Section~\ref{self-regulating sec}, with $\Rpeak$ value of
$311$ Gbps and $\Trit$ is set to 24 hours.  

As evident, there is no difference in the $\Ravg$ for the \liqsystem\
between simulations with \transientf s and those without.  No \objloss\
events were observed for the \liqsystem\ in the $\expp{9}$ years of
simulation with both \nodef s and \transientf s.  The $\Ravg$ for
\tradsystem\ labeled as ``SC, 30min`` is however higher for the
simulation with \transientf s and achieves an \mttl\ that is less than
half of what is achieved when there are no \transientf s. The $\Ravg$
for \tradsystem\ labeled as ``SC, 24hr`` is the same between simulations
with \transientf s and those without, but they fail to achieve a high
enough \mttl\ and do not provide adequate protection against \objloss.

Thus, the \liqsystem s do a better job of handling \transientf s without
requiring more effective average bandwidth than \tradsystem s.

\begin{figure}
\centering
\includegraphics[width=1 \textwidth]{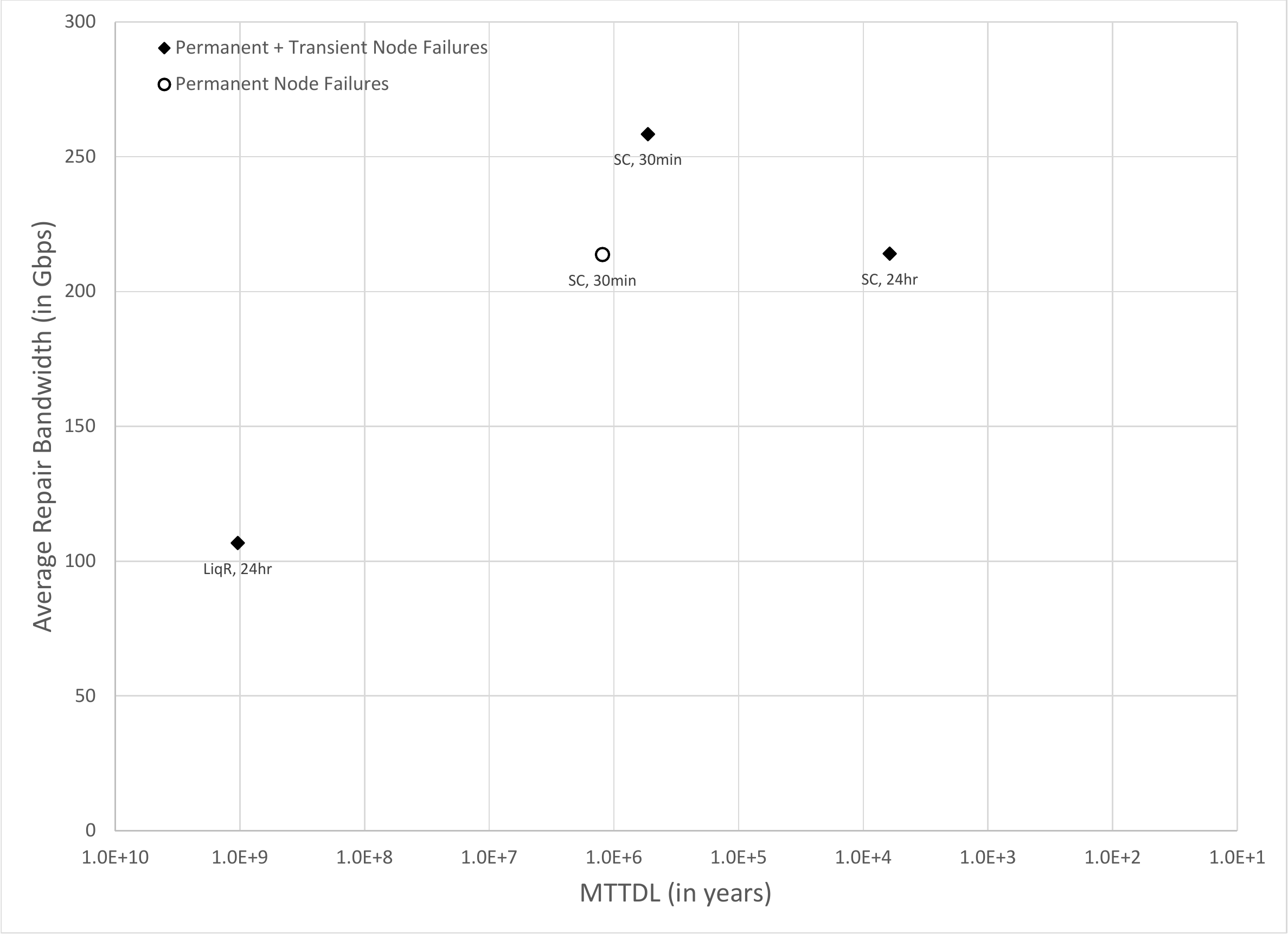}
\caption{Average \rrepairrate\ versus \mttl\ for a $M=402$ node system for transient $1/\lambdap$.} 
\label{fig_repair_plot_402_temp_fail}
\end{figure}

\subsection{Scaling parameters}
\label{scaling params sec}

The simulation results shown in Section~\ref{repair sim sec} are based on sets of parameters from
published literature, and from current trends in deployment practices.   As an example of a trend, a few years 
ago the storage capacity of a node was around $S = 16$ TB, whereas the storage capacity of a node in some systems
is $S = 1$ PB or more, and thus the storage capacity per node has grown substantially.  

Repair bandwidths scale linearly as a function of the storage capacity $S$ per node (keeping other input parameters the same).  
For example, the values of $(\Rpeak,\Ravg)$ for a system with $S = 512$ TB 
can be obtained by scaling down by a factor of two the values of $(\Rpeak,\Ravg)$ 
from simulation results for $S = 1$ PB,
whereas the $\mttl$ and other values remain unchanged.

Repair bandwidths and the $\mttl$ scale linearly as a function of concurrently scaling the node failure rate $\lambdap$ 
and the repair initiation timer $\Trit$.
For example, the values of $(\Rpeak,\Ravg)$ for a system with $1/\lambdap = 9$ years
and $\Trit = 72$ hours can be obtained by scaling down by a factor of three the values of $(\Rpeak,\Ravg)$ 
from simulation results for $1/\lambdap = 3$ years and $\Trit = 24$ hours, 
whereas the $\mttl$ can be obtained by scaling up by a factor of three, and other values remain unchanged.

\section{Encoding and decoding objects}
\label{fragment organization sec}

In a {\em \tradorg} commonly used by \tradsystem s, objects are partitioned into 
$k$ successive source fragments and encoded to generate $r=n-k$ repair fragments 
(we assume the \tradcode\ is MDS).
In effect a fragment is a symbol of the \tradcode, 
although the \tradcode\ may be intrinsically defined over smaller symbols, e.g., bytes.
A high level representation of a \tradorg\ for a simple example using a $(9,6,3)$ 
Reed-Solomon code is shown in Fig.~\ref{fig_blk_data_org},
where the object runs from left to right across the source fragments as indicated by the blue line.

\begin{figure}
\centering
\includegraphics[width=0.75\textwidth]{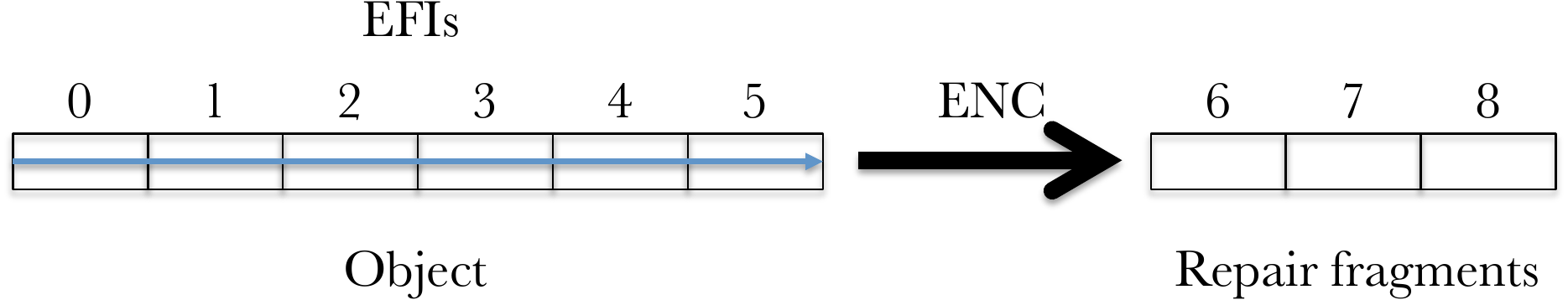}
\caption{Example of \tradorg. A data object is segmented into $k=6$ source fragments and encoded 
by a \tradcode\ to generate $r=3$ repair fragments.}
\label{fig_blk_data_org}
\end{figure}

A relatively small chunk of contiguous data from the object that resides within a 
single source fragment can be accessed directly from the associated node if it is available.  
If, however, that node is unavailable, perhaps due to a \transientf\ or \nodef, then $k$ corresponding chunks 
of equal size must be read from other nodes in the \placegroup\ and decoded to generate the desired chunk.   
This involves reading $k$ times the size of the missing chunk and performing a decoding operation,
which is referred to as a {\em degraded read}~\cite{Dimakis13},~\cite{ErasureCodingvsReplication}. 
Reading $k$ chunks can incur further delays if any of the nodes storing the chunks are busy and thus non-responsive. 
This latter situation can be ameliorated if a number of chunks slightly larger than $k$ are requested 
and decoding initiated as soon as the first $k$ chunks have arrived, as described in~\cite{Soljanin12}. 

A {\em \liqorg}, used by \liqsystem s, operates in a stream fashion. 
An $(n,k,r)$ \liqcode\ is used with small symbols of size $\symsize$, 
resulting in a small source block of size $\blksize = k \cdot \symsize$.  
For example, with $k = 1024$ and symbols of size $\symsize=64$ Bytes, 
a source block is of size $\blksize = 64$ KB.  
An object of size $\objsize$ is segmented into 
\[N = \frac{\objsize}{\blksize}\] source blocks, and the $k$ source symbols 
of each such source block is independently erasure encoded into $n$  symbols.  
For each $i = 0,\ldots,n-1$, the fragment with EFI $i$ consists of the concatenation of the 
$i^{\rm th}$ symbol from each of the $N$ consecutive source blocks of the object.
An example of a \liqorg\ showing a $(1536,1024,512)$ \liqcode\ with $N=24,576$ is shown in Fig.~\ref{fig_flow_data_org},
where each row corresponds to a source block and runs from left to right across the source symbols
as indicated by the dashed blue lines, and the object is the concatenation of the rows from top to bottom.

\begin{figure}
\centering
\includegraphics[width=0.75\textwidth]{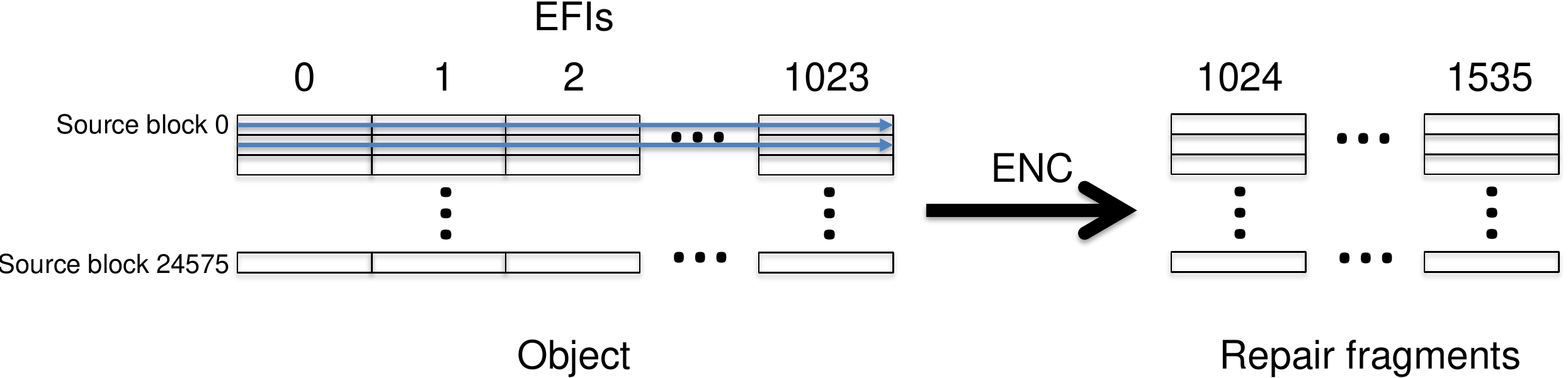}
\caption{Example of a \liqorg. A data object is segmented into $24,576$ source blocks, and each source block is segmented into $k=1024$ source symbols and encoded by a \liqcode\ to generate $r=512$ repair symbols.
The concatenation of the $i^{\rm th}$ symbol from each of the source blocks forms the fragment with EFI $i$.  Thus the vertical columns labeled by the EFIs correspond to fragments.}
\label{fig_flow_data_org}
\end{figure}

A chunk of data consisting of a consecutive set of source blocks can be accessed as follows.  
For a fragment with EFI $i$, the consecutive portion of the fragment that corresponds to the $i^{\rm th}$ 
symbol from each of consecutive set of source blocks is read.  When such portions from at least $k$ fragments
are received (each portion size a $1/k$-fraction of the chunk size), the chunk of data can be recovered 
by decoding the consecutive set of source blocks in sequence.
Thus, the amount of data that is read to access a chunk of data is equal to the size of the chunk of data.
Portions from slightly more than $k$ fragments can be read to reduce the latency due to nodes 
that are busy and thus non-responsive.  Note that each access requires a decoding operation.

Let us contrast the two organizations with a specific example using an object of size $\objsize = 1.5$ GB. 
Consider an access request for the $32$ MB chunk of the object in positions $576$ MB through $608$ MB,
which is shown as the shaded region in Fig. ~\ref{fig_access_compare} with respect to both a \tradorg\ and a \liqorg.
 
With the \tradorg\ shown in Fig.~\ref{fig_blk_data_org}, the $32$ MB chunk is within the fragment with EFI = $2$
as shown in Fig.~\ref{fig_access_compare}. 
If the node that stores this fragment is available then the chunk can be read directly from that node.  
If that node is unavailable, however, then in order to recover the $32$ MB chunk the 
corresponding $32$ MB chunks must be retrieved from any six available nodes 
of the \placegroup\ and decoded.  In this case a total of $192$ MB of data is 
read and transferred through the network to recover the $32$ MB chunk of data. 

\begin{figure}
\centering
\includegraphics[width=0.75\textwidth]{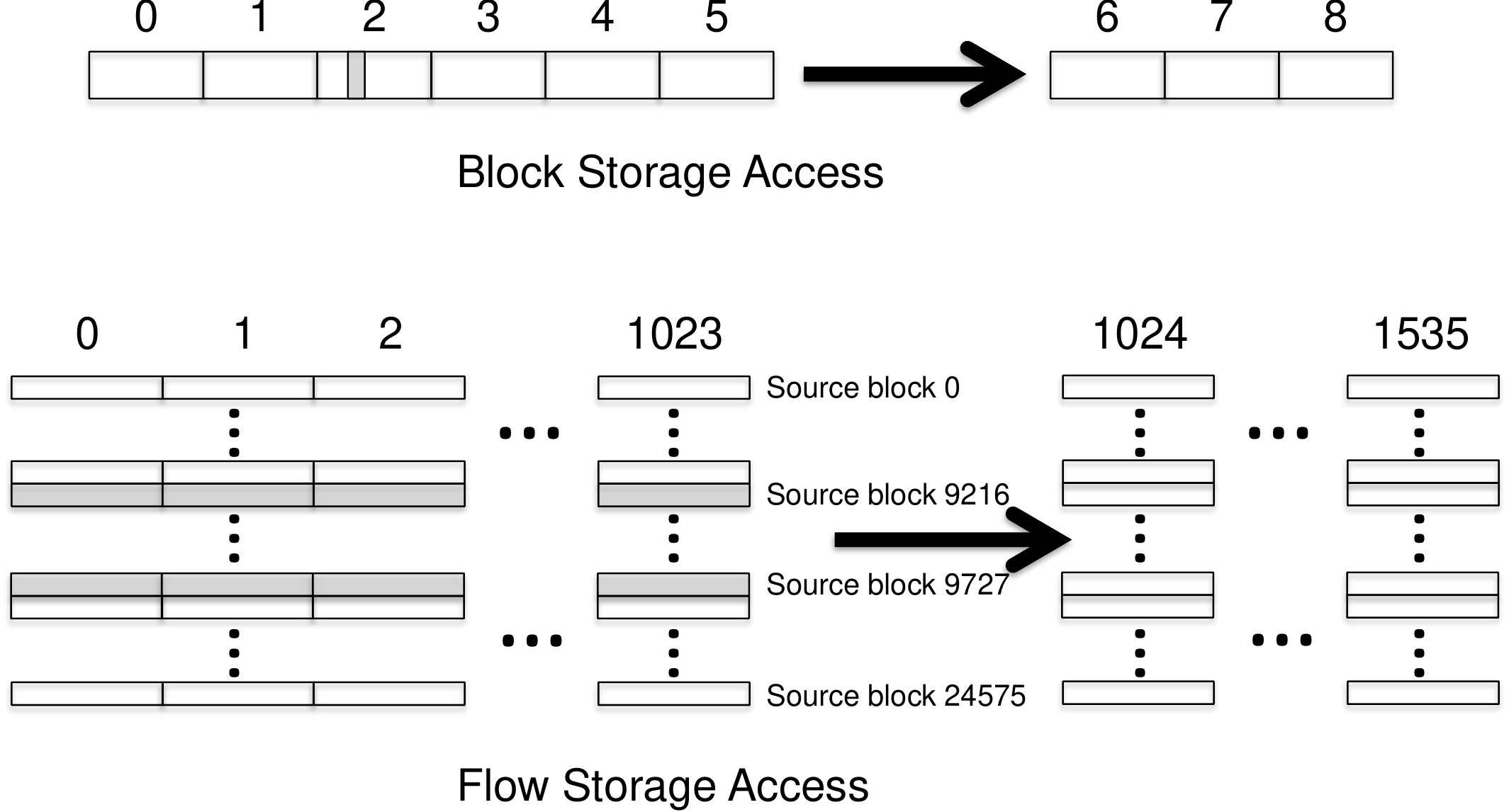}
\caption{Comparison of accessing a $32$ MB chunk of an object for \tradorg\ and \liqorg.  The shaded region indicates the chunk of the object to be accessed.}
\label{fig_access_compare}
\end{figure}

With the \liqorg\ shown in Fig.~\ref{fig_flow_data_org},
the chunk of size $32$ MB corresponds to $512$ consecutive source blocks from 
source block 9216 to block 9727 as shown in Fig.~\ref{fig_access_compare},
which can be retrieved by reading the corresponding portion of the fragment from each of at least $1024$ fragments. 
From the received portions, the chunk can be recovered by decoding the $512$ source blocks.
Thus, in a \liqorg\ a total of $32$ MB of data is read and transferred through the 
network to recover the original $32$ MB chunk of data. 

Naturally, the symbol size, source block size, and the granularity of the data to be accessed should be chosen to optimize the overall design, taking into account constraints such as whether the physical storage is hard disk or solid state.

\section{Prototype Implementation}
\label{implementation sec}

Our team implemented a prototype of the \liqsystem.  
We used the prototype to understand and improve the 
operational characteristics of a \liqsystem\ implementation in a real world setting,
and some of these learnings are described in Section~\ref{operational sec}.  
We used the prototype to compare access performance of \liqsystem s and \tradsystem s.
We used the prototype to cross-verify the \liqrepair\ simulator that uses a 
fixed \rrepairrate\ as described in Section~\ref{repair liquid sec} to ensure
that they both produce the same \mttl\ under the same conditions.
We used the prototype to validate the basic behavior of a regulated \rrepairrate\ 
as described in Section~\ref{self-regulating sec}.

The hardware we used for our prototype consists of a rack of 14 servers.  The servers are
connected via 10 Gbit full duplex ethernet links to a switch, and are
equipped with Intel Xeon CPUs running at 2.8 GHz, with each CPU having 20 cores.  
Each server is equipped with an SSD drive of 600 GB capacity that is used for storage.

The liquid prototype system software is custom written in the Go
programming language.  It consists of four main modules:
\begin{itemize}
\item
\emph{Storage Node} (SN) software.
The storage node software is a HTTP server.  It is used to store
fragments on the SSD drive.  Since a \liqsystem\ would generally
use many more than 14 storage nodes, we ran many instances of the
storage node software on a small number of physical servers, thereby
emulating systems with hundreds of nodes.

\item
An \emph{Access Generator} creates random user requests for user data 
and dispatches them to accessors.  It also collects the resulting access times 
from the accessors.

\item
\emph{Accessors} take user data requests from the access generator 
and create corresponding requests for (full or partial) fragments from the storage nodes.  
They collect the fragment responses and measure the amount of time it takes until 
enough fragments have been received to recreate the user data. 
There are in general multiple accessors running in the system.

\item A \emph{Repair Process} is used to manage the repair queue and
regenerate missing fragments of stored objects.  The repair process can be 
configured to use a fixed \rrepairrate, or to
use a regulated \rrepairrate\ as described in Section~\ref{self-regulating sec}.

\end{itemize}
In addition, the team developed a number of utilities to test and
exercise the system, e.g., utilities to cause \nodef s, and tools to collect data about the test runs.

\subsection{Access performance setup}

The users of the system are modeled by the access generator module.
We set a target user access load $\rho$ ($0 < \rho < 1$) on the system as follows.
Let $C$ be the aggregate capacity of the network links to
the storage nodes, and let $s$ be the size of the user data requests that will
be issued.  Then $C/s$ requests per unit of time would use all available capacity.  
We set the mean time between successive user requests to $t = \frac{s}{\rho \cdot C}$,
so that on average a fraction $\rho$ of the capacity is used.
The access generator module uses an exponential random variable to 
generate the timing of each successive user data request, where the mean of
the random variable is $t$.

As an example, there is a 10 Gbps link to each of the six storage nodes in the setup 
of Fig.~\ref{fig_access_prototype}, and thus $C = 60$ Gbps.  
If user data requests are of size $s =10$ MB $= 80 \cdot 2^{20}$ bits, 
and the target load is $\rho = 0.8$ ($80\%$ of capacity), 
then the mean time between requests is set to
$$t = \frac{80 \cdot 2^{20}}{0.8 \cdot 60 \cdot \expp{9}} \approx 1.75 \text{ milliseconds}.$$

The access generator module round-robins the generated user data requests across the accessor modules.
When an accessor module receives a generated user data request from the access generator module,
the accessor module is responsible for making fragment requests to the storage nodes and collecting response payloads; 
it thus needs to be fairly high performance.  
Our implementation of an accessor module uses a custom HTTP stack, 
which pipelines requests and keeps connections statically alive over extended periods of time.  
Fragment requests are timed to avoid swamping the switch with response data and this
reduces the risk of packet loss significantly.  The server network was also tuned, e.g.,  
the TCP retransmission timer was adjusted for the use case of a local network.  
See~\cite{IncastBerkeley} for a more detailed discussion of network issues in a cluster
environment.  These changes in aggregate resulted in a configuration that has low response latency 
and is able to saturate the network links fully.

\subsection{Access performance tests}
\label{access tests sec}

The prototype was used to compare access performance of \liqsystem s and \tradsystem s.  
Our goal was to evaluate access performance using \liqsystem\ implementations, i.e., 
understand the performance impact of requesting small fragments (or portions of fragments) 
from a large number of storage servers.  
For simplicity we evaluated the network impact on access performance.
Fragment data was generally in cache and thus disk access times were not part of the evaluation.
We also did not evaluate the computational cost of decoding.  
Some of these excluded aspects are addressed separately in Section~\ref{operational sec}.

Figure \ref{fig_access_prototype} displays the setup that we use for testing access speed.  
Each physical server has a full duplex 10 GbE network link which is connected over a switch to all the other
server blades.  Testing shows that we are able to saturate all of the 10
Gbps links of the switch simultaneously.  Eight of our 14 physical
servers are running accessors, and one of the eight is running the
access generator as well.  The remaining six servers are running
instances of the storage node software.  This configuration allows us to test the 
system under a load that saturates the network capacity to the storage nodes: 
The aggregate network capacity between the switch and the accessors is 80 Gbps, 
which is more than the aggregate network capacity of 60 Gbps between the
switch and the storage nodes, so the bottleneck is the 60 Gbps
between the switch and the storage nodes.

\begin{figure}
\begin{center}
\includegraphics[scale=0.5]{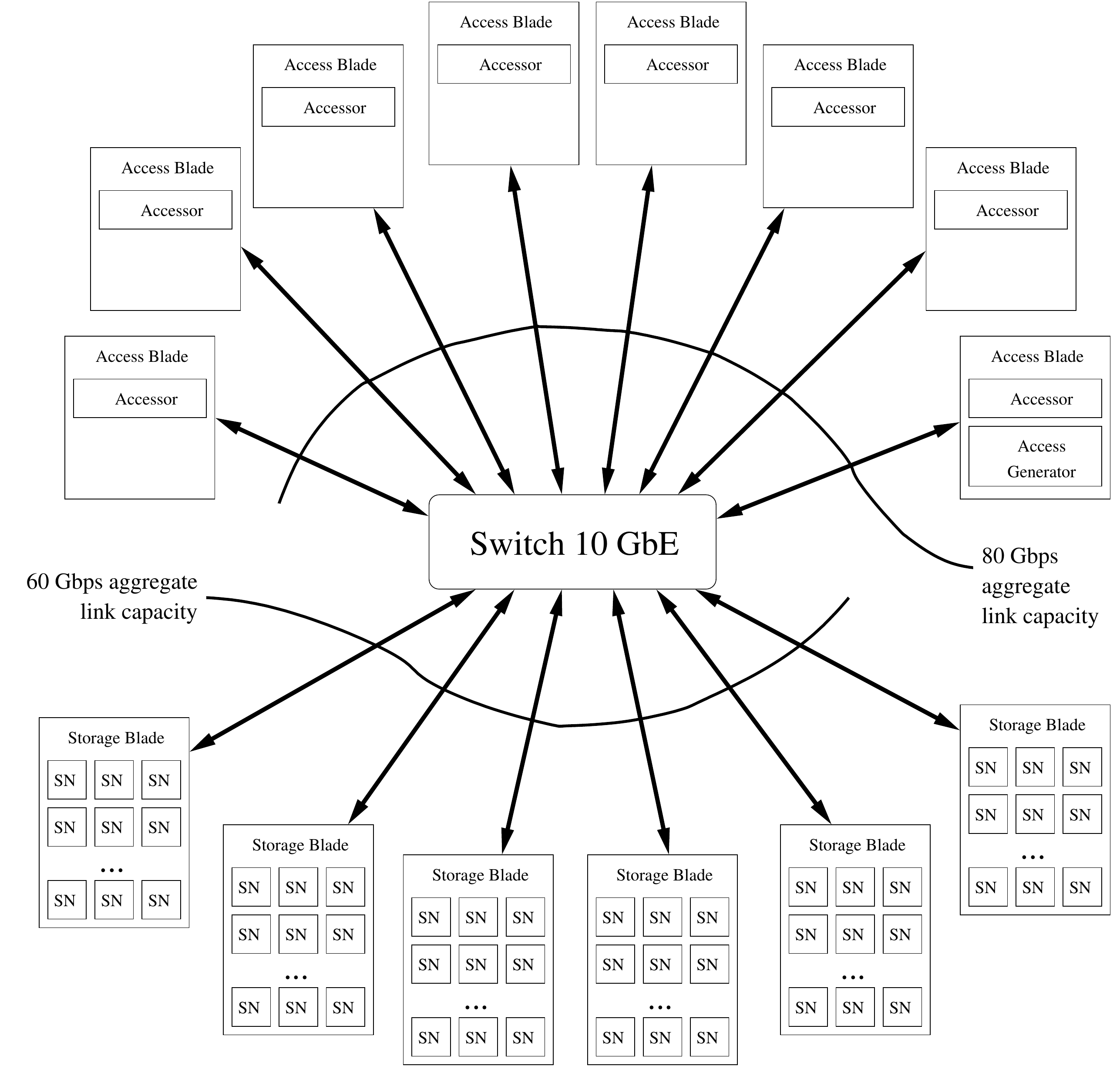}
\end{center}
\caption{Setup used for access performance testing}
\label{fig_access_prototype}
\end{figure}

% On each of the storage node servers there are 67 instances of the
% storage node software running, for a total of 402 storage node
% instances.  The system is thus modeled after one of the systems studied
% in section \ref{repair sim sec}.

% The access generator generates access requests at random, and dispatches
% them in a round-robin fashion to accessors.  The accessors issue the
% corresponding fragment requests, and collect the responses.  Accessors
% then return the access time data to the access generator.  For the
% access tests, we only measure data transfer time, and exclude
% processing times such as FEC decoding.

During access performance tests there are no \nodef s,
all $n$ fragments for each object are available,
and the repair process is disabled.
For all tests,  $67$ storage node instances run on each storage server, 
which emulates a system of $402$ nodes. We operate with a \storeoverhead\ of $33.3\%$, 
a $(402,268,134)$ \liqcode\ is used for the \liqsystem, and a $(9,6,3)$ \tradcode\ is used
for the \tradsystem.  We tested with $10$ MB and $100$ MB user data request sizes.

We run the access prototype in three different configurations.   
The ``Liq'' configuration models \liqsystem\ access of user data. 
For a user data request of size $s$, the accessor requests $k + E$ fragment portions,
each of size $s/k$, from a random subset of distinct storage nodes,
and measures the time until the first $k$ complete responses are received
(any remaining up to $E$ responses are discarded silently by the accessor).
Each fragment portion is of size around $25.5$ KB when $s = 10$ MB,
and around $255$ KB when $s = 100$ MB.
We use $E = 30$, and thus the total size of requested fragment portions
is around $10.75$ MB when $s=10$ MB, and around $107.5$ MB when $s=100$ MB.  
See \cite{Soljanin12} for an example of this access strategy.

The ``SC'' configuration models \tradsystem\ normal access of user data, 
i.e., the requested user data is stored on a storage node that is currently available.
For a user data request of size $s$, the accessor requests one fragment portion of size $s$ 
from a randomly selected storage node,
and measures the time until the complete response is received.
The total size of the requested fragment portion is $10$ MB when $s=10$ MB, 
and $100$ MB when $s=100$ MB.

The ``SC-Deg'' configuration models \tradsystem\ {\em degraded access} of user data,
i.e., the requested user data is stored at a storage node that that is currently unavailable, 
e.g., see~\cite{Dimakis13}, \cite{ErasureCodingvsReplication}.
For a user data request of size $s$, the accessor requests $k$ fragment portions 
of size $s$ from a random subset of distinct storage nodes and measures 
the time until the first $k$ complete responses are received.
The total size of requested fragment portions is $60$ MB when $s=10$ MB, 
and $600$ MB when $s=100$ MB.
Most user data is stored at available storage nodes,
and thus most accesses are normal, not degraded, in operation of a \tradsystem.
Thus, we generate the desired load on the system with
normal accesses to user data, and run degraded accesses at a rate that
adds only a nominal aggregate load to the system, and only the times
for the degraded accesses are measured by the accessors.
Similar to the ``Liq'' configuration, we could have requested $k+1$ fragment portions for the ``SC-Deg'' configuration
to decrease the variation in access times, but at the expense of even more data transfer over the network.

Figures \ref{fig_access10MB} and \ref{fig_access100MB} depict access time
results obtained with the above settings under different system loads. 
The average time for \liqsystem\ accesses is less than for \tradsystem\ normal accesses in most cases, 
except under very light load when they are similar on average.  
Even more striking, the variation in time is substantially smaller for \liqsystem s than for \tradsystem s
under all loads.
\Liqsystem\ accesses are far faster (and consume less network resources) than \tradsystem\ degraded accesses.

Our interpretation of these results is that \liqsystem s spread
load equally, whereas \tradsystem s tend to stress individual
storage nodes unevenly, leading to hotspots.  With \tradsystem s,
requests to heavily loaded nodes result in slower response times.  
Hotspots occur in \tradsystem s when appreciably loaded, 
even when the fragment requests are distributed uniformly
at random.  It should be expected that if the requests are not uniform,
but some data is significantly more popular than other data, 
then response times for \tradsystem s would be even more variable.

\begin{figure}
\centering
\includegraphics[scale=0.6]{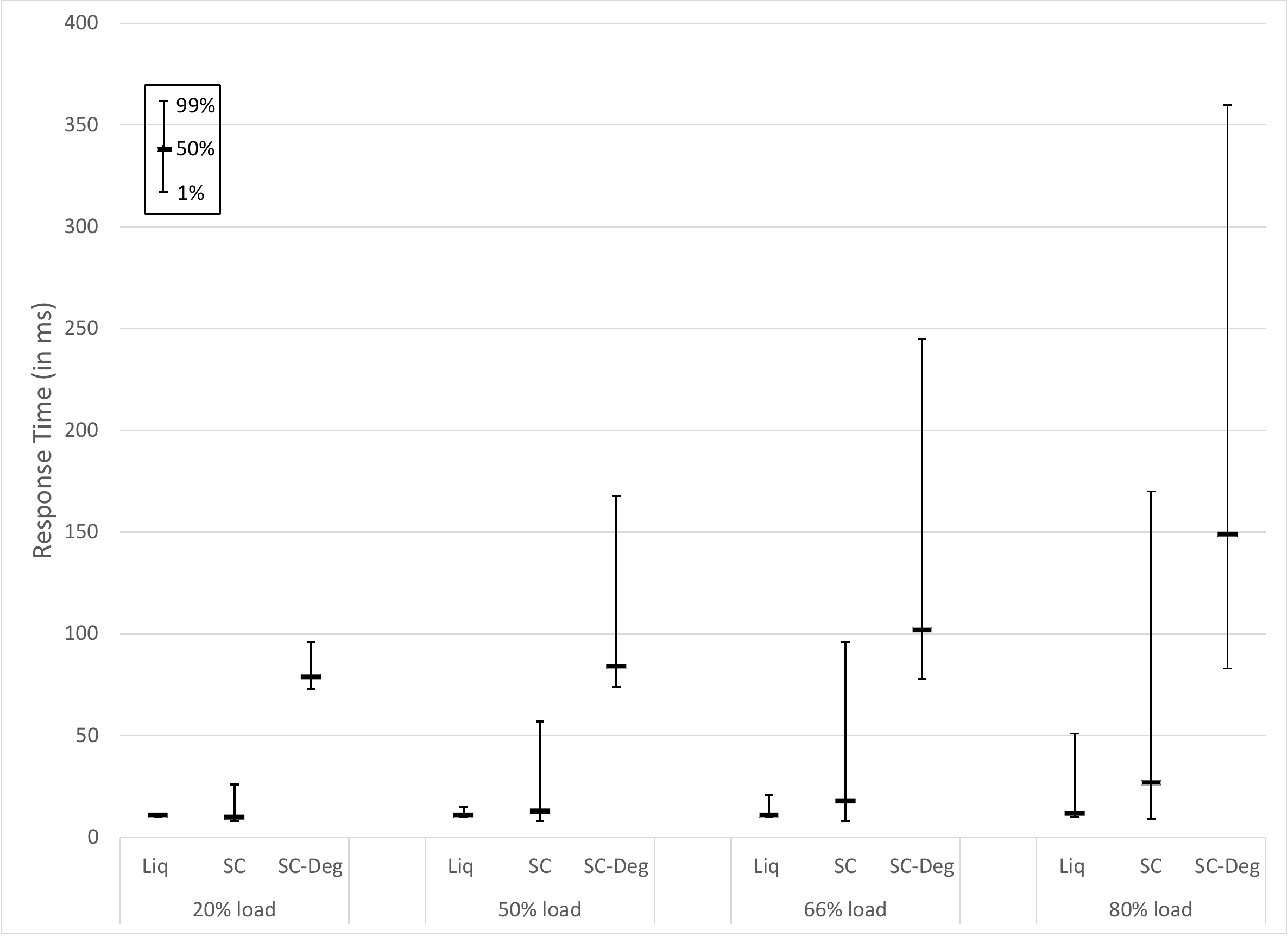}
\caption{Access Results for 10 MB object requests.}
\label{fig_access10MB}
\end{figure}

\begin{figure}
\centering
\includegraphics[scale=0.6]{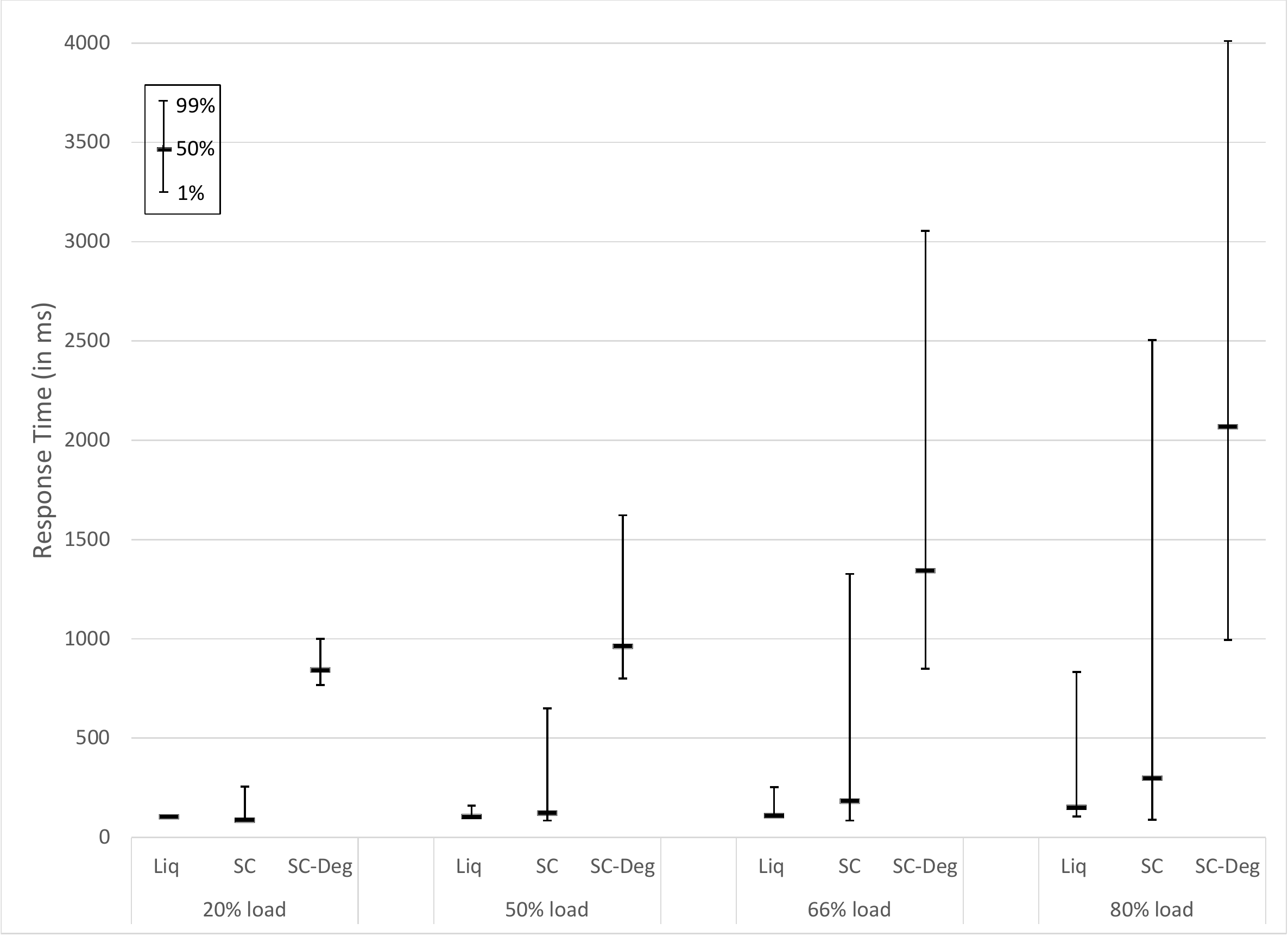}
\caption{Access Results for 100 MB object requests.}
\label{fig_access100MB}
\end{figure}

\subsection{Repair tests and verification of the simulator}

%  Simulator Verifikation
%  - Einigermassen reale Bedingungen (node failure rates, objects, etc)
%  - scaled down (speed up time by e.g. 10^6); tweak object read sizes
%  - run trace on the sped up system
%  - verify the outcome with simulated and calculated predictions.

We used the prototype repair process to verify our \liqsystem\ repair simulator.  
We started from realistic sets of parameters, and then sped them
up by a large factor (e.g., 1 million), in such a way that we could run
the prototype repair process and observe \objloss\ in a realistic amount of time.
We ran the prototype repair process in such a configuration, 
where \nodef s were generated artificially by software.  
The resulting measured statistics, in particular the \mttl, were compared to an 
equivalent run of the \liqsystem\ repair simulator.
This allowed us to verify that the \liqsystem\ repair simulator statistics matched 
those of the prototype repair process and also matched our analytical predictions.

\section{Implementation considerations}
\label {operational sec}

Designing a practical \liqsystem{} poses some challenges that are
different from those for \tradsystem{}s.  In the following sections we
describe some of the issues, and how they can be addressed.

\subsection{Example Architecture}

As discussed in Section~\ref{liqsystem overview sec},
Fig.~\ref{fig_lds_sys_arch} shows an example of a \liqsystem\ architecture.
This simplified figure omits components not directly related 
to \srcdata\ storage, such as system management components or access control units.

Concurrent capacity to access and store objects scales with the number of access proxy servers,
while \srcdata\ storage capacity can be increased by adding more storage nodes to the system,
and thus these capacities scale independently with this architecture.  
Storage nodes (storage servers) can be extremely simple, as their only function is 
to store and provide access to fragments, which allows them to be simple and inexpensive.
The \liqrepair\ process is independent of other components; thus
reliability is handled by a dedicated system.

%\subsection{Incorporating \liqsystem{}s into existing storage system
%solutions}
%
%It is also possible, and perhaps easier in some cases, to incorporate a
%\liqsystem{} into an existing storage system.  This makes it possible to
%reuse of large parts of a system that are not directly related to data
%storage and reliability itself, such as infrastructure to support
%authentication, access control, content encryption, indexing, system
%managment, and so forth.
%
%For a given system, there are generally several approaches in how a
%liquid system could be integrated into 
%
%XXX finish this part.
%
%\begin{figure}
%\centering
%\includegraphics[scale=0.5]{fig_ceph_integration.pdf}
%\caption{Integration of \liqsystem{} into Ceph}
%\label{fig_ceph_integration}
%\end{figure}

\subsection{Metadata}

The symbol size used by a \liqsystem\ with a \liqorg\
is substantially smaller than the size used by a \tradsystem\ with a \tradorg. 
For a \tradsystem\ each symbol is a fragment, 
whereas for a \liqsystem\ there are many symbols per fragment and, 
as described in Section~\ref{fragment organization sec}, the symbol to fragment mapping can be tracked implicitly.
For both types of systems, each fragment can be stored at a storage node as a file.

The names and locations of each fragment are explicitly tracked for \tradsystem s.
There are many more fragments to track for a \liqsystem\ because of the use of a \liqcode,
and explicit tracking is less appealing.
Instead, a mapping from the \placegroup\ to available nodes can be tracked,
and this together with a unique name for each object
can be used to implicitly derive the names and locations
of fragments for all objects stored at the nodes for \liqsystem s.

\subsection{Network usage}

The network usage of a \liqsystem\ is fairly different from that of a \tradsystem.  
As described in Section~\ref{access tests sec}, 
\liqsystem s use network resources more smoothly than \tradsystem s,
but \liqsystem s must efficiently handle many more requests for smaller data sizes.   
For example, a \tradsystem\ using TCP might open a new TCP
connection to the storage node for each user data request.  
For a \liqsystem, user data requests are split into many smaller fragment requests, 
and opening a new TCP connection for each fragment request is inefficient.  
Instead, each access proxy can maintain permanent TCP connections to each
of the storage nodes, something that modern operating systems can do easily.
It is also important to use protocols that allow pipeline data transfers over
the connections, so that network optimizations such as jumbo packets can be effectively used.  
For example the HTTP/2 protocol \cite{RFC7540} satisfies these properties.

Technologies such as RDMA \cite{RFC5040} can be
used to eliminate processing overhead for small packets of data.
For example, received portions of fragments can be placed autonomously via RDMA
in the memory of the access proxy, and the access proxy can decode or
encode one or more source blocks of data concurrently.

\subsection{Storage medium}

A key strength of \liqsystem s is that user data can be accessed as a stream,
i.e., there is very little startup delay until the first part of a user data request is available,
and the entire user data request is available with minimal delay,
e.g., see Section~\ref{access tests sec}.
Nevertheless, random requests for very small amounts of user data is more challenging for a \liqsystem.  

The storage medium used at the storage nodes is an 
important consideration when deciding the \liqcode\ parameters $(n,k,r)$. 
Hard disk drives (HDDs) efficiently support 
random reads of data blocks of size 500 KB or larger, 
whereas solid state drives (SSDs) efficiently support random reads
of data blocks of size 4 KB or larger.  
We refer to these data block sizes as the {\em basic read size} of the storage medium.
The parameters $(n,k,r)$ should be chosen so that if $s$ is the typical requested user data size 
then $s/k$ is at least the basic read size, which ensures that a requested fragment portion size
is at least the basic read size. 
In many cases this condition is satisfied, e.g., when the requested user data size is generally large,
or when the storage medium is SSD.  

In other cases, more advanced methods are required
to ensure that a \liqsystem\ is efficient.
For example, user data request performance can be improved  
by assigning different size user data to different types of objects:
Smaller user data is assigned to objects that are stored
using smaller values of $k$, and larger 
user data is assigned to objects that are 
stored using  larger values of $k$, where in each case
$s/k$ is at least the basic read size if $s$ is the requested user data size.  
In order to achieve a large \mttl\ for all objects,
the \storeoverhead\ $\frac{r}{n}$ is set larger
for objects stored using smaller values of $k$ 
than for objects stored using larger values of $k$.
Overall, since the bulk of the \srcdata\ is likely to be large user data,
the overall \storeoverhead\ remains reasonable with this approach.

Only a small proportion of user data is frequently accessed in many use cases, 
e.g., in the use case described in \cite{Facebook14}, only about 
10\% of the user data is accessed frequently, and this user data is easy to 
identify based on the amount of time it has been stored in the storage system.
Thus, a caching layer can be used to store and make such user data available
without having to request the user data from the storage system.  
For a \liqsystem, it makes most sense to cache user data that is both 
small and frequently accessed.

Another strategy to minimize requests for small user data is to store
related user data together in objects and coalesce accesses.  
Suppose for example, the storage system stores web content.  
To load a web page, dozens of small files of user data are typically loaded, 
such as HTML content and small images.
Since the same content is accessed whenever such a web page is
displayed, the system can store and access this user data together.

\subsection{RaptorQ software performance}
\label{raptorq sec}

With \tradsystem{}s, user data can usually be accessed directly
without requiring erasure decoding.  
With \liqsystem{}s, each access of user data typically requires erasure decoding.

RaptorQ codes are the \liqcode s of choice for \liqsystem s.
The encoding and decoding complexity of RaptorQ codes is inherently
linear by design \cite{CRaptorQ11}, i.e., for a $(n,k,r)$ RaptorQ code,
the encoding complexity is linear in $n$, and the decoding complexity is
linear in $k$, which makes the use of a RaptorQ code with large values
of $(n,k,r)$ practical.  

Software implementations of RaptorQ codes have been deployed in a number
of applications, including real-time applications that require low
encoding and decoding complexity.  Encode and decode speeds of over 10
Gbps are achieved using a software implementation of RaptorQ running on
one core of a 20 core Intel Xeon CPU ES-2680 server running at 2.8 GHz,
with code parameters $(n,k,r)=(1500,1000,500)$ and 256 byte symbols.
Aggregate encode and decode speeds of around 150 Gbps are achieved using
all 20 cores of the same server, with the same code parameters and
symbol size.

These numbers suggest that RaptorQ decoding would add a fractional extra
cost to the data processing pipeline: It is expected that other
transformations are performed when accessing user data, such as for example
checksumming, decompression and decryption.  Of those MD5 checksums can
be computed on the same machine at a rate of at most about 4 Gbps.  As
for decompression, Google's run-time performance oriented Snappy codec
achieves reportedly about 4 Gbps on a Core i7 processor
\cite{GoogleSnappy}, suggesting that RaptorQ coding is not a bottleneck.

\subsection{RaptorQ hardware performance}

Hardware based implementations have the potential to further cut down on 
power usage by the RaptorQ codes.  Based on a hardware
RaptorQ code chip design by Qualcomm for a different application, our
estimate is that the power usage in an ASIC implementation would be
about 3 Watts per 100 Gbps encoding or decoding speed.

To put this into perspective, we compare this to the power use of an SSD
to access data.  For example, the Samsung SM863a 2 TB enterprise SSD
drive \cite{SamsungSM863aSpec}, achieves a data transfer rate of up to
about 4.1 Gbps and uses about 2.5 W of power when reading.  This
translates to a power use of about  61 W to achieve 100 Gbps.  Thus, the
power needed for the RaptorQ code is about 5\% of the power needed to
read the data off a drive.  If other necessary aspects such as network
transfer, etc, were incorporated as well, the RaptorQ code share drops well below 5\%.

%\subsection{Access latency of the large code}
%
%For a small code system, the typical latency of accessing an object is
%the latency for accessing a single server.  Since nodes are unequally
%loaded in small code systems, that latency can vary significantly.
%Furthermore, in cases where the server storing the source data is
%unavailable, and the data has to be reconstructed using FEC coding, the
%latency can be higher.
%
%For a large code system, the access latency strongly depends on the
%strategy that is used to request.  If the requestor requests exactly $k$
%fragments, which is the minimum amount needed, then decoding can only
%start when the last request is being received.  That is, the worst
%response time of the $k$ nodes matter.  Since there is a small
%probability for any node to have a very high latency, the
%worst-among-$k$ behaviour is generally not favorable.
%
%However, very good response times can be achieved by requesting just a
%bit more than $k$ fragments.  Suppose a requestor requests $1.1 \cdot k$
%fragment.  Then decoding can start as soon as the $k$-th fragment is
%coming in, hence at about the $0.91$-th percentile of the respose times.
%This leads to a fairly stable response times, with variance
%substantially better than the single-node access of a small code
%solution.  In addition, because data from many nodes is received in
%parallel, and nodes are more equally loaded, the average time is lower
%as well.  XXX Reference.

\subsection{RaptorQ decode reliability}
\label{RaptorQnotMDS sec}

Extremely high durability is required for storage systems.  MDS codes
have the property that an object can be recovered from 
any $k$ fragments for the object.  
RaptorQ codes are essentially MDS, but not literally MDS:
it is not necessarily the case that an object can be recovered from any $k$ fragments.
We present simulation results showing the probability of \objloss\ due to 
RaptorQ codes not being MDS can be essentially ignored.

Consider a \liqsystem\ using a $(1200,1000,200)$ RaptorQ
code, where the target $w = 50$, i.e., the number of available fragments
for any object is always at least $k+w= 1050$.  If $1/\lambdap=3$
years, then on average there are $400$ \nodef s per year.
We group \nodef s together into consecutive batches of
$w-4 = 46$, and consider testing for decoding using the set of $k+w$
EFIs for available fragments at the beginning of the batch minus the $w-4 = 46$ EFIs 
for fragments lost within the batch, and thus we decode from 
$k+w-(w-4) = k+4 = 1004$ EFIs.  Note that if decoding is possible from this set of $1004$ EFIs
then decoding is possible for all objects at each point during the batch.  

As shown in Table~\ref{raptorq 1000 table}, there were no RaptorQ
decoding failures in an experiment where $\expp{11}$ random sets of
$1004$ EFIs were tested for decodability.  Since there are on average
$400$ \nodef s per year and the batch size is $46$,
there will be on average $\frac{400}{46} < 10$ such batches per year,
and thus  these $\expp{11}$ decoding tests cover $10$ billion years of
simulation time.  Thus there would be no decoding failures for RaptorQ
within $10$ billion years, which is a thousand times a target \mttl\ of
$10$ million years.   

\begin{table}
%\small
%\caption{$(1200,1000,200)$ RaptorQ properties}
%\begin{center}
%\begin{tabu} to 0.45\textwidth{r r} 
%\hline
%\hline 
%Number encoded symbols  &  Number of decode failures seen \\
%\hline
%$1000$ & $50 \cdot \expp{6}$ \\
%$1001$ & $250 \cdot \expp{3}$ \\
%$1002$ & $1.3 \cdot \expp{3}$ \\
%$1003$ & $11$ \\
%$1004$ & $0$ \\
%\hline
%\end{tabu}
%\end{center}
%\label{raptorq 1000 table}
\begin{center}
\begin{tabu} to \textwidth{| l || l| l| l |l| l |}
\hline
Number of EFIs used to decode &
	1000 &			1001 &
	1002 &			1003 &
	1004 \\ \hline 
Fraction of failed decodings &
	$4.9 \cdot \expm{3}$ &	$2.4 \cdot \expm{5}$ &
	$1.3 \cdot \expm{7}$ &	$1.1 \cdot \expm{9}$ &
	$0$ \\
\hline
\end{tabu}
\end{center}
\caption{Decode failures in $10^{11}$ decodings of a $(1200,1000,200)$ RaptorQ
code}
\label{raptorq 1000 table}
\end{table}

\section{\Rrepairrate\ regulation}
\label{self-regulating sec}

The simple Poisson process model of \nodef s discussed in
Section~\ref{node failures sec} is useful to benchmark the reliability
of a distributed storage system.  However, \nodef s in real systems are
not necessarily faithful to a simple statistical model.

Detailed statistical assumptions of the \nodef\ process do not
significantly affect \tradrepair\ decisions for a \tradsystem, i.e., the
actions taken by the \tradrepair\ process are largely event driven and
immediate.  For example, repairs are executed as soon as fragments are
lost due to \nodef s, with no dependency on assumptions on future node
failures.  The anticipated \nodef\ rate is only germane to the initial
provisioning of repair resources (e.g., total repair bandwidth).  

This is in sharp contrast to \liqrepair\ for a \liqsystem, where repairs
are postponed to achieve greater repair efficiency.  The degree to which
repairs can be postponed is clearly dependent on some expectation of
future \nodef s.  The tradeoff between repair efficiency and \mttl\
rests on assumptions about the long term \nodef\ rate.  In the
\liqsystem\ with fixed \rrepairrate\ this dependence appears prominently
in the choice of the \rrepairrate.  Lower \rrepairrate\ results in
greater repair efficiency but higher probability of \objloss.  This is
clearly evident in the (tight) \mttl\ estimate for the fixed rate system
derived in Appendix~\ref{liquid mttl analysis sec}. 

For example, consider a \liqsystem\ using a $(402,268,134)$ \liqcode.
If the \nodef\ rate $\lambdap$ is $1/3$ per year and the \rrepairrate\ is
set so that $\lambdap \cdot T = 0.21$ then the achieved \mttl\
(estimate) is  $3.6\cdot \expp{9}$ years.  If instead $\lambdap$ is
$10\%$ higher then the \mttl\ is  $1.0\cdot \expp{7}$ years, and if
$\lambdap$ is $20\%$ higher then the \mttl\ is $8.7\cdot \expp{4}$
years.  Thus, the \mttl\ of a fixed \rrepairrate\ design is sensitive to
error in knowledge of $\lambdap$.  If there is uncertainty in the
\nodef\ rate then either some repair efficiency must be sacrificed or a
reduction in \mttl\ must be tolerated for a fixed \rrepairrate\ design. 

\begin{comment}
\[
\mttl\ \simeq \frac{1}{\lambdap k q(r)}
\]
\[
q(r) = \binom(n,r) p^{n-r} (1-p)^r
\]
\[
p = e^{-\lambdap T}
\]

\end{comment}

Fortunately, \liqsystem s can be fitted with an adjustment algorithm
that smoothly and automatically adjusts the \rrepairrate\ in response to
\nodef\ conditions so that accurate a priori knowledge of \nodef\ rates
is not required.  Suppose for example that the \nodef\ rate increases at
some point in time, possibly due to component quality issues.  This will
result in a relative increase of the number of missing fragments among
objects in the repair queue.  Moreover, over time the increase in
\nodef\ rate can be detected by a \nodef\ rate estimator.  Both of these
occurrences can signal to the \liqrepair\ process a need to increase the
\rrepairrate.  Properly designed, a regulated \rrepairrate\ process can
achieve significantly larger \mttl than a fixed \rrepairrate\ process
running at the same average \rrepairrate, at the cost of some variation
in \rrepairrate.

Even in the case where the \nodef\ process is a Poisson process of known
rate, regulation of the \rrepairrate\ can bring significant benefit.  To
develop a good regulation method it is helpful to consider in more
detail the behavior of a fixed \rrepairrate\ process.  We assume a
\liqsystem\ with $n$ nodes and a \nodef\ rate $\lambdap$.  We assume a
fixed \rrepairrate\ that we characterize using $\tTot,$ the time
required to repair all objects in the system.  We work in the large
number of objects limit where the size of a single object is a
negligible fraction of all \srcdata.  Objects are repaired in a fixed
periodic sequence.  We view all objects as residing in a repair queue
and we use $x$ to denote the relative position in the queue.  Thus, an
object at position $x=1$ is under repair and the object in position
$x=0$ has just been repaired.  An object in position $x$ will be
repaired after a time $(1-x)\cdot \tTot$ and was repaired a time $x
\cdot\tTot$ in the past.

For an object currently in position $x$ let $f(x) \cdot n$ denote
the number of its fragments that have been erased due to node
failures.\footnote{Note here that while $f(x)\cdot n$ is integer valued
we generally consider $f(x)$ as arbitrarily valued on $[0,1].$ We will
make frequent use of this slight abuse of notation.} Assume $0\le x\le y
\le1$ and $0\le s \le t \le 1,$ and that the \rrepairrate\ is
characterized by $\tTot.$ The number of erased fragments $f(y) \cdot n$
the object will have when it reaches position $y$ is given by the
transition probability
\begin{equation}\label{eqn:transitionProbability}
\Prob{f(y) \cdot n = t \cdot n \mid f(x) \cdot n = s \cdot n}
=
B\left((1-t)\cdot n,(t-s)\cdot n,e^{-\lambdap \cdot \tTot \cdot
(y-x)}\right)
\end{equation}
where $B(n, m, q)$ is the probability mass at $m$ of a binomially
distributed random variable with parameters $(n, q)$.  That is,
\[
 B(n,m,q) = \binom{n}{m} \cdot q^{n-m} \cdot \bar{q}^m\,
\]
where we have introduced the notation $\bar{q}=1-q.$ Applying this to
the particular case with initial condition $f(0)\cdot n=0$ we have
\(
\Prob{f(x)\cdot n = m} = B(n,m,e^{-\lambdap \cdot\tTot \cdot x})\,.
\)
Using Stirling's approximation $\ln(n!) = n\cdot \ln(n) + O(\ln(n))$ we
can write 
\(
\ln(\Prob{f(x)\cdot n = f}) = -n \cdot E(f,\lambdap \cdot \tTot \cdot x) + O(\ln(n)) 
\)
where
\[
E(f,\lambdap \cdot \tTot \cdot x) = - H(f) + \bar{f} \cdot \lambdap \cdot 
\tTot \cdot x - f \cdot \ln(1-e^{ - \lambdap \cdot \tTot \cdot x})
\]
where $H(f) = - f \cdot {\rm log}_2 (f) -\bar{f} \cdot {\rm
log}_2(\bar{f})$ is the entropy function.  As a function of $f,$  the
exponent $E(f,\lambdap \cdot \tTot \cdot x)$ is minimized at  $f =
1-e^{ -\lambdap \cdot \tTot \cdot x},$ which is also the expected value
of $f(x),$ being the solution to the differential equation 
\[
\frac{d(1-f(x))}{dx}  = -\lambdap \cdot \tTot \cdot (1-f(x))\,
\]
which governs the expected value.  Thus, in the large system limit the
fraction of erased fragments as function of repair queue position
concentrates around the function $1-e^{ - \lambdap \cdot \tTot \cdot
x}.$  In particular we identify $n \cdot (1-e^{ - \lambdap \cdot \tTot
})$ as the number of fragments that are typically repaired.  We can
interpret the quantity $f_\rmtar := 1-e^{ - \lambdap \cdot \tTot }$ as a
{\em target} fraction of repair fragments.  Note that  $f_\rmtar$ and
$\lambdap \cdot \tTot$ are in a one-to-one relationship so, given the
\nodef\ rate $\lambdap,$ we can view $f_\rmtar$ as a specification of
the \rrepairrate.  This perspective on specification of system behavior
is a convenient one for consideration of regulated \rrepairrate\
systems.

To gain a further understanding of the fixed \rrepairrate\ system
behavior it is helpful to consider typical behavior under \objloss.  We
assume that \objloss\ occurs when an object has more than $r =\beta
\cdot n$ fragments erased and consider the number of erased fragments as
a function of queue position for such an object.  Thus, we consider the
distribution 
\(
\Prob{f(x) \cdot n \mid f(1) \cdot n > \beta \cdot n},
\)
which, using Bayes rule, can be written as
\begin{align*}
& \Prob{f(x) \cdot n = f \cdot n \mid f(1) \cdot n > \beta \cdot n} \\
= & \frac{\Prob{f(1) \cdot n > \beta \cdot n \mid f(x) \cdot n = f \cdot n} \cdot \Prob{f(x) \cdot n = f \cdot n)}}
{\Prob{f(1) \cdot n > \beta \cdot n}}
\end{align*}
Using \eqref{eqn:transitionProbability} and Stirling's approximation as
above, we find that the in the large $n$ limit the solution concentrates
around $f(1) \cdot n = \beta \cdot n,$ and, more generally,
\(
f(x) = \frac{\beta}{f_\rmtar} \cdot (1-e^{-\lambdap \cdot T  \cdot x})
\)
which is simply a scaled version of the nominal queue function
$1-e^{-\lambdap \cdot T \cdot x}$. Note that this solution satisfies the
equation
\(
\frac{d(1-f(x))}{dx} = \lambdap \cdot \frac{\beta}{f_\rmtar} \cdot (1-\frac{f_\rmtar}{\beta} \cdot f(x))\,.
\)
Thus, for small $x,$ where $f(x)$ is small, the solution corresponds to
a typical behavior for the system with \nodef\ rate $\lambdap \cdot
\frac{\beta}{f_\rmtar}.$   Another likely contribution to typical
\objloss\ is a skewing of \nodef s  to those nodes which, for the given
object in the given repair cycle, have not previously failed.  This
effect can explain the second factor $ 1-\frac{f_\rmtar}{\beta}\cdot
f(x)$ which replaces  $1-f(x)$ in the equation for the expected value.
Thus, we observe that \objloss\ typically involves a sustained increase
in the observed \nodef\ rate over an entire repair cycle and possible
skewing of the \nodef\ selection.  The fluctuation in \nodef\ rate can
be detected by a \nodef\ rate estimator.  Both the increased \nodef\
rate and the node selection skewing can be detected from the atypical
behavior of the repair queue.  We shall first consider a regulator that
assumes a known fixed \nodef\ rate $\lambdap$ and responds only to the
repair queue, and later extend the design to incorporate estimators of
\nodef\ arrival rates.

Under regulated repair we would expect that if the repair queue state is
typical, i.e., the fraction of erased fragments for an object in
position $x$ is near its nominal value of $1-e^{-\lambdap \cdot \tTot
\cdot x}$ then the \rrepairrate\ should be kept near its nominal value.
If, however, at some points in the repair queue the number of missing
fragments is larger than the nominal amount then the \rrepairrate\
should be increased.  As a general philosophy we adopt the view that
each object in the repair queue takes responsibility for itself to
ensure that the probability of \objloss\ for that object is kept small.
Each object, knowing its position in the queue and knowing its number of
erased fragments, desires a certain system \rrepairrate\ such that, if
adopted, would adequately protect that object against \objloss.  The
repair system will apply the highest of the desired \rrepairrate s thus
satisfying all \rrepairrate\ requests for all the objects.  

The regulator design defines a function $\phi(f,x)$ such that an object
in position $x$ with a fraction $f$ missing fragments requests a
\rrepairrate\ corresponding to a system repair time of
$\lambda^{-1} {\phi(f,x)}.$ Recall that a \rrepairrate\
corresponds to time required to repair all object in the system $\tTot$
and that $\lambda \cdot \tTot$ is the expected number of times a single
node fails during the time it takes to repair all objects.  In this
form, $\phi(f,x)$ expresses a desired \rrepairrate\ in terms of the
desired value for $\lambda \cdot \tTot.$ In general $\phi(f,x)$ will be
monotonically decreasing in $f$ and increasing in $x$.  Under these
assumptions, as well as some assumptions on node failure rate
estimators, it follows that the \rrepairrate\ will always be determined
by certain critical objects in the repair queue.  The critical objects
are those that were at the tail of the queue when a \nodef\ occurred
(see Appendix~\ref{liquid mttl analysis sec} for more detail).  Note
that the fixed \rrepairrate\ case corresponds to setting $\phi$
constant.

As a basis for designing the function $\phi(f,x)$ we adopt the notion
that objects desire to recover certain properties of the nominal
trajectory.  One quantity of particular interest is the probability of
\objloss\ for an object.  Given a position $x$ in the queue and a
fraction $f(x) <\beta$ erased fragments the probability that the object
will experience \objloss\ can be expressed using
\eqref{eqn:transitionProbability}.  For an object on the nominal
trajectory, $f(x) = 1-e^{\lambdap \cdot T \cdot x} = 1 -
(1-f_\rmtar)^x$, this probability is easily seen to be decreasing in
$x$.   Thus, if the object enjoys typical behavior early in the repair
queue then its probability of \objloss\ decreases.  As a design
principle, therefore, we can stipulate that an object in position $x$
with a fraction $f$ of erased fragments will request a system repair
parameter $\lambdap \cdot T,$ as represented by the function
$\phi(f,x),$ such that the probability of \objloss\ for that object
would be the same as an object at that position on the nominal
trajectory.  Under this stipulation, we obtain $\phi(f,x)$ implicitly as
the solution to \[ \sum_{s \cdot n > \beta \cdot n} B(\bar{f}_\rmtar^x
\cdot n,(\bar{f}_\rmtar^x-\bar{s}) \cdot n,\bar{f}_\rmtar^{1-x}) =
\sum_{s \cdot n > \beta \cdot n} B(\bar{f}\cdot n,(\bar{f}-\bar{s})
\cdot n,e^{-\phi(f,x)\cdot (1-x)})\,.  \] Note that in this context the
quantity $\beta \cdot n$ represents a particular threshold that the
regulator attempts avoid reaching and need not necessarily correspond
with the erasure correcting limit of the used erasure code.  For that
reason we introduce a different notation $f_T$ for the target hard
threshold.  In practice one will typically have $f_T \simeq \beta.$

Simpler expressions of the above principle for deriving appropriate
$\phi$ functions can be obtained in a number of ways.  One possibility
is note that the above sums are dominated by the term with $s \cdot
n=\beta \cdot n$ and then, using the Stirling approximation, to equate
the rate functions on both sides.  Another approach is to use a Gaussian
approximation of the above binomial distributions and equate the
integrals that correspond to the above sums.  Under the Gaussian
approximation we obtain the equation \begin{align}\label{eqn:sepGovForm}
\frac{1-e^{-(1-x) \cdot \phi(f,x) }}
{1-e^{-(1-x) \cdot \phi_\rmnom}}
=
\left(\frac{\bar{f}_\rme-\bar{f}_T }{\bar{f}_\rmtar -\bar{f}_T}\right)^2
\cdot \frac{\bar{f}_\rmtar}{\bar{f}_\rme }
\end{align}
where $\phi_\rmnom = -\ln (1-f_\rmtar)$ and $f_\rme$ denotes
the expected fraction of erased fragments that will be repaired for an
object in position $x$ with $f$ erased fragments assuming that the
\rrepairrate\ is subsequently fixed according to $\phi(f,x),$ i.e., we have
$\bar{f}_\rme = \bar{f} e^{-\phi(f,x) \cdot (1-x)}.$ Substituting this
into the above we obtain a quadratic equation for $f_\rme$ whose
solution then determines $\phi$. A good approximation is 
\[
 \frac{1-e^{-(1-x)\cdot\phi(f,x) }}
{1-e^{-(1-x)\cdot \phi_\rmnom}}
\simeq
\frac{\phi(f,x)}{\phi_\rmnom}
\]
under which the level curves of $\phi(f,x)$ correspond exactly to asymptotic
trajectories for the regulated system that can be parameterized by their
implied target fragment repair fraction $f_\rme.$  This implies that an object that
is controlling the repair rate expects on average to hold the repair constant until it is repaired. 

For the simulation results given below we used the Gaussian
approximation approach to obtain $\phi.$ It is not hard to show that
$\phi$ is bounded above (\rrepairrate\ never goes to $0).$ A lower bound
is $0$ (infinite \rrepairrate) which is requested if the $f(x)$ is
sufficiently close to $f_T.$ In practice the \rrepairrate\ will obviously
be limited by available resources, and such a limit can be brought into
the definition of $\phi$ by limiting its minimum value and for
simulations we typically saturate $\phi$ at some postive value, often
specified as fraction of $\phi_\rmnom.$

Recall that $\lambdap^{-1}\phi$ represents desired time to repair all
objects precisely once.  This indicates how to incorporate estimates for
the \nodef\ rate into the operation of the regulator:  one simply
replaces $\lambdap$ in the expression $\lambdap^{-1}\phi$ with a
suitable estimate of $\lambdap$. In general, an estimator of the node
failure rate will be based on past \nodef s . One pertinent quantity
in this context is the appropriate time scale to use to form the
estimate. We may choose to vary the time scale as a function of the
number of erased fragments the object experiences, or other aspects of
the \nodef\ process.  

Let us assume a system with a total of $\nObjects$ objects of equal
size.  Then each object has a position in the queue $x$ with 
$x \cdot \nObjects \in \{ 0,...,\nObjects-1 \}.$  When an object is repaired the
position $x$ of each object is incremented by $1/\nObjects.$ According
to the regulated \rrepairrate\ process the repair time for the next object is
then given by
\[
\frac{\min_{x}  \phi (f(x),x)}{\lambdap \cdot \nObjects}
\]
where $f(x)$ denotes the fraction of erased fragments for the object in
position $x.$  Here we tacitly assume $f(x) < f_T$ for all $x$ but for
analytical purposes we may extend the definition of $\phi$ to all $f$
and $x$ in $[0,1].$ In practice $\phi$ is effectively bounded from below
by the repair speed limitations of the system.  In our simulations we
typically assume a floor for $\phi$ of $\gamma \cdot \phi_\rmnom$ with
$\gamma$ set to a small fraction, e.g. $\gamma=\frac{1}{3}.$ Note
further that we always have 
$\phi(0,0) = \phi_\rmnom = -\ln(1-f_\rmtar)$ 
and so this gives a lower bound on the \rrepairrate\ and we
can impose a ceiling on the function $\phi$ of value $\phi_\rmnom$
without altering system behavior.

\begin{comment}
We now give a rough description of the \rrepairrate\ adjustment algorithm.
The algorithm has a number $g$ ($0 \leq g \leq r$) given as its input.
The value $g$ is a target number of fragments in need of repair at
object repair time, and is sometimes called the \emph{repair
efficiency}.  The idea is that the algorithm will attempt to pick the
\rrepairrate\ such that each object, at the time it is being repaired, has
approximately $g$ fragments missing.  When the \nodef\ rate
fluctuates, the algorithm attempts to change the \rrepairrate\ accordingly
to maintain this property.
\end{comment}
The choice of $f_\rmtar$ involves a tradeoff:  lower values of
$f_\rmtar$ lead to higher \rrepairrate s, but also a more reliable system,
and one that is potentially less prone to large \rrepairrate\ variations,
whereas a larger value of $f_\rmtar$ gives higher repair efficiency at
the price of larger \rrepairrate\ fluctuations and a reduced \objloss\ 
safety margin of the system.  For simplicity, all the results that we
present in this section with the repair adjustment algorithm were
obtained by running the algorithm with $f_\rmtar = \frac{2\cdot r}{3}$, which
is a reasonable choice in all but some contrived circumstances.
%Moreover, in the simulations the \rrepairrate\ was updated upon node failure.

\begin{comment}
Whenever a \nodef\ is registered, the adjustment algorithm
reevaluates the situation and decides about a possible change of the
\rrepairrate.  To do this, it evaluates each object $o$ in the repair
queue as follows: It first computes the number of bytes $b_o$ ahead of
it in the repair queue, and then the number of node further node
failures $z_o$ that can be tolerated so that the target $g$ is
approximately met for that object.  A \nodef\ rate estimate is then
used to compute the estimated time $t_o$ for $z_o$ \nodef s  to
occur.  A good \rrepairrate\ estimate for object $o$ is then $R_o := b_o /
t_o$, i.e., the rate needed so that $o$ will be repaired when it's about
reached the target $g$.  After inspection of all the objects, the repair
rate proposed by the adjustment algorithm is then the maximum of all the
$R_o$, where $o$ runs over all the objects in the repair queue.
\end{comment}

For certain forms of failure rate estimators it is feasible to compute
lower bounds on \mttl.  The details of the models and the calculations
can be found in Appendix~\ref{liquid mttl analysis sec}.
The bounds can be used to designing system parameters such as $f_T$ and more generally $\phi$
so as to  maximize repair efficiency while ensuring a desired MTTDL.

%Some of the finer points in this algorithm, e.g., details on the target
%$g$ is tracked or how the \nodef\ rate is estimated can be found in
%Appendix~\ref{selfregulating details sec}.

\subsection{Regulated \rrepairrate\ simulation}
\label{sim reg sec}

Fig.~\ref{fig:rr_fixed_graph} illustrates how the regulated
\rrepairrate\ automatically varies over time when the \nodef\ rate is
$1/3$ years for a \liqsystem\ using a $(402,268,134)$ \liqcode\ with
\storeoverhead\ $33.3\%$, and the first row of
Table~\ref{tab_regualator_examples} shows a comparison between the
regulated \repairrate\ with an average of $\Ravg = 107$ Gbps and a fixed
\rrepairrate\ of $\Rpeak = 104$ Gbps (set so that the \mttl\ is $\approx
\expp{7}$ years).  In this table, the \mttl\ of $\gg \expp{9}$ for the
regulated \rrepairrate\ indicates there were no object losses in
$\expp{9}$ years of simulation (based on estimates using the results
from Appendix~\ref{liquid mttl analysis sec}, the \mttl\ is orders of
magnitude larger than $\expp{9}$ years), and $\Rpeak$ is a hard upper
limit on the used \rrepairrate.  The variability of the regulated
\rrepairrate\ is moderate, e.g., $99.99\%$ of the time the \rrepairrate\
is less than $231$ Gbps.  The second row of
Table~\ref{tab_regualator_examples} shows a comparison using a
$(402,335,67)$ \liqcode\ with \storeoverhead\ $16.7\%$.
\begin{figure}
\centering
\includegraphics[scale=0.5]{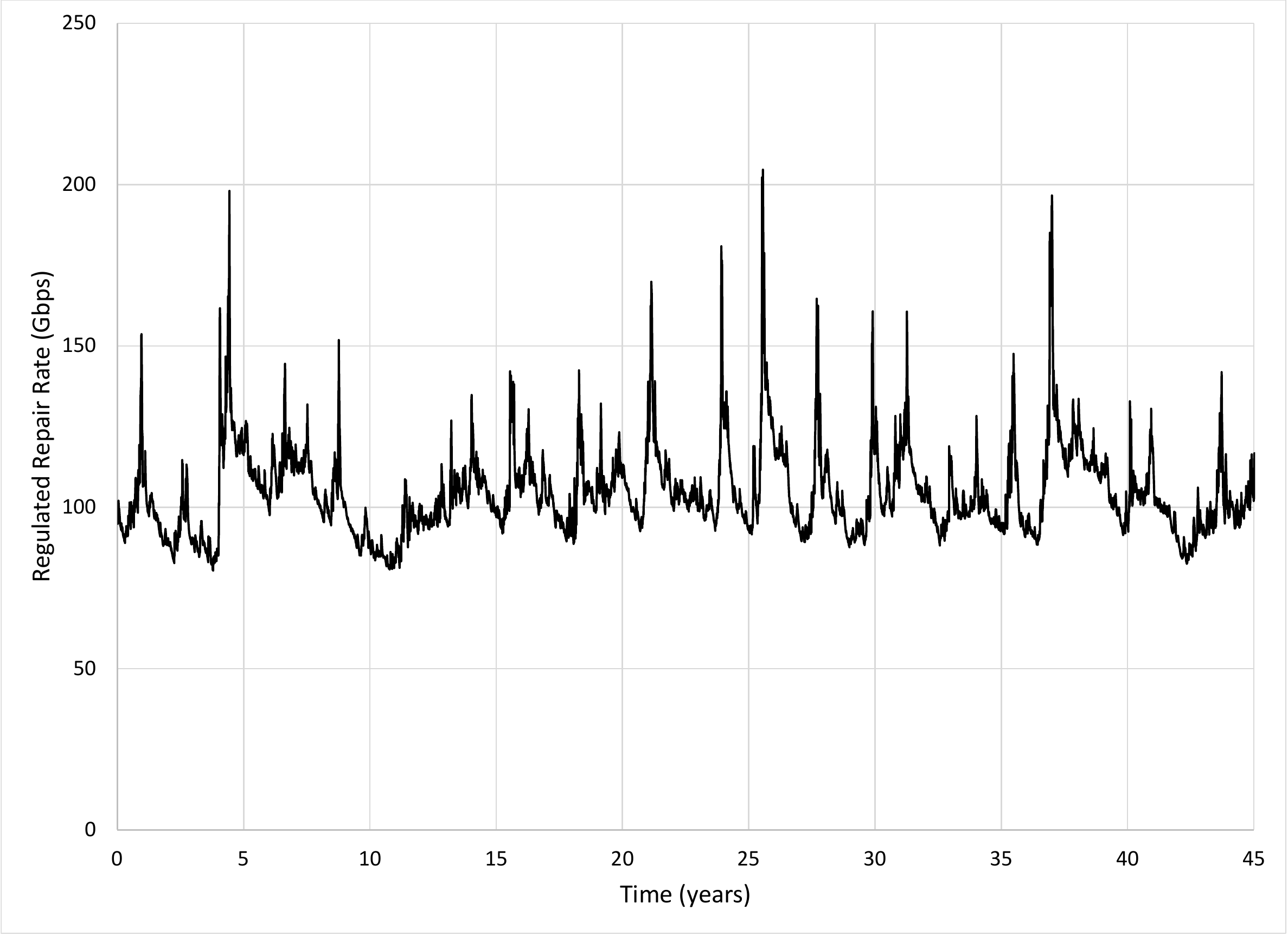}
\caption{Repair regulator for a $(402,268,134)$ \liqsystem\ with fixed $1/\lambdap=3$ years.}
\label{fig:rr_fixed_graph}
\end{figure}

\begin{table}[H]
\resizebox{\textwidth}{!}{
\centering
\begin{tabu}{|c|c|c|c||c|c|c|c|c||c|c|}
\hline
\multicolumn{4}{|c||}{System parameters} &
\multicolumn{5}{|c||}{Regulated \rrepairrate} &
\multicolumn{2}{|c|}{Fixed \rrepairrate} \\
\hline
$n$	& $r$ & $\beta$ & $S$ & 
$\Ravg$	& $\Rnnn$ & $\Rnnnnn$ & $\Rpeak$ & \mttl\ & 
$\Rpeak$ & \mttl\
\\
\hline
402	& 134 &  33.3\% & 1 PB & 
106 Gbps & 154 Gbps & 226 Gbps & 311 Gbps & $\gg \expp{9}$ years & 
104	Gbps & $3.2 \cdot \expp{7}$ years
\\ \hline
402	& 67 &  16.7\% & 1 PB &
298 Gbps	& 513 Gbps & 975 Gbps & 1183 Gbps & $\gg \expp{9}$ years &
394	Gbps & $9.0 \cdot \expp{7}$ years
\\ \hline
\end{tabu}
}
\caption{Comparison of regulated \rrepairrate\ versus fixed \rrepairrate.}
\label{tab_regualator_examples}
\end{table}

Fig.~\ref{fig:rr_spiky_graph_3x} shows the behavior of regulated
\rrepairrate\ when the \nodef\ rate is varied as described in
Section~\ref{simulations varying sec}.  The upper portion of
Fig.~\ref{fig:rr_spiky_graph_3x} displays the average \nodef\ rate
varying over time, and the lower portion of
Fig.~\ref{fig:rr_spiky_graph_3x} displays the corresponding
\rrepairrate\ chosen by the regulator.  This is a stress test of the
regulator, which appropriately adjusts the \rrepairrate\ to the quickly
changing \nodef\ rate to avoid \objloss es while at the same time using
an appropriate amount of bandwidth for the current \nodef\ rate.  There
were no object losses in $\expp{9}$ years of simulation using the
regulator running in the conditions depicted in
Fig.~\ref{fig:rr_spiky_graph_3x}, and based on estimates using the
results from Appendix~\ref{liquid mttl analysis sec} the \mttl\ is
orders of magnitude larger than $\expp{9}$ years.  The results from
Appendix~\ref{liquid mttl analysis sec} show that the regulator is
effective at achieving a large \mttl\ when the \nodef\ rate varies even
more drastically.

A repair adjustment algorithm could additionally incorporate more
advanced aspects.  For example, since the bandwidth consumed by the
\rrepairrate\ is a shared resource, the \rrepairrate\ may also be
adjusted according to other bandwidth usage, e.g., the \rrepairrate\ can
automatically adjust to higher values when object access or storage
rates are low, and conversely adjust to lower values when object access
or storage rates are high.  Furthermore, the \rrepairrate\ could be
regulated to take into account known future events, such as scheduled
upgrades or node replacements.

\begin{figure}
\centering
\includegraphics[scale=0.5]{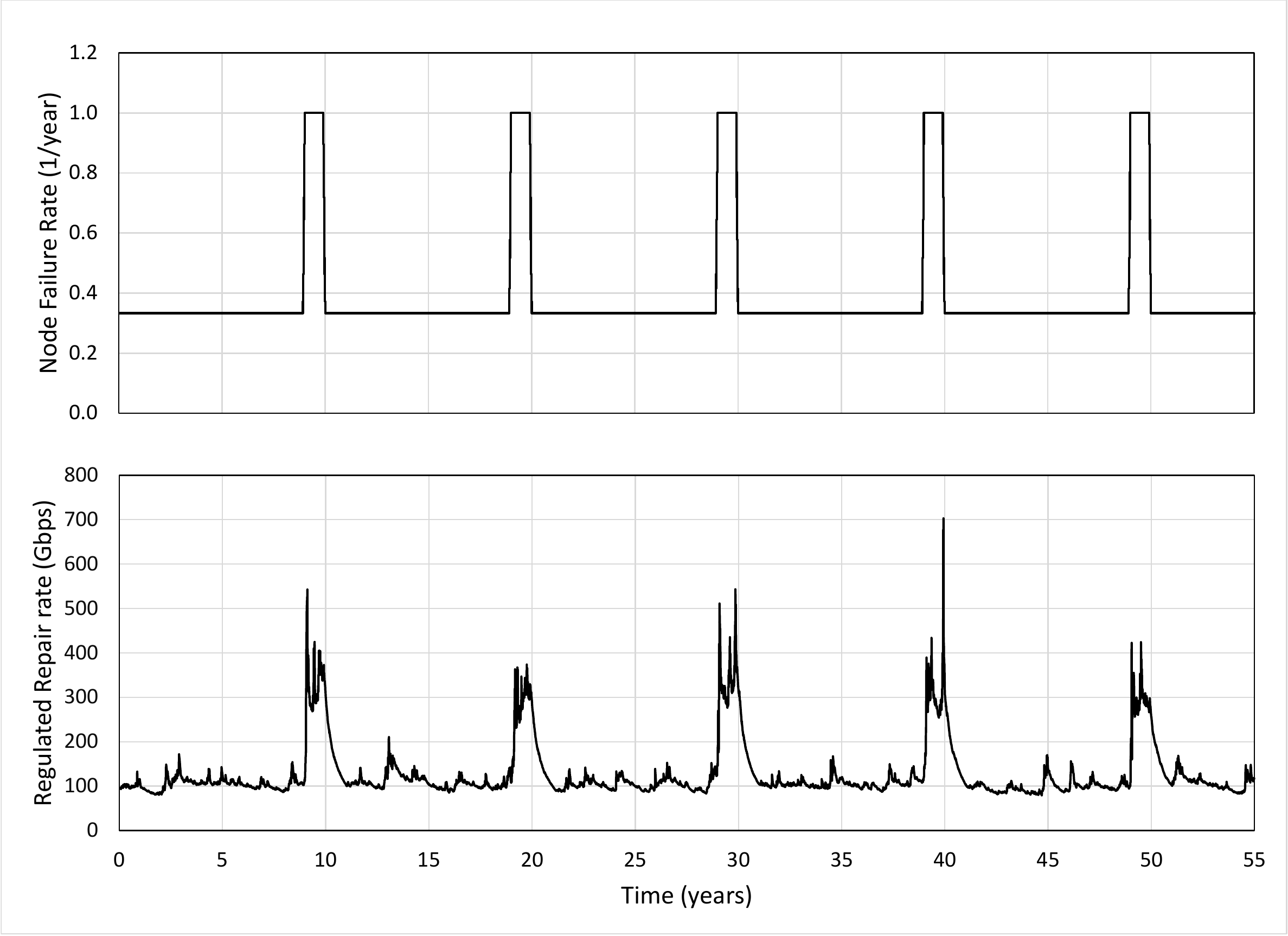}
\caption{Regulated repair for a $(402,268,134)$ \liqsystem\ with \nodef\ rate varying by a factor of three}
\label{fig:rr_spiky_graph_3x}
\end{figure}

% Missing: maxRepair description.

% Properties:
%  - fairly smooth when the failure rate is constant
%    (opposed to small code)
%  - aggressive when it needs to be
%  - maintain a low \rrepairrate\ on average, and in the peak.

% Latent failure section.

\section{Conclusions}

We introduce a new comprehensive approach to distributed storage,
\liqsystem s, which enable flexible and essentially optimal combinations
of storage reliability,  \storeoverhead, repair bandwidth usage, and
access performance.  The key ingredients of a \liqsystem\ are a low
complexity \liqcode, a \liqorg\ and a \liqrepair\ strategy. The repair regulator design we present provides further robustness to \liqsystem s against varying and/or unexpected \nodef . Our repair and access simulations establish that a
\liqsystem\ significantly exceeds performance of \tradsystem s in all
dimensions, and allows superior operating and performance trade-offs
based on specific storage deployment requirements. 

We address the practical aspects to take into consideration while implementing a \liqsystem\ and provide an example architecture. A \liqsystem\ eliminates network and compute hot-spots and the need for urgent repair of failed infrastructure. 

While not detailed in the paper, we believe a \liqsystem\ provides superior geo-distribution flexibility and is applicable to all kinds of Object, File, Scale-out, Hyper-converged, and HDD/SSD architectures/use-cases. There are also optimizations that the team has worked on that involve pre-generation of EFIs and result in better storage efficiencies at the storage nodes.

\section*{Acknowledgment}
Rick Dynarski, Jim Kelleman, and Todd Mizenko developed the prototype software for the
\liqsystem\ described in Section~\ref{implementation sec}, which provides a basic validation 
of the overall \liqsystem\ properties presented herein.  This prototype software
is used to cross validate the simulation results presented in Section~\ref{repair sim sec}.

Christian Foisy helped to produce the RaptorQ performance
results presented in Section~\ref{raptorq sec}, and also helped cross verify
the prototype software for the \liqsystem.  Menucher Menuchehry
and Valerie Wentworth provided high level comments
that greatly improved the presentation.

We would like to thank Rick,  Jim, Todd, Christian, Menucher, and Valerie for their valuable contributions.

% or

\appendix

  
\section{\mttl\ analysis of \liqsystem}
\label{liquid mttl analysis sec}
In this section we analyze regulated repair of the \liqsystem~and derive
some bounds on \mttl.  In the case of constant repair rate the bound is
particularly tight.  Although we typically assume that the node failure
process is a time invariant Poisson process the analysis can be applied
to cases where the rate is not known and time varying.  We assume
throughout that the number of nodes $M$ is the same as the code length
$n,$ and we use $n$ to denote this number.

The framework for the analysis ignores the possibility that an object
may be unrepairable.  In effect we suppose that as each object leaves
the repair queue (because it reaches the head of the queue and,
normally, is repaired) a new object enters at the tail.  This notion is
an analytical convenience that takes advantage of the relative ease of
determining or estimating the asymptotic steady state distributions of
this system.  We assume a large number of objects $\nObjects,$ and for
some of the analysis we ignore effects of a finite number of objects and
consider the infinite object limit. This has no significant effect on
the results but simplifies the analysis considerably.

For the analysis we generally assume that the repair rate is updated
upon object repair.  In the limit of infinitely many objects this
amounts to continuous update.  With a finite number of objects the queue
state, as sampled at object repair times, the system forms a Markov
chain with a unique invariant distribution.  In the case where the node
failure rate is known it is a finite state chain.  Of particular
interest is the marginal distribution of the number of erased fragments
for the object at the head of the queue.  In the case of fixed repair
rate we can express this distribution in closed form, and for the
regulated case we will derive a bound on the cumulative distribution.

\subsection{\mttl{} Estimates}

We define the \mttl{} as the expected time given a perfect initial queue
state (all objects have zero erased fragments) until an object achieves
$r+1$ missing fragments prior to completing repair.  For the system with
node failure rate estimation, we also need to specify the initialization
of the estimator.  The system is typically designed to achieve a very
large \mttl.  We will show in the constant repair rate case that we can
obtain a sharp closed-form lower bound on the \mttl.  The bound can be
applied even in the case of a mismatch between an assumed node failure
rate and an actual node failure rate.  In the regulated repair case the
analysis is more difficult, but we obtain an efficiently computable
generalization of the lower bound for the fixed node failure arrival
rate case.  We further extend the analysis to cases where the node
failure rate is estimated based on the node failure realization.

Between node failure events the system evolves in a purely deterministic
fashion; objects simply advance in the repair queue according to the
repair rate policy.  Under a time invariant Poisson model the node
failure interarrival times are independent exponentially distributed
random variables and the selection of the failing node is uniformly
random over all of the nodes.  Thus, the repair can be viewed as a
random process consisting of a sequence of samples of the queue state
taken upon node failure.  The process is generally Markov where the
state space consists of a pair of queue states, that just prior to the
node failure and that after, and, when there is failure rate estimation,
a state representation of the failure rate estimator.  The transition to
the next state has two components.  First, there is the time until the
next node failure, which is an exponential random variable.  Given that
time the queue state advances deterministically from the state after the
current node failure to the state just prior to the next.  The second
component is the selection of the failing node.   That choice determines
which objects in the queue experience additional fragment loss and which
do not.  The update of the state of the failure rate estimators can also
occur at this point.  Because of the nested nature of the process there
will be a point $x$ in the queue such that objects in positions less
than $x$ experience fragment loss and those more than $x$ do not,
because they were already missing the fragment associated to the failing
node.  Note that the queue remains monotonic: the set of erased
fragments for an object in position $x$ is a subset of the erased
fragments for an object in position $y$ whenever $x\le y.$ In the case
of node failure rate estimation, those estimates will be updated upon
node failure.

\subsubsection{The Constant Repair Rate Case}

When the repair rate is held constant the deterministic update portion
of the process is particularly simple.  Let $q$ denote the distribution
of the number of erased fragments for the object at the head of the
repair queue immediately prior to node failure.  Because node failure is
a time invariant Poisson process the distribution $q$ is also the
distribution of the number of erased fragments for the object at the
head of the repair queue at an arbitrary time.  In other words, $q$ is
the the distribution of the number of erased fragments upon repair under
continuous (steady state) operation.  We therefore have 
\[
q(s) = B(n,s,e^{-\lambdap T}).
\]
Let $F(t)$ denote the number of erased fragments for the object at the
head of the repair queue at time $t$. If $t$ is a node failure time let
$F(t-)$ and $F(t+)$ denote value of $F(t)$ immediately prior to and
after the node failure.  The probability of simultaneous node failures
is $0$ so we assume $F(t-) \le F(t+) \le F(t-)+1$. Let $t$ be a node
failure time, then we have
\[
\Prob{F(t+) = r+1 \wedge F(t-) = r} = \frac{n-r}{n} q(r)\,. 
\]
Note that this event $F(t+) = r+1 \wedge F(t-) = r$ corresponds to a data
loss event and that the first time this occurs, starting from the
perfect queue state, is the first data loss event.  We will refer to
this event as a {\em data loss transition}.  Thus, the \mttl{} is by
definition the expected time to the first data loss transition starting
from a perfect queue state.  Applying Kac's theorem to the Markov chain
we have $(\frac{n-r}{n} \cdot q(r))^{-1}$ is the expected number of node
failure events between data loss transitions under steady state
operation.  The expected {\em time} between such events is given by
$(\lambdap \cdot (n-r) \cdot q(r))^{-1}$ since $\lambdap \cdot n$ is the rate
of node failure.  To obtain a rigorous bound on \mttl{} some further
analysis is needed because of differing assumptions on the initial
condition.  The rigorous correction is negligible in regimes of interest
so we use the estimate
\begin{equation}\label{eqn:MTTDLestimate}
\frac{1}{\lambdap \cdot (n-r) \cdot q(r)} \lesssim \mttl
\end{equation}
as a basis for choosing repair rates ($T$) as a function of \mttl.

Starting at time $0$ let $T_S$ denote the expected value of the first
data loss transition time where the initial queue state is distributed
according to the the steady state distribution conditioned on $F(0) \le
r.$ Let us also introduce the notation $q(>r) = \sum_{s > r} q(s).$

\begin{lemma}\label{lem:MTTDLsandwich}
Assume a constant repair rate with system repair time $T.$ Then we have
the following bound
\begin{align*}
 \frac{1}{\lambdap \cdot (n-r) \cdot q(r)} - \frac{T}{1-q(>r)} 
 \le
   \mttl
%    \le  \frac{1}{\lambdap k  B(n,r,e^{-\lambdap T}))} + \frac{1}{1-q_{>r}} T\,.
\end{align*}
\end{lemma}
\begin{proof}
We prove the Lemma by establishing the following two inequalities
\begin{align*}
T_S \le \mttl &\le \frac{1}{1-q(>r)} T+ T_S
\\
 \frac{1}{\lambdap \cdot (n-r) \cdot q(r)} & \le \frac{1}{1-q(>r)} T + T_S\,.
\end{align*}

Consider a fixed node failure process realization.  Among all initial
queue conditions with $F(0)\le r$ the perfect queue will have a maximal
time to data loss.
%For any initial queue state with at most $r$ erased fragments at the head of the queue the first data loss event can not occur later
%than if the queue had been perfect.
%a data loss transition is upper bounded by the expected time from the perfect queue state.
Thus, $T_S \le \mttl.$

Assume the system is initialized at time $0$ with an arbitrary initial
queue state and let us consider the expected time until the next data
loss transition.  Note that $k \cdot T$ is time when the object initially at
the tail of the queue is repaired for the $k$-th time.  Let $k^*$ denote
the smallest $k>0$ where upon repair that object has at most $r$ erased
fragments.  Then the queue state at time $t={k^*} \cdot T$ is the steady state
distribution conditioned on the object at the head of the queue having
at most $r$ erased fragments, i.e., conditioned on $F(t) \le r.$ Hence,
the expected time to the next data loss transition can be upper bounded
by  $\expectation ({k^*} \cdot T) + T_S.$ Now, the probability that $k^* = k$
is $(1-q(>r)) \cdot q(>r)^{k-1}$ hence
\(
  \expectation (k^* \cdot T) \le \sum_{k>1} (1-q(>r)) \cdot q(>r)^{k-1}
  \cdot k \cdot T = \frac{T}{1-q(>r)}\,.
\)
It follows that $\mttl \le \frac{T}{1-q(>r)} + T_S$ and that
the expected time between data loss transitions, which is
$(\lambda \cdot (n-r) \cdot q(r))^{-1},$ is also upper bounded by
$\frac{T}{1-q(>r)} + T_S.$
\end{proof}
%Once the event $E$ has occured it will with either recur within time $T$  
%We expect that the bound will be tight is the probability of a recurrance of $E$ within time $T$ is small.

\subsection{Regulated Repair Rate}

For the regulated repair rate case the memory introduced into the system
by the selection of the maximum requested repair rate makes it quite
difficult to analyze.  In particular the rate of advancement of the
repair queue depends on the entire state of the queue, so there is no
convenient way to reduce the size of the Markov chain state space, as
was effectively done in the fixed repair rate case.  The addition of
node failure rate estimation further complicates the problem.  It turns
out, however, that by approximating queue advancement with a greedy form
of operation we can obtain a related lower bound on the \mttl.  When
applied to the constant repair rate case the bound is smaller than the
one above by a factor of $(1-f_\rmtar).$

Although we shall consider the case of known node failure arrival rate,
much of the analysis is common with the case including node failure
estimation.  Thus, we will first develop some general results that admit
node failure rate estimation.

We introduce the notation $\FDis(s,t)\,$ to denote the number of {\em
distinct} node failures in the interval $(s,t).$  As before, if $s$ is a
node failure time we use $s-$ and $s+$ to denote times infinitisemally
before or after the node failure respectively.  For convenience we
assume that $\phi(f,x)$ is bounded below by a positive constant, is
defined for all $f$ and $x$ and is non-increasing in $f$ and
non-decreasing in $x.$ For computation we may also assume that
$\phi(f,x)$ is bounded above by the nominal value $\lambdap \cdot T$ so that,
in the known node failure rate case, $T$ is an upper bound on the time
between repairs for a given object.  Recall that since $\lambdap \cdot T =
\phi(0,0)$ we always have $\inf_{x\in [0,1]}\phi(f(x),x) \le \lambdap \cdot T$
so this restriction does not affect the operation of the system.  In the
case with node failure rate estimation we will typically assume some
minimum postive value for the estimated node failure rate so that there
will be some finite time $T$ during which a complete system repair is
guaranteed.

When we admit failure rate estimation we suppose that each object has
it's own failure rate estimate.  That estimate $\hat{\lambda}$ figures
in the requested repair rate through the implied requested system repair
time of $\hat{\lambda}^{-1} \cdot \phi(f,x).$  For an object that entered the
queue at time $s$ we write its node failure rate estimate at time $t$ as
$ \hat{\lambda}(s,t).$ The estimate is allowed to depend on past node
failures.  It can depend on $s$ only through the partition of node
failures into `past' and `future' implied by time $s.$ Consequently, all
objects that enter the queue between the same pair of successive node
failures use the same value for $\hat{\lambda}.$ In other words, as a
function of $s,$ the estimate $\hat{\lambda}(s,t)$ is piecewise constant
in $s$ with discontinuity points only at node failure times.  We remark
that this can be relaxed to having $\hat{\lambda}(s,t)$ non-decreasing
in $s$ between node failure points.  In practice the node failure rate
estimate used will generally be a function of relatively recent node
failure history which can be locally stored.  Note that the estimator
for an object may depend on its number of erased fragments.  

We call an object in the repair queue a {\em critical} object if it was
at the tail of the queue at a node failure point.  Letting $s(x,t)$
denote the queue entry time of the object in position $x$ of the queue
at time $t,$ an object is critical if $s(x,t)$ is a node failure time.
To be more precise, a critical object is an object that was the last
object completely repaired prior to a node failure.   Such an object has
an erased fragment while it is still at the tail of the queue.

Let $\chi(t)$ denote the set of critical object queue locations $x \in
[0,1]$ at time $t.$ Let $f(x,t) =\FDis(s(x,t),t)/n$ and let
$\hat{\lambda}_x(t)$ denote $\hat{\lambda}(s(x,t),t).$ Under the current
assumptions the repair regulator need only track the critical objects in
the system.
\begin{lemma}
The repair rate is always determined by a critical object, i.e., 
\[
\inf_{x\in[0,1]} \hat{\lambda}^{-1}_x(t) \cdot \phi(f(x,t),x)
=
\min_{x\in \chi(t)} \hat{\lambda}^{-1}_x(t) \cdot \phi(f(x,t),x)\,.
\]
\end{lemma}
\begin{proof}
Of all objects that entered the queue between the same pair of
successive node failures, the critical object has the minimal value of
$x$ and all such objects have the same value of $\hat{\lambda}$ and $f.$
Since $\phi(f,x)$ is non-decreasing in $x$ the result follows.
\end{proof}

%Thus
%This assumption implies that Lemma \ref{lem:regbnd} holds and the greedy analysis still applies.
%Indeed, the equation governing $y(s,t)$ is now given by
%\[
%\frac{d}{dt} y(s,t) =  \hat{\lambda}(s,t)  ( \phi(\FDis(s,t),y(s,t)))^{-1}
%\]
%and the key inequality in the proof of Lemma \ref{lem:regbnd} can be written
%\begin{align*}
% \frac{d}{dt} (x(s,t)-y(s,t)) &= \bigl(\inf_{z\in [0,1)}   \hat{\lambda}^{-1} (s(z),t) \phi(f(z),z)\bigr)^{-1} - \bigl(   \hat{\lambda}^{-1}(s,t) \phi(\FDis(s,t),y(s,t))\bigr)^{-1}
%\\
%& \ge   \hat{\lambda}^{-1}(s,t) \Bigl( \bigl( \phi(\FDis(s,t),x(s,t))\bigr)^{-1} - \bigl( \phi(\FDis(s,t),y(s,t))\bigr)^{-1}\Bigr)
%\end{align*}
%where $ \hat{\lambda}(s(z),t)$ denotes the failure rate arrival estimate for the object in position $z$ in the queue.

To obtain a bound on the \mttl{} we will adopt a greedy approximation in
which we effectively assume that objects control the repair rate during
their time in the repair queue.  To this end, for any time $s$ set
$y(s,s)=0$ and for $t\ge s$ let us define $y(s,t)$ as the solution to
\[
\frac{d}{dt} y(s,t) = \hat{\lambda}(s,t) \cdot (\phi(\FDis(s,t),y(s,t)))^{-1}
\]
which is well defined as long as $y(s,t) \le 1.$  Further, let $\tau(s)$
be the minimal $t$ such that $y(s,t)=1.$ An interpretation of $y$ is
that it is the position in the queue of an object that was at the tail
of the queue at time $s$ and is designated to be the object that
determines the system repair rate as long as it remains in the queue.
Let $x(s,t)$ denote the actual position of the object at time $t$ given
that it is at the tail of the queue at time $s.$
\begin{lemma}\label{lem:regbnd}
For any $s$ we have $x(s,t) \ge y(s,t).$
\end{lemma}
\begin{proof}
For each fixed $f$ the function $\phi(f,z)$ is non-increasing in $z \in [0,1].$
Thus if $0 \le x \le y \le 1$ then
\(
 (\phi(f,x))^{-1} - (\phi(f,y))^{-1} \ge 0\,.
 \)
Let us fix $s$ and assume $0 \le y(s,t) , x(s,t) \le 1.$  Then we have
\begin{align*}
  \frac{d}{dt} (x(s,t)-y(s,t))
    &= \bigl(\inf_{z\in [0,1)} \hat{\lambda}^{-1}_z(t) \cdot \phi(f(z,t),z)\bigr)^{-1}
     - \bigl(\hat{\lambda}^{-1}(s,t) \cdot \phi(\FDis(s,t),y(s,t))\bigr)^{-1} \\
    &\ge \hat{\lambda}^{-1}(s,t) \Bigl(\bigl(\phi(\FDis(s,t),x(s,t))\bigr)^{-1}
           - \bigl(\phi(\FDis(s,t),y(s,t))\bigr)^{-1}\Bigr)
\end{align*}
Thus, $x(s,t) - y(s,t) < 0$ implies $ \frac{d}{dt} (x(s,t)-y(s,t)) \ge
0$ and since $x(s,s) = y(s,s) = 0$ we see that $x(t,s) - y(t,s) < 0$
cannot occur by the mean value theorem.
\end{proof}

For data loss to occur there must be node failure times $s\le t$ such
that both $\FDis(s-,t+) \ge r+1$ and 
\(
 x(s-,t)<1.
 \footnote{Here we may interpret $x(s-,t)$ as the limit of $x(s',t)$ as
 $s'$ approaches $s$ from below.}
\)
In this event data loss has occured to the critical object that was at
the tail of the queue immediately prior to time $s.$ By Lemma
\ref{lem:regbnd} this implies that $\FDis(s-,\tau(s)) \ge r+1.$ The
condition $\FDis(s-,\tau(s)) \ge r+1$ denotes a {\em greedy data loss}
where the object arriving at the tail of the queue at time $s$ greedily
controls the repair rate and suffers data loss.  We refer to the node
failure at time $s$ as a greedy data loss node failure.  Note that the
greedy data loss actually occurs after time $s.$
%Further, note that while $\hat{\lambda}(s,s)$ is determined by node failures up time $s,$ given that value the greedy object queue process
%is dependent only on subsequent node failures.

\subsubsection{The Regulated Repair Case with Fixed Failure Rate}

We now assume that the estimated repair rate $\hat{\lambda}$ is constant
and may or may not be equal to the actual node failure rate $\lambdap.$
In this case there is no ambiguity in the intialization of
$\hat{\lambda}$ and the \mttl{} can be lower bounded by the expected
value of $s$ where $s$ is the time of the first greedy data loss node
failure.

Assume $s$ is a node failure time and let $q$ now denote the
distribution of $\FDis(s-,\tau(s-)).$ By Kac's theorem,  $1/q(>r)$ is
the expected number of node failures between greedy data loss node
failures.  Hence $1/(\lambdap \cdot n \cdot q(>r))$ is the expected time between
greedy data loss node failures.

Let $T_G$ denote the expected time time until the first greedy data loss
node failure assuming an unconditioned future node failure process.
Note that the time to a greedy data loss is independent of the initial
queue state.  It follows that $1/(\lambdap \cdot n \cdot q(>r))$ is upper bounded by
$T+T_G.$ This is because if we condition on $0$ being a greedy data loss
node failure time, then the node failure process after time $T$ is
independent of this condition.  Recalling that 
\(
T_G \le \mttl
\)
we obtain
\[
\frac{1}{\lambdap \cdot n \cdot q(>r)} - T \le \mttl
\]

In regimes of interest we will have  $q(>r) \simeq q(r+1)$ and
assuming a fixed repair rate we have
\(
 n \cdot q(r+1) = n \cdot B(n-1,r,e^{-\lambdap \cdot T}) = e^{\lambdap \cdot T}
 \cdot (n-r) \cdot B(n ,r,e^{-\lambdap \cdot T})\,.
\)
If $\hat{\lambda} = \lambdap$ then we see that the resulting \mttl{} bound
is approximately a factor $e^{-\lambdap \cdot T} = 1-f_\rmtar$ worse than the
previous one.

Given $\phi$ it is not difficult to quantize the system to compute the
distribution $q.$ To perform the computation we consider quantizing $x$
with spacing $\delta x.$ \footnote{An analog of Lemma \ref{lem:regbnd}
can be obtained for the quantized system with the addition of a small
correction term, but the impact on the results is negligible for
reasonable size systems.} This is essentially equivalent to assuming a
finite number, $\nObjects = \delta^{-1},$ of objects and we describe the
quantization from that perspective.  Assume that repair rates are
updated upon object repair.  Under regulated repair the repair system
then selects the next object repair time, which is now given by the
minimum of $\delta \cdot \phi(f(j\delta),j \delta)$ for $j=0,\ldots,\nObjects - 1.$ Assume that the next repair time $\Delta$ is given and that we are
given that the object in position $k\cdot\delta$ has a random number of
erased fragments $F$ that is distributed according to $p_k(F).$ Then the
distribution of the number of erased fragments for the object when it
reaches position $(k+1)\cdot\delta$ is given by
\[
p_{k+1}(F) = \sum_{0\le F' \le F} p_k(F') \cdot B(n-F',F-F',e^{-\lambdap
\cdot \Delta})\,.
\]
To compute the distribution $q$ we proceed similarly but, in effect,
track the greedy controlling object as it moves through the repair
queue, assuming that that object is the one determining the repair rate.
Letting $q_k(F)$ denote the distribution of the number of missing
fragments when the object is (greedily) in position $k\delta.$ We have
\[
q_{k+1}(F) = \sum_{0\le F' \le F} q_k(F') \cdot B(n-F',F-F',e^{-\delta
\cdot \phi(F'/n,k \delta)})\,.
\]
Note that the distribution $q$ is simply $q_{\nObjects}(\cdot)$
assuming we initialize with $q_0(F) = \indicator_{F-1}$ (the
distribution with unit mass at $F=1.)$ Using the above expression we see
that the distribution $q$ can be computed with complexity that scales as
$n \cdot \nObjects.$

Note that if the above computation we instead use the initial condition
$\indicator_{F}$ then the resulting distribution $q$ provides an upper
bound on the steady state distribution of the number of erased fragments
upon repair.  More precisely, the resulting $q(>s)$ upper bounds the
probability that more than $s$ fragments of an object are erased when
that object is repaired.  Using this value of $q$ we obtain
$\sum_{s=0}^n q(>s)$ as an upper bound on the expected number of
repaired fragments, which represents the repair efficiency of the
system.

In Fig. \ref{fig:knownBounds} we plot the computed cumulative distribution of the number of missing fragments upon repair
and the computed upper bound.  The five shown bounds are different only in the limit placed on the maximum repair rate which
is indicated as a factor above the nominal rate.  The ``1x'' case corresponds to the fixed repair rate policy.  Already by allowing
a ``2x'' limit a dramatic improvement is observed.  It is observed that the greedy bound tracks the true distribution, as found through simulation
for the ``3x'' limited case, with some shift in the number of missing fragments.  Although somewhat conservative, it is
evident that the computed bound can provide an effective design tool and that rate regulation can provide dramatic improvements in MTTDL.

\begin{figure}
\centering
\includegraphics[scale=0.5]{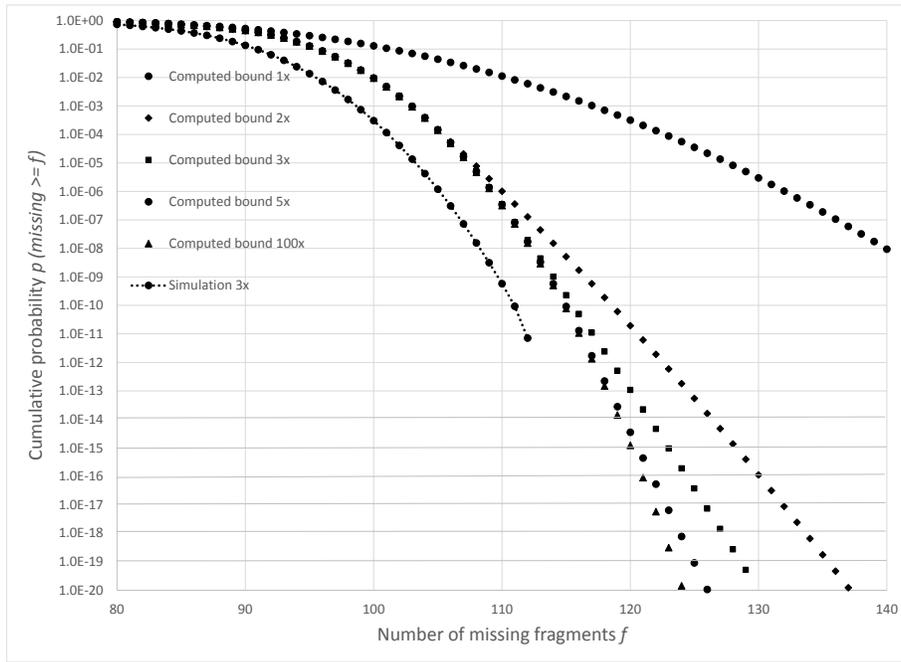}
\caption{Comparison of simulated and computed results for regulation with known node failure arrival rate.  Different curves place
different limits on the peak repair rate, indicated as factors above the nominal rate}
\label{fig:knownBounds}
\end{figure}

Fig. \ref{fig:knownBoundsvary} shows how the cumulative distribution of the number of missing fragments changes with node failure rate mismatch.
We set the repair rate limit to 10 times the nominal rate in increased the node failure rate by factors of 1.25, 1.5, and 2.0 respectively.
The regulator responds by increasing the average repair rate and reaching a new asymptotic queue state equilibrium with respectively 
larger repair targets.  This is indicated in the shifts of the cumulative distribution.
For very large shifts in failure rate this level of adaptation would be inadequate to ensure data durability.
Fortunately node failure rate estimation is easily incorporated into the regulator design, leading to a significantly more robust system.

\begin{figure}
\centering
\includegraphics[scale=0.5]{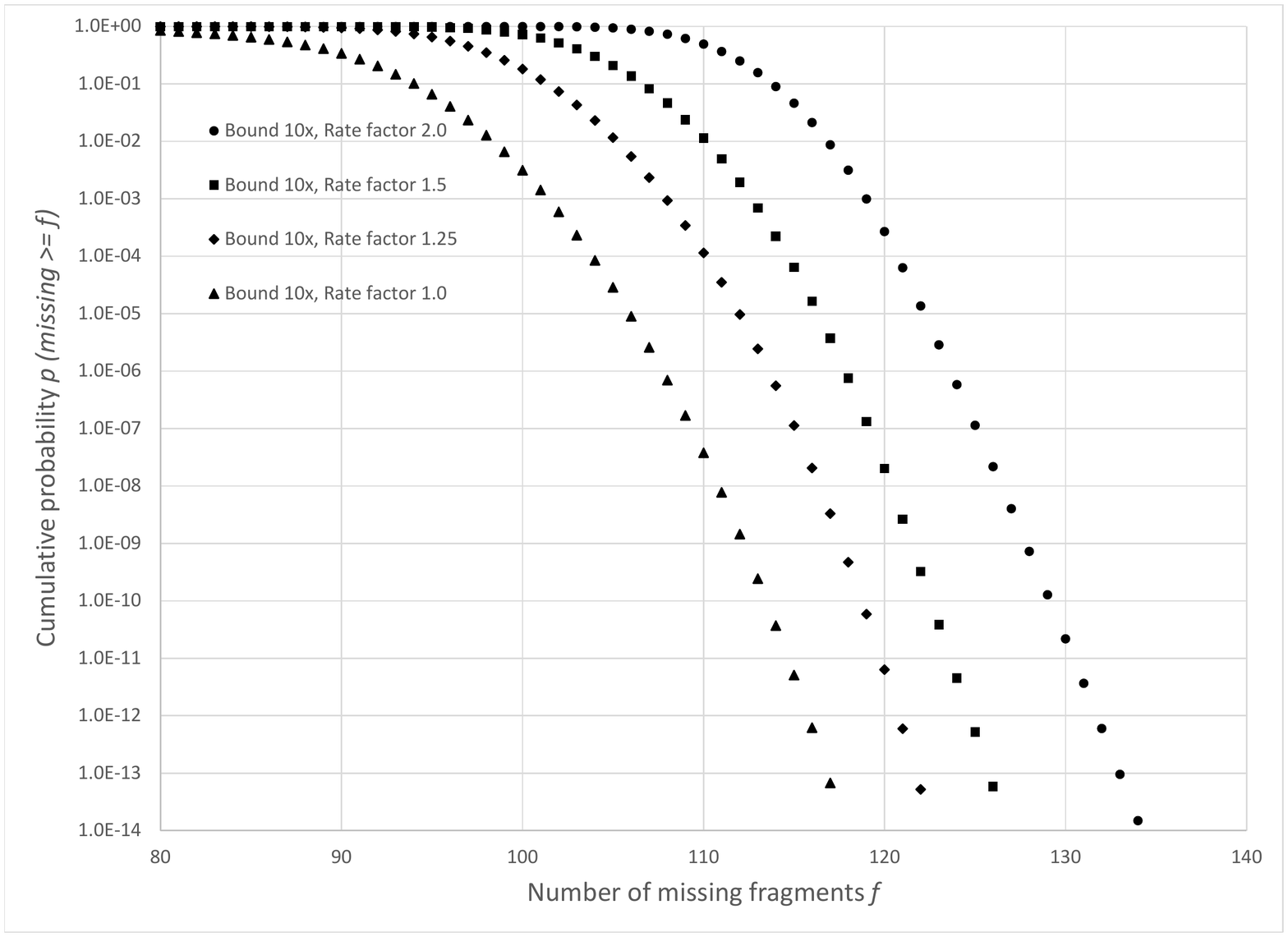}
\caption{Cumulative distribution bounds with mismatch between assumed design node failure rate and actual node failure rate.}
\label{fig:knownBoundsvary}
\end{figure}

%We consider $\phi$ designed according to equation \ref{eqn:sepGovForm} except that its maximum value is limited to $\phi_{\text{nom}}$ and the its
%minimum to $A \phi_{\text{nom}}.$
%[[TO DO: ESTIMATED CASE]]

\subsubsection{The Regulated Repair Case with Failure Rate Estimation}

The rate regulator as described above is designed to concentrate the
repair queue near its nominal behavior when the actual  node failure
rate matches the design's assumed failure rate.  As demonstrated, the
regulator can compensate for some variation  in the failure rate.  If
the actual node failure rate is persistently different from the design
assumed rate then the regulator typical behavior is shifted relative to
the nominal one.  By using an estimator of the node failure rate the
system can bring the typical behavior back toward the nominal one.
Moreover, by using an estimate of the node failure rate the system can
adapt to much larger variations in the actual node failure rate.

Many forms of estimators may be considered.  We shall consider some
relatively simple forms that are amenable to computation analysis
similar to the fixed arrival rate case.  These estimates are based on
first order filtering of some function of the node failure interarrival
times.  A typical update of the node failure interarrival time will
take the form
\[
\xi(\hat{T}) \leftarrow \alpha \cdot \xi(\hat{T}) + \bar{\alpha} \cdot \xi({T})\,.
\]
where $T$ is a most recent node failure interarrival time and the
functional $\xi$ is an invertible function.  A convenient case to
consider is $\xi(T) = T$ but other potentially interesting forms are
$\xi(T) = \log T$ or $\xi (T) = 1/T,$ or regularized versions of these
functions such as $\xi(T)=\log (1+T).$ For simplicity we will focus on
the case $\xi(T) = T.$ The node failure rate estimate $\hat{\lambda}$ is
simply given as $\hat{T}^{-1}.$ 

%We will use the assumption that node failure rate estimates are updated only at node failure times.
%We could equivalently require that estimates are independent of the time since the last node failure.
%This assumption is made primarily to simplify the analysis and to simplify computation of bounds.
Corresponding to our assumptions on failure rate estimators,  we assume
that the value of $\alpha$ used by an object is a function only of the
number of erased fragments for that object and the total number of node
failures that have occurred since that object entered the repair queue,
for which we introduce the notation $\FAll(s,t).$ Let us now consider
how to compute the distribution of the greedy behavior.  In this case we
take a slightly different approach then previously.  We aim to
recursively compute the joint greedy queue position and failure rate
estimate distribution at node failure times.  Since a node failure may
or may not result in a fragment erasure for the greedy object, we must
also jointly consider the distribution of the number of erased
fragments.

Let is fix $s$ and consider $F(t)=\FDis(s,t), F^+(t)=\FAll(s,t),
\hat{\lambda}(s,t), y(s,t),$ for $t\le \tau(s)$  and extend these
definitions for $t \ge \tau(s).$ Specifically, for $t \ge \tau(s)$ set
$y(s,t)=1,$  set $F(t) = F(\tau(s)),$ set
$\hat{\lambda}(s,t)=\hat{\lambda}(s,\tau(s)).$ We will leave
$F^+(t)=\FAll(s,t)$ unchanged, i.e.  $F^+(t)$ continues to count node
failures.  In short we track $F(t),y(s,t),$ and $\hat{\lambda}(s,t)$
until $y(s,t)=1$ which then marks an absorbing state.

We will now briefly discuss how, give  $F^+,$ we can computationally
update a joint distribution of $F,y,\hat{\lambda}$ to the next node
failure point.  Assume that at time $t_0$ the greedy object has an
estimate $ \hat{\lambda}(s,t_0)$ and greedy queue position $y=y(s,t_0)$
and a certain value of $F.$  Then the joint distribution of $y$ and
$\hat{\lambda}(s,t)$ at the next node failure point lies on a curve
parameterized by $t_1,$ the next node failure time.  Note that since we
assume that no additional node failures occur, we have that
$\hat{\lambda}(s,t), F(t),$ and $F^+(t)$ are all constant and we can
drop the dependence on $s$ and $t.$ Note also that $\alpha$ is detemined
and is independent of the next node failure time.  For $t \ge t_0$ such
that $y(s,t) < 1$ we have 
\[
y(s,t) =  y(s,t_0) +  \hat{\lambda} \int_{t_0}^{t} \phi^{-1}(F/n,u) du\,.
\]
Assuming $t=t_1$ is the next node failure the update of $\hat{\lambda}$
is given by
\[
\hat{T}(s,t) = \alpha \cdot \hat{T}(s,t_0) + \bar{\alpha}(t-t_0)
\]
and we obtain the curve $y(s,t_1),\hat{T}(s,t_1)$ for $y(s,t_1)<1.$
Note, however, that at $y(s,t_1)=1$ we actually have the discontinuity 
 \(
\hat{T}(s,t_1) = \hat{T}(s,t_0) 
\)
since in this case the object reached the head of the queue prior to
additional node failures.  To obtain the probability density on this
curve observe that until the time of the next node failure $t_1$ we have
\[
 t-t_0 = \hat{\lambda}^{-1} \int_{y(s,t_0)}^{y(s,t)} \phi(F/n,u) du\,.
\]
The distribution of $t_1-t_0$ is exponential with rate $\lambdap \cdot n$ so
the probability that $t-t_0$ is greater than a quantity $t'> 0$ is given
by $e^{-\lambdap \cdot n \cdot t'}.$ Since $y(s,t)$ is monotonic in $t,$ the
probability that $y(s,t_1)$ is greater than $z$  where $y(s,t_0) \le z <
1$ is given by
\[
\Prob{y(s,t_1)>z} = e^{-\lambdap \cdot n \cdot \hat{\lambda}^{-1} \int_{y(s,t_0)}^{z} \phi(F/n,u) du}
\]
and this gives the distribution along the curve
$y(s,t_1),\hat{T}(s,t_1).$

\begin{comment}
QUANTIZED FORM
In particular if the positions $z$ are quantized with spacing $\delta$
then we have
\[
p(y(s,t_1)\in \delta(k-\frac{1}{2},k+\frac{1}{2})) = e^{-\lambdap n  \hat{\lambda}^{-1}(s,t) \int_{y(s,t_0)}^{(k-\frac{1}{2})\delta} \phi(f,u) du}
-
e^{-\lambdap n  \hat{\lambda}^{-1}(s,t) \int_{y(s,t_0)}^{(k+\frac{1}{2})\delta} \phi(f,u) du}
\]
If, furthermore we have $y(s,t_0) = j \delta$ then we  may approximate
\(
\int_{j\delta}^{(k-\frac{1}{2})\delta} \phi(f,u) du \simeq\frac{1}{2}  \delta \phi(f,\delta i) + \sum_{i=j+1}^{k-1} \delta \phi(f,\delta i)
\)
and
\(
\int_{j\delta}^{(k+\frac{1}{2})\delta} \phi(f,u) du \simeq \frac{1}{2}  \delta \phi(f,\delta i) +\sum_{i=j+1}^k \delta \phi(f,\delta i)\,.
\)
Using this form we can quantize the interval
$\delta(k-\frac{1}{2},k+\frac{1}{2}))$ to $\delta k.$ Then, to complete
the calculation of the joint distribution we represent the distribution
of $\hat{\lambda}$ conditioned on final queue position $k\delta.$ The
range of $t_1$ associated to this position is $[\hat{T}\sum_{i=j}^{k-1}
\delta \phi(f,\delta i),\hat{T}\sum_{i=j}^{k} \delta \phi(f,\delta i)]$
from which the distribution of the updated $\hat{T}$ can be easily
computed and appropriately quantized and represented.  Note that
according to our assumption on the dependence of $\alpha,$ it's value is
independent of the value of $t_1-t_0.$
\end{comment}
%to the extent that $\alpha$ is determined.  
%In general we may allow $\alpha$ to depend on the number of
%erased fragments and/or on the total number of node failures since the object entered the queue.

Let us introduce the notation $Q_{F^+}(F,y,\hat{\lambda})$ to denote the
conditional density of the extended definitions of $F,y,\hat{\lambda}$
conditioned on $F^+.$ Using the above described calculations we can
update this distribution to the distribution at the next node failure
point.  We will use the notation $\tilde{Q}_{F^+}(y,\hat{\lambda})$ to
denote this updated distribution.  Assuming $y<1,$ the next node failure
is incrementing with probability $(1-F/n)$ thus for $y<1$ we have
\[
Q_{F^+}(F,y,\hat{\lambda}) = (F/n) \tilde{Q}_{F^+-1}(F,y,\hat{\lambda}) + (1-(F-1)/n) \tilde{Q}_{F^+-1}(F-1,y,\hat{\lambda})\,.
\]
and for $y = 1$ we have
\[
Q_{F^+}(F,y,\hat{\lambda}) =  \tilde{Q}_{F^+-1}(F,y,\hat{\lambda}) \,.
\]
With suitable quantization of $y$ and $\hat{\lambda}$ the above
procedure can be used to compute the greedy distribution.

To obtain the greedy bound for an arbitrary initial time we initialize
the  distrbution $Q_{0}$ so that the marginal distribution of $y$ is
$\indicator_y$ (and, necessarily, $F=0$ with probability $1.$) Given a
Poisson node failure model and a particular estimator the distribution
$\hat{\lambda}(s,s)$ may be determined.  In general the distribution of
$\hat{\lambda}$ can be set according to the desired assumptions of the
computation.  The limit as $F^+$ tends to $+\infty$ of
$Q_{F^+}(F,1,\hat{\lambda})$ is the greedy approximation to the
distribution upon repair of $F$ and $\hat{\lambda}.$ The probability of
large $F^+$ with $y<1$ decays exponentially in $F^+,$ so this
computation converges quickly in practice.

To obtain a bound on the \mttl{} we need to perform the corresponding
computation for critical objects.  This requires only one additional
initialization step in which the distribution $Q_{1}(F,y,\hat{\lambda})$
is modified to project out $y$ and put all probability mass on $y=0.$
Note that when $F^+ = 1$ we have $F=1$ with probability $1.$ With such a
modification of the initialization process we can obtain a lower bound
on the \mttl{} as before.  We remark, however, that in this case we take
the definition of \mttl{} to include a random initialization of
$\hat{\lambda}$ according to the distribution of $\hat{\lambda}(s,s)$
and the above computation should also use this initial condition.

Fig. \ref{fig:estimateBounds} shows the comulative distribution of the number of missing fragments upon repair $q$ as obtained from simulation
and the computed bound.  The simulation objects used for an estimate of node failure rate the average of the node failure interrarrival times
for the last $155-F$ node failures, where $F$ denotes the number of erased fragments in a object.  Note that the value of
$155$ was chosen slightly larger than $r+1 = 135.$
The computed bound used a slightly different estimator consistent with the above analysis.  The estimator used is a first order filter
but, to best model the corresponding simulation, the time constant of the filter was adjusted as a function of $F$ to match $155-F.$
In both cases the peak repair rate was limited to 3 times the nominal rate as determined by the actual node failure rate.
It can be observed from these the plots that even compared with know node failure arrival rate, the estimated case can give larger MTTDL.
This is intuitive since data loss is generally associated to increases in node failure rate and the increase in the estimated node failure rate case gives rise
to faster repair.  In addition, as shown previously, a regulator with rate estimation can accomodate large swings in node failure rates.

\begin{figure}
\centering
\includegraphics[scale=0.5]{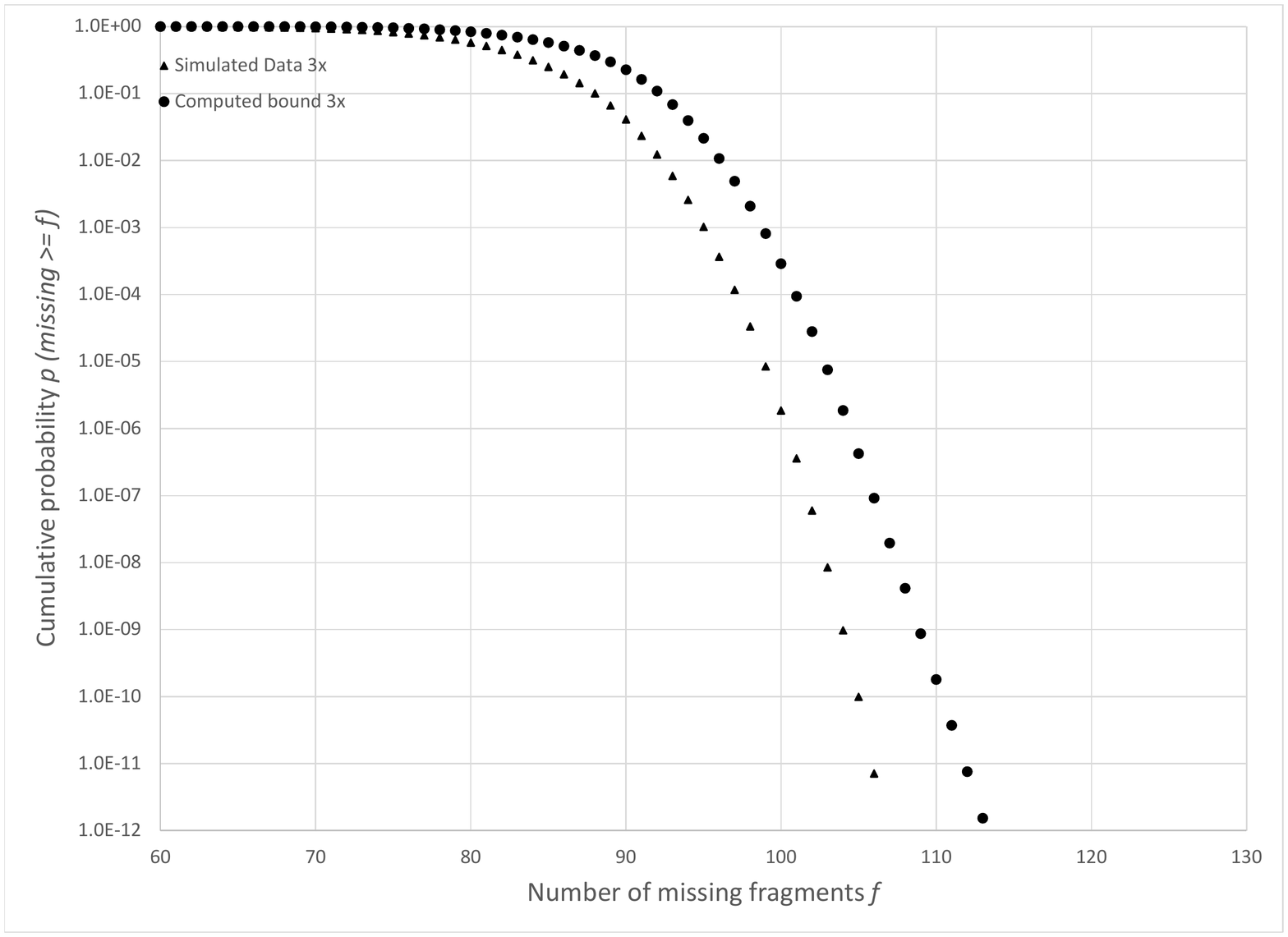}
\caption{Comparison of simulated and computed bound for repair regulation with failure rate estimation and a peak repair rate limit of 3 times the nomimal rate. }
\label{fig:estimateBounds}
\end{figure}

\end{document}